\documentclass[11pt]{imsart}

\usepackage[margin=1in]{geometry}

\usepackage[toc]{appendix}
\usepackage{statmath}
\usepackage{bbold}
\RequirePackage[OT1]{fontenc}
\RequirePackage[authoryear]{natbib}
\usepackage{adjustbox}
\usepackage{amsmath,amssymb,amsthm,bm,latexsym}
\usepackage{graphicx,psfrag,epsf}
\usepackage{epigraph,xcolor}
\usepackage{enumerate}
\usepackage{url}
\usepackage{booktabs}
\usepackage{multirow}
\usepackage{hyperref}
\hypersetup{colorlinks=true,citecolor=blue,
	pdfpagemode=FullScreen,linktocpage=true}
\usepackage{cleveref}
\usepackage{amsfonts}
\usepackage{amsmath}
\usepackage{orcidlink}
\usepackage{amsthm}
\usepackage[T1]{fontenc}
\usepackage{soul}
\usepackage{caption}
\usepackage{subcaption}

\startlocaldefs
\def\E{\mathbb{E}}
\def\I{\mathbf{I}}
\def\P{\mathbb{P}}

\def\L{\bm{L}}

\def\Z{\mathbb{Z}}

\def\var{{\textrm{var}}\,}
\def\Cov{{\textrm{Cov}}\,}

\def\N{\mathbb{N}}
\def\Z{\mathbb{Z}}

\def\R{\mathbb{R}}

\newcommand{\convD}{\xrightarrow[]{\mathcal{D}}}
\newcommand{\convP}{\xrightarrow[]{\mathcal{P}}}

\newcommand{\id}{\mathbf{I}}

\renewcommand{\P}{\mathbb{P}}
\renewcommand{\L}{\mathcal{L}}

\newcommand{\RNum}[1]{\uppercase\expandafter{\romannumeral #1\relax}}

\newcommand{\norm}[1]{\left\| #1 \right\|}

\newtheoremstyle{general}
{3mm} 
{3mm} 
{\it} 
{} 
{\bfseries} 
{.} 
{.5em} 
{} 

\theoremstyle{general}

\newtheorem{theorem}{Theorem}
\newtheorem{corollary}{Corollary}

\endlocaldefs

\begin{document}

\begin{frontmatter}
\title{A nonparametric approach to understand multivariate quantile dynamics in financial time series}

\runtitle{ Nonparametric methods in multivariate financial time series}

\begin{aug}
\author[IIMB]{\fnms{Kunal} \snm{Rai}\,\ead[label=e1]{kunal.rai21@iimb.ac.in}},
\author[IISERP]{\fnms{Archi} \snm{Roy}\,\ead[label=e2]{roy.archi@students.iiserpune.ac.in}},
\author[ISR]{\fnms{Itai} \snm{Dattner}\,\ead[label=e3]{idattner@stat.haifa.ac.il}},
\author[IIMB]{\fnms{Soudeep} \snm{Deb}\,\orcidlink{0000-0003-0567-7339}\ead[label=e5]{soudeep@iimb.ac.in}}

\address[IIMB]{Indian Institute of Management Bangalore, Bannerghatta Main Rd, Bangalore, KA 560076, India.}
\address[IISERP]{Indian Institute of Science Education and Research, Dr Homi Bhabha Rd, Pune, MH 411008, India.}
\address[ISR]{Department of Statistics, University of Haifa, Rabin Building, Mount Carmel, Haifa 3498838, Israel.}


\runauthor{Rai et al. }
\end{aug}
\bigskip

\begin{abstract}
 Over the last decade, nonparametric methods have gained increasing attention for modeling complex data structures due to their flexibility and minimal structural assumptions. In this paper, we study a general multivariate nonparametric regression framework that encompasses a broad class of parametric models commonly used in financial econometrics. Both the response and the covariate processes are allowed to be multivariate with fixed finite dimensions, and the framework accommodates temporal dependence, thereby introducing additional modeling and theoretical hurdles. To address these challenges, we adopt a functional  dependence structure which permits flexible dynamic behavior while maintaining tractable asymptotic analysis. Within this setting, we establish strong and weak convergence results for the estimators of the conditional mean and volatility functions. In addition, we investigate conditional geometric quantiles in the multivariate time series context and prove their consistency under mild regularity conditions. The finite sample performance is examined through comprehensive simulation studies, and the methodology is illustrated by modeling the stock returns of Maersk and Lockheed Martin as a nonparametric function of a geopolitical risk index.
\end{abstract}

\begin{keyword}
    \kwd{geometric quantiles}
    \kwd{Nadaraya-Watson estimator}
    \kwd{stochastic regression}
    \kwd{time series}
\end{keyword}

\end{frontmatter}

\section{Introduction} 
\label{sec:introduction}

Despite their widespread use in modeling time series data, parametric frameworks have been criticized for their limited flexibility, particularly in some types of financial datasets where they fail to capture the rich and evolving dynamics appropriately \citep[see, for example,][]{bingham2002semi}. This has spurred interest in nonparametric methods, which offer a flexible and scalable alternative. In this paper, our focus in this paper is on a novel extension of existing nonparametric theory which can be useful in modeling and analyzing the dynamics of multivariate financial data. Specifically, we develop a unified nonparametric framework, along with suitable asymptotic theory, to explore the drift, volatility and quantiles in a multivariate time series.

It is of essence here to recall that one of the earliest approaches in modeling financial returns nonparametrically was through continuous time diffusion models \citep[see][among others]{jiang1997nonparametric, fan2005selective}: $dY_t=\mu(s_t)dt+\sigma(s_t)dW_t$, where $\{Y_t\}$ is the series to be modeled as a function of a covariate $\{s_t\}$, $\mu(\cdot)$ and $\sigma(\cdot)$ are respectively the drift and the volatility function and $\{dW_t\}$ is a standard Brownian motion. In a nonparametric approach, no structural assumptions are imposed on the functional forms of $\mu(\cdot)$ and $\sigma(\cdot)$ functions; rather, they are globally estimated using a kernel smoother. Notably, \cite{zhao2008confidence} established the asymptotic properties of these functional estimators in the univariate setting. However, extending this theory to a multivariate framework, especially in the presence of temporal dependence, introduces substantial technical challenges, which we address in this paper. We consider the model
\begin{equation}
    \label{eq:main_model_multivariate}
    \bm{Y_t}=\mu(\bm{X_t})+\Sigma^{1/2}(\bm{X_t})\bm{\bm{e}_t}, \; t=1,\dots,n,
\end{equation}
where $\{\bm{Y_t}\}\in\mathbb{R}^p$ represents a $p$-variate response variable, $\{\bm{X}_t\} \in \mathcal{X} \subseteq \mathbb{R}^k$ is a $k$-dimensional covariate process with a compact $\mathcal{X}$, and $\{ \bm{e}_i \} \in \R^p$ represents a multivariate stochastic independent and identically distributed noise. Our theoretical results are derived under the multi-dimensional case where $p < \infty$ and fixed, and $n\to\infty$. Further, our setup allows for both heavy-tailed and long memory in $\bm{X}_t$, making it especially appropriate for application to financial data.  While we do not impose any parametric form on the functions $\mu:\mathbb{R}^k\rightarrow \mathbb{R}^p$ and $\Sigma:\mathbb{R}^{k}\rightarrow \mathbb{R}^{p\times p}$, we emphasize that \eqref{eq:main_model_multivariate} encompasses a large class of parametric models commonly encountered in time series literature. For example, putting $\mu(\bm{x}) = A \bm{x}$ and $\Sigma(\bm{x}) = \Sigma$ for a constant matrix $A\in \mathbb{R}^{k\times p}$ and a positive definite matrix $\Sigma\in \mathbb{R}^{p\times p}$, with $\bm{X}_t=\bm{Y}_{t-1}$ gives us the VAR(1) model. If $\bm X_t$ further includes exogenous variables, then it is a VARX model. Similarly, putting $\mu(\bm{x}) = A_1 \bm{x} \bm{1}_{\{\bm{x}\in \mathcal{R}_1\}} + A_2 \bm{x} \bm{1}_{\{\bm{x}\in \mathcal{R}_2\}}$ with $\mathcal{R}_1,\mathcal{R}_2$ being a partition of $\mathcal{X}$, $A_1, A_2 \in \mathbb{R}^{p\times k}$ and $\Sigma(\bm{x})=\Sigma$, gives us the multivariate threshold TAR model. A multivariate GARCH(1,1) model is obtained by specifying a constant (or linear) conditional mean and allowing the conditional covariance matrix to evolve dynamically. For example, we obtain a BEKK-GARCH type specification by putting $\Sigma(\cdot)=H_t$ where $H_t=C+A_1 \bm{Y}_{t-1}\bm{Y}_{t-1}^\intercal A_1^\intercal+B_1 H_{t-1} B_1^\intercal$ for matrices $A_1,B_1,C$ of suitable dimensions. Similarly, other models can also be obtained by suitably choosing the specifications.
 
In subsequent sections, we define the multivariate local constant estimates of drift and volatility in the generalized framework \eqref{eq:main_model_multivariate}, establish their consistency and weak convergence. Then, we focus on the estimation of conditional quantiles, motivated by their relevance in financial risk estimation. At this stage, it is natural to address why a multivariate notion of risk is required. In recent years, multivariate risk measures have attracted considerable attention since they explicitly account for cross-sectional dependence across assets. When tail behavior must be assessed jointly for multiple financial returns whose dependence structure evolves with economic conditions, a vector-valued conditional quantile becomes essential. A motivating example for us is connected to modeling the joint lower-tail behavior of defense and transportation stocks as a function of geopolitical risk indices,  where separate univariate quantile estimates fail to capture the possibility of simultaneous extreme losses and the strengthening of dependence during periods of heightened geopolitical tension. A multivariate conditional quantile, returning a vector of the same dimension as the response, provides a coherent object that captures both marginal tail behavior and their joint structure. Empirical evidence, such as \cite{santos2013comparing}, suggests that multivariate approaches outperform univariate methods in portfolio value-at-risk evaluation. In this context some early contributions \citep[e.g.][]{jouini2004vector,ekeland2012comonotonic,cousin2013multivariate} proposed various multivariate risk measures, which were later criticized for lacking clear economic interpretation. Recent contributions include order-statistic and quantile-regression–based methods \citep{chun2012conditional, sun2018integrated}, copula-based constructions \citep{coblenz2018nonparametric}, and optimal-transport–based quantiles \citep{bercu2024monge}. In this study, we adopt geometric quantiles \citep{chaudhuri1996geometric, chowdhury2019nonparametric}, which provide a natural and well-defined multivariate extension of univariate quantiles without requiring an ordering in higher dimensions. Defined through a convex optimization problem, geometric quantiles enjoy uniqueness under mild conditions, are affine equivariant, and retain a direct connection to classical univariate quantile regression \citep{koenker1978regression}. This makes them particularly suitable for conditional multivariate time series modeling.

Overall, the contribution of this paper is three-fold. First, we use Nadaraya-Watson type nonparametric methods for estimating multivariate conditional mean and variance and derive suitable asymptotic theory which has not been formally done before. Second, we propose a nonparametric methodology for estimating conditional geometric quantiles of multivariate time series, where the resulting quantile is vector-valued and matches the dimension of the response. Relevant asymptotic properties of the proposed conditional quantile estimator are derived under mild regularity conditions, allowing for heavy-tailed innovations and general dependence structures in the covariates. Overall, this provides an integrated treatment of location, scale, and tail behavior for multivariate financial data. As the third contribution of this paper, we provide detailed empirical exploration of the developed theory with a real-life study on the behavior of two important geopolitical risk-prone stocks -- Lockheed Martin, which is a defense company in United States of America, and Maersk, which is a shipping company in Denmark -- under the influence of the geopolitical risk. A simulation study that mimics similar time series data is also included in the paper for completeness. 

The rest of the article is organized in the following manner. We review relevant literature in \Cref{sec:literature_review}. The main theoretical results pertaining to the asymptotic theory of the drift, volatility and quantile estimates are established in \Cref{sec:main_results}. The efficacy of the proposed estimators are illustrated with a brief simulation study and a detailed real data application on defense and shipping markets in \Cref{sec:simulation_studies} and \Cref{sec:application}, respectively. Finally, we conclude with some necessary remarks in \Cref{sec:conclusions}. In the interest of space, all proofs are deferred to the Appendix.

\section{Literature Review}
\label{sec:literature_review}

A vast literature has developed on nonparametric regression for estimating the conditional mean function $\mu(\cdot)$. Standard approaches include local constant, local linear, and local polynomial estimators, along with robust variants such as nonparametric M-estimators and smoothing splines. Early methodological contributions include design-adaptive regression \citep{fan1992design}, nearest-neighbor methods \citep{altman1992introduction}, and regression with noisy covariates \citep{mammen2012nonparametric}. Theoretical properties of these estimators are well understood in the independent data setting \citep{takezawa2005introduction}, while extensions to higher-dimensional settings have largely relied on sparsity assumptions or variable selection techniques \citep{lin2006component, bertin2008selection, yang2015minimax}. In contrast, we consider multivariate nonparametric regression in a time series framework with temporal dependence and without imposing sparsity restrictions.

Compared to the literature of conditional mean estimation, there have been less work devoted to the estimation of the conditional volatility. Early contributions in this direction include \cite{muller1987estimation, hall1990asymptotically}, among others. These studies primarily focused on estimating the conditional variance as a single scalar quantity, rather than modeling it as a function of the covariate process. A particularly influential contribution is \cite{wang2008effect}, which established minimax optimal rates for the estimation of conditional heteroskedasticity in nonparametric regression settings. In context of estimation of the conditional variance function, \cite{chib2006analysis} proposed a Bayesian approach based on Markov chain Monte Carlo sampling from standard distributions, drawing inspiration from techniques developed for the estimation of stochastic volatility models, as illustrated in \cite{kim1998stochastic}. An innovative approach was proposed by \cite{cai2008adaptive}, who developed a data-driven estimator of the conditional heteroskedasticity function by applying wavelet thresholding to the squared first-order differences of the observations. Parallel approaches to the estimation of the conditional variance include least squares–based estimators \citep{tong2005estimating}, methods based on weighted empirical processes \citep{chown2018detecting}, techniques employing second-order U-statistics \citep{liu2021adaptive}, and more recent approaches based on variational mode decomposition \citep{palanisamy2017smoothing}. While most of these techniques can be adapted to multivariate data, they are largely reliant on the data generating process being independent. In \cite{kulik2011nonparametric}, the authors considered the heteroskedasticity estimation problem for nonparametric regression with dependent errors, and more recently \cite{hu2013nonparametric} proposed a kernel-based estimator which retained asymptotic consistency under strongly mixing in the data generating process. Complementing the existing literature, we propose estimation of the conditional variance under very general structure of the multi-dimensional covariate process, which covers the case of both short and long memory. We further establish their asymptotic consistency and weak convergence properties under minimal structural assumptions.

Turning attention to the estimation of multivariate quantiles, we note the parametric and semi-parametric approaches of \cite{chakraborty2003multivariate, ma2019quantile, jurevckova2024estimation}, while nonparametric extensions of local linear methods for real-valued responses appear in \cite{lejeune1988quantile, yu1998local, horowitz2005nonparametric}. However, these methods are primarily designed for scalar responses and do not directly address genuinely multivariate quantile objects. More recent work on multivariate quantiles includes directional, kernel-based, copula-based, and conditioning-based constructions; see \cite{chavas2018multivariate}, \cite{huang2018nonparametric}, \cite{coblenz2018nonparametric}, and \cite{galvao2025multivariate}. Our approach instead builds on the geometric quantile framework, which formulates multivariate quantiles through a convex optimization problem and avoids the need for ordering, conditioning sequences, or dimension-reduction schemes. Unlike much of the existing literature, which is developed under independence or structured model assumptions, we establish asymptotic results for conditional geometric quantiles in a multivariate time series setting with general dependence. This places multivariate quantile estimation within a fully nonparametric dynamic framework and complements the existing body of work on mean and variance estimation under dependence.

\section{Methodology}
\label{sec:main_results}

Throughout this article, the symbols $\R,\Z,\L^k$ are used to indicate respectively the set of real numbers, the set of integers, and the set of all random variables with finite moments up to the $k^{th}$ order. Whenever used, $\id$ represents an identity matrix of appropriate order. Otherwise, bold-faced letters are used to denote vectors unless otherwise stated. We shall use $\convD$ (similarly, $\convP$) to indicate convergence in distribution (similarly, convergence in probability). Finally, $\mathcal{N}_p(\bm{\theta},\Sigma)$ refers to a $p$-variate Gaussian distribution with mean $\bm{\theta}$ and covariance matrix $\Sigma$.

\subsection{Mathematical framework}
\label{subsec:Mathematical framework}

Recall our stochastic regression model described in \eqref{eq:main_model_multivariate} that covers a wide variety of commonly encountered multivariate time series models. We first define the problem setting in which our unified approach of multivariate nonparametric estimation of conditional mean, variance, and quantiles is developed. The covariate process $\{\bm{X_t}\}$ is taken to be stationary and independent of the error process $\{\bm e_t\}$. It is defined through a measurable function of an independently and identically distributed (iid) sequence of random variables $\{\eta_t\}_{t\in \Z}$ as
\begin{equation}
    \bm{X}_t = \mathcal{G}(\mathcal{F}_t),
\end{equation}
where $\mathcal{F}_t$ is the $\sigma$-field generated by $(\dots, {\eta}_{t-1}, \eta_t)$, and $\mathcal{G}$ is a measurable function such that $\bm{X}_t$ is well defined. Such an assumption is commonly used in time series literature; see \cite{fan2005selective} for a detailed review. As an example to illustrate this data generation process, let $\bm{X}_i = \bm{c} + \sum_{j=1}^{q} \bm{B}_j \bm{\varepsilon}_{i-j} + \bm{\varepsilon}_i$ where $\bm{\varepsilon}_{i-j}$ is a past white noise vector of index $i-j$, $\bm{c}$ is a vector of constants, $\bm{B}_j$ are matrices of coefficients for each lag $j$ and $\bm{\varepsilon}_i$ is a vector of error terms. To capture the dependence structure in $\{\bm{X_t}\}$ we define a projection operator $\mathcal{P}_k$, which for a random variable $W \in \L^1$ and for $k \in \Z$, is defined as: $\mathcal{P}_kW = \E(W \mid \mathcal{F}_k) - \E(W \mid \mathcal{F}_{k-1})$. Denote $F_{\bm{X}}$ and $F_{\bm{e}}$ as the multivariate distribution functions of $\bm{X}_0$ and $\bm{e}_0$ respectively; with $f_{\bm{X}} = F'_{\bm{X}}$ and $f_{\bm{e}} = F'_{\bm{e}}$ being the corresponding density functions. Further, for $i \in \Z$, let $F_{\bm{X}}(\bm{x} \mid \mathcal{F}_i) = \P(\bm{X}_{i+1} \leqslant \bm{x} \mid \mathcal{F}_i)$ denote the multivariate conditional distribution function of $\bm{X}_{i+1}$ given $\mathcal{F}_i$ and $f_{\bm{X}}(\bm{x} \mid \mathcal{F}_i) = \partial F_{\bm{X}}(\bm{x} \mid \mathcal{F}_i) / \partial x $ the corresponding conditional density. Then, following \cite{wu2005nonlinear}, define the functional dependence measure
\begin{equation}
\label{eq:functional_dependence_measure}
    \theta_i = \sup_{\bm{x} \in \R^k} [\norm{\mathcal{P}_0 f_{\bm{X}}(\bm{x} \mid \mathcal{F}_i)} + \norm{\mathcal{P}_0 f'_{\bm{X}}(\bm{x} \mid \mathcal{F}_i)}],
\end{equation}
where $f'_{\bm{X}}(\bm{x} \mid \mathcal{F}_i) = \partial f_{\bm{X}}(\bm{x} \mid \mathcal{F}_i) / \partial x$, if it exists. Roughly, the term $\theta_i$ measures the contribution of $\bm{e}_0$ in predicting $\bm{X}_{i+1}$. Using this, for $n \in \N$, we further define 
\begin{equation}
\label{eq:srd_lrd}        
    \Theta_n = \sum_{i=1}^n \theta_i, \;
    \quad
    \Lambda_n = n\Theta_{2n}^2 + \sum_{k=n}^{\infty} \left( \Theta_{n+k} - \Theta_k \right)^2.
\end{equation}
We consider $\Theta_{\infty} < \infty$, which implies that the cumulative contribution of $\eta_0$ in predicting future values is finite, thus pointing to the short-range dependence (SRD) setting. In this case, $\Lambda_n = \mathcal{O}(n)$. When $\Theta_{\infty} \to \infty$, we have the long range order.

\subsection{Estimation procedure and asymptotic theory for conditional mean and volatility}
\label{subsec:Estimation procedure and asymptotic theory for conditional mean and volatility}

To establish the theoretical properties of the estimates in \eqref{eq:multivariate_mu_expression} and \eqref{eq:multivariate_variance_function}, we begin by introducing some notation of future interest. Let $\mathcal{C}^0(\bm{x})$ denote a set of continuous functions in the $\epsilon$-neighborhood of $\bm{x}$, defined as $\bm{x}^{\epsilon}  = \bigcup_{\bm{y} \in \R^k} \{ \bm{y}: \norm{\bm{x}-\bm{y}} \leqslant \epsilon \} \in \R^k$ for some $\epsilon > 0$. Next, define $\mathcal{C}^N(\bm{x}^{\epsilon}):=\{ g(\cdot): \sup_{\bm{x} \in \bm{x}^{\epsilon}} |g^{(l)}(\bm{x})|<\infty, \; \text{for } l=0,\dots,N \in \Z\}$ as the set of all the functions that have bounded derivatives on $\bm{x}^{\epsilon}$ up to $N$. The density function $f_{\bm{X}}$ is assumed to satisfy the conditions $f_{\bm{X}} > 0$ and $f_{\bm{X}} \in \mathcal{C}^2(\{\bm{x}^{\bm{\epsilon}}\})$ for all $\bm\epsilon$ satisfying $\|\bm{\epsilon}\| > 0$. Next, we assume $\bm{e}_t \in \mathcal{L}^2$, where a random variable $W$ is considered to be in $\mathcal{L}^p$ for $p>0$, if $\|W\|_p := [\E(|W|^p)]^{1/p} < \infty$. In the same context, denote $\mathbb{V}_e = \var[\bm{e}_t]$. Finally, let $\mathcal{K}$ be the set of kernels which are symmetric, whose values are bounded, and have bounded derivative and bounded support. The kernel functions used in our methods are considered to be taken from this set.  Using $\langle \cdot,\cdot\rangle$ to indicate inner products, two terms of interest, $\psi_K$ and $\phi_K$, are defined as $\psi_K = \frac{1}{2} \int_{\R^k} \langle \bm{t}, \bm{t} \rangle K(\bm{t})d\bm{t}$ and $\phi_K = \int_{\R^k}K^2(\bm{t})d\bm{t}$.

The model \eqref{eq:main_model_multivariate} explores relationships among various covariates and the multi-dimensional dependent variable through the mean function $\mu(\cdot)$, and the conditional variance function $\Sigma(\cdot)$. We let $\mu, \Sigma \in \mathcal{C}^2(\{\bm{x}\}^{\bm{\epsilon}})$ for all $\bm\epsilon$ satisfying $\|\bm{\epsilon}\| > 0$, along with $\Sigma(\bm x)$ being a positive definite matrix for all $\bm x$. We also assume that our time series variables are stationary and $\E[\|\bm{Y}_0\|^4] < \infty$. We use nonparametric Nadaraya-Watson type estimators to estimate these functions. It is imperative to point out that a local constant estimator instead of local linear estimator is attractive due to mathematical properties, however the asymptotic theory follows in similar lines with slight modifications. The estimator for the multivariate mean and variance function are respectively given by
\begin{equation}
    \label{eq:multivariate_mu_expression}
            \widehat{\mu}(\bm{x}) = \sum_{t=1}^{n}{\bm{Y}_t\nu_t(\bm{x})},
\end{equation}
and
\begin{equation}
    \label{eq:multivariate_variance_function}
    \widehat\Sigma(\bm{x})=\sum_{t=1}^{n}{\left(\bm{Y}_t-\widehat\mu(\bm{X}_t)\right) \, \nu_t \, (\bm{x})\left(\bm{Y}_t-\widehat\mu(\bm{X}_t)\right)^\intercal},
\end{equation}
where
\begin{equation}
\label{eq:weight_definition}
    \nu_t(\bm{x})=\frac{K_{b_n}\left(\bm{x}-\bm{X}_t\right)}{\sum_{t=1}^{n}{K_{b_n}\left(\bm{x}-\bm{X}_t\right)}}, K_{b_n}(\bm{v})=\frac{1}{b_n^k}K\left(\frac{\bm{v}}{b_n}\right).
\end{equation}

In the above, $K:\R^k\rightarrow\R$ is an appropriately chosen kernel function and  $b_n$ is the bandwidth satisfying  $b_n \to 0$ and $nb_n \to \infty$ as $n \to \infty$. For notational ease, we shall use the same bandwidth $b_n$ throughout this paper, although each estimator can also be computed using different bandwidth choices, with minimal modifications to the theory.

Now, the following theorem establishes the consistency and the asymptotic Gaussian behavior of the pointwise estimate  of conditional mean under certain bandwidth conditions.

\begin{theorem}
    \label{thm:theorem1}
    \begin{enumerate}
        \item[] 
        \item[i)] For any fixed $\bm{x} \in \R^k$, under the aforementioned settings, the proposed estimate of the mean function converges to the population equivalent in probability, i.e., $\widehat\mu(\bm{x}) \convP \mu(\bm{x})$.
        \item[ii)] Further assume the bandwidth satisfies
\begin{equation}
    b_n + nb_n^{k+4} + \frac{1}{nb_n^k} + \Lambda_n \left( \frac{b_n^3}{n} + \frac{1}{n^2} \right) \rightarrow 0,
\end{equation}
and let $\rho_{\mu}(\bm{x})=\nabla^2 \mu(\bm{x})+2\nabla \mu(\bm{x})\frac{\nabla f_{\bm{X}}(\bm{x})}{f_{\bm{X}}(\bm{x})}$. Then as $n \rightarrow \infty$,
\begin{equation}
\label{eq:thm1}
    \frac{\sqrt{nb_n^k \hat{f}_{\bm{X}}(\bm{x})}}{\sqrt{\phi_K}} (\Sigma^{1/2}(\bm{x})\mathbb{V}_e\Sigma^{1/2}(\bm{x})^{\top})^{-1/2}(\bm{x}) \left[ \widehat{\mu}_{b_n}(\bm{x})-\mu(\bm{x})-b_n^{2}\psi_K\rho_{\mu}(\bm{x}) \right] \convD N_p(\bm{0}, \bm{I}).
\end{equation}
    \end{enumerate} 
\end{theorem}

The first part of the proof of \Cref{thm:theorem1} follows from the Taylor expansion of the nonparametric estimate. The asymptotic distribution is derived by expressing $\left(\widehat\mu(\bm{x})-\mu(\bm{x})\right)$ as a martingale difference sequence; followed by an application of the martingale central limit theorem. The details, for brevity, are added in the Appendix. 
Note from model \eqref{eq:main_model_multivariate} that the conditional covariance matrix of $\bm{Y}_t$ given $\bm{X}_t=\bm{x}$ is $\Upsilon(\bm{x}) := \text{Var}(\bm{Y}_t \mid \bm{X}_t=\bm{x}) = (\Sigma^{1/2}(\bm{x}) \mathbb{V}_e \Sigma^{1/2}(\bm{x})^\top)$, which appears as the scaling factor in the above asymptotic distribution. We hereafter suppose, $\mathbb{V}_e = \I$ for the purpose of brevity and convenience, similar results can be established with a general, $\mathbb{V}_e$, as well. Note that if $\mathbb{V}_e = \I$ then $\Upsilon(\bm{x}) = \Sigma(\bm{x})$. The next result establishes the consistency of the variance function. It allows us to substitute the population conditional variance function by a sample estimate, as discussed below. 

\begin{theorem}
    \label{thm:theorem2} For any fixed $\bm{x} \in \R^k$, under the aforementioned settings, the proposed estimate of the covariance function converges to the population equivalent in probability, i.e., $\widehat{\Sigma}(\bm{x}) \inprob \Sigma (\bm{x})$. 
\end{theorem}

The proof of \Cref{thm:theorem2} uses the oracle estimator $\widetilde{\Sigma}$ and compares it to the Nadaraya Watson based estimator $\widehat{\Sigma}$, subsequently connecting it to the true covariance function, $\Sigma$. The details are deferred to the Appendix. It further helps us in obtaining the following result.

\begin{corollary}
\label{cor:corollary1}
    For significance level $\alpha$, a $100(1-\alpha)\%$ confidence interval for linear functional of the form $\bm a^{\intercal}\mu, \, \bm a \in \R^p$, is obtained as
    \begin{equation}
    \label{eq:confidence_bands}
    \bm a^{\intercal}\mu(\bm{x}) \in \bm a^{\intercal}\widehat{\mu}^{*}_{b_n}(\bm{x}) \pm \sqrt{\frac{\phi_K}{nb_n^k\hat{f}_{\bm{X}}(\bm{x})}\chi^2_{p;1-\alpha}{\bm a^{\intercal}\widehat{\Sigma}(\bm{x}) \bm a}}.
    \end{equation}
\end{corollary}

The above gives us a band that covers all possible linear contrasts. Here, following standard literature, $\widehat{\mu}^{*}_{b_n}(\bm{x})$ is a jackknife-type estimator to provide bias-corrected multivariate mean estimate. Note that setting $\bm a$ equal to empirical basis vectors would yield marginal intervals for each component $\mu_j(\bm x)$ with family-wise coverage $1-\alpha$. \Cref{cor:corollary1} is straightforward from the results of \Cref{thm:theorem1}, part $(ii)$, \Cref{thm:theorem2}, and Slutsky's theorem.

Estimation of the conditional mean function in \eqref{eq:main_model_multivariate} characterizes the expected or average behavior of the response variable given the covariates. For a sample of size $n$ and covariate dimension $p$, the conditional mean of the dependent variable is expressed in \eqref{eq:multivariate_mu_expression}, where the associated weights are defined in \eqref{eq:weight_definition}. Beyond modeling the average behavior, the framework also allows estimation of the conditional variance function which captures the time-varying dispersion of the response variable. Together, these two components describe both the expected return and the associated uncertainty conditional on available information. In financial applications, this distinction is particularly important. The conditional mean provides insight into expected returns or price movements, while the conditional variance reflects risk or volatility. Accurate estimation of both functions enables improved forecasting, portfolio allocation, and risk management, as investors require not only predictions of average performance but also reliable measures of the variability and potential downside risk surrounding those predictions.

We propose to use multivariate Epanechnikov or parabolic kernel function. It is bounded on its support, the $p$-dimensional unit ball, and is given by:
\begin{equation}
\label{eq:epanechnikov}
K(\bm{u}) = \frac{k+2}{2 c_k} \left(1 - \bm{u}^\intercal \bm{u}\right) \mathbb{I} \left(\norm{\bm{u}} \leqslant 1\right),
\end{equation}
where $c_k = \pi^{k/2}/\Gamma(\frac{k}{2} + 1)$ is defined through the gamma function, and $\norm{\bm{u}} = \sqrt{\sum_{i=1}^{k} u_i^2}$. The purpose of using a kernel function with a bounded support in this case is the computational ease that it offers. We may also use kernels with no bounded support, such as the multivariate gaussian kernel function, and the performance will remain similar.

\subsection{Estimation procedure and asymptotic theory for conditional quantiles}
\label{subsec:Estimation procedure and asymptotic theory for conditional quantiles}

We adopt the geometric quantile definition to estimate the multivariate quantiles following the formulation proposed by \cite{chowdhury2019nonparametric}. The use of geometric quantiles stems from its theoretical properties, including existence, uniqueness, and strict convexity of the objective function as well as empirical usefulness in capturing directional features of multivariate distributions. Geometric quantiles provide a coherent and non crossing characterizations indexed by direction vectors. We first define population version of the conditional geometric quantile of $\bm{Y}_t$ given the covariate $\bm{X}_t = \bm{x}$ is defined as
\begin{equation}
\label{eq:population_quantile}
    Q(\bm{u},\bm{x}) := \argmin_{q \in \R^p} M_{\bm{u}}^{(p)}(\bm{q}); \,\, \text{where,} \,\, M_{\bm{u}}^{(p)}(\bm{q}) := \E [\mid\mid \bm{Y}_t-\bm{q}\mid\mid + \langle \bm{u}\mid_{\mathcal{S}_p} \bm{Y}_t-\bm{q}\rangle \mid \bm{X}_t=\bm{x}],
\end{equation}

Further, we define the conditional quantiles for sample based on the definition for population,
\begin{equation}
\label{eq:estimated_quantile}
    \hat{\bm{q}}_n(\bm{u},\bm{x}) := \arg\min_{\bm{q}\in\R^p} M_{\bm{u},n}^{(p)}(\bm{q}); \,\, \text{where,} \,\,     M_{\bm{u},n}^{(p)}(\bm{q}) := \frac{1}{n} \sum_{t=1}^n \{ \| \bm{Y}_t-\bm{q}\|_{K_{b_n}(\cdot)} + \langle \bm{u}| \bm{Y}_t-\bm{q} \rangle_{K_{b_n}(\cdot)}\}.
\end{equation}

The minimization problem involved in \eqref{eq:estimated_quantile} is non-trivial. We propose using an iteratively re-weighted least squares (IRLS) estimation technique, which offers an efficient method for the estimation of the required quantiles. The key updating step of the minimization algorithm is given as
\begin{equation}
    \label{eq:IRLS_equation}
    \hat{\bm{q}}_{n}^{(k+1)}(\bm{u},\bm{x}) = \frac{\frac{1}{2}\sum_{t=1}^n K_{b_n}(\bm{X}_t-\bm{x})\bm{u} + \sum_{t=1}^n w_t(\hat{\bm{q}}^{(k)})K_{b_n}(\bm{X}_t-\bm{x})^2 \bm{Y}_t}{\sum_{t=1}^n w_t(\hat{\bm{q}}^{(k)})K_{b_n}(\bm{X}_t-\bm{x})^2},
\end{equation}
where $w_t(\bm{q}^{(k)}) = \|\bm{Y}_t - \hat{\bm{q}}_{n}^{(k)}(\bm{u} \mid \bm{x})\|^{-1}_{K_{b_n}(\bm{x}-\bm{X}_t)}$ is a sequence of weights. 

The following theorem establishes both the algorithmic stability and the structural characterization of the geometric quantile estimator. It first demonstrates the numerical convergence of the IRLS scheme by utilizing its descent property to show that the sequence of iterates converges to the unique global minimizer of the objective function. This guarantees the stability of computation and the validity of the algorithmic implementation. In addition, the theorem provides the first order optimality condition satisfied by the sample minimizer and its corresponding population counterpart. In this way, it links the computational procedure to the underlying theoretical structure, establishes the existence and uniqueness of the minimizer, and connects the sample and population geometric quantiles through their respective estimating equations.

\begin{theorem}
\label{thm:theorem3}
Assume that for each accumulation point $\bar{\bm{q}}$ of the IRLS iterates 
$\{\bm{q}^{(k)}\}_{k\geqslant 0}$, there exists a constant $\delta>0$ (independent of $n$ and $p$) such that $\min_{1\leqslant t\leqslant n}\ \|\bm{Y}_t-\bar{\bm{q}}\|_{K(\cdot)} \geqslant \delta$. Under the previously stated regularity assumptions and the bandwidth condition, the IRLS scheme defined through \eqref{eq:IRLS_equation} produces a sequence of iterates whose objective values 
$\{M^{(p)}_{\bm{u},n}(\bm{q}^{(k)})\}_{k\geqslant 0}$ are non increasing, and the iterates 
$\bm{q}^{(k)}$ converge to the unique global minimizer 
$\hat{\bm{q}}_n(\bm{u},\bm{x})$ of the objective function $M^{(p)}_{\bm{u},n}$. Moreover, the sample minimizer $\hat{\bm{q}}_n(\bm{u},\bm{x})$ satisfies the first order optimality condition
\begin{equation}
\label{eq:objective_function_incarnation_sample}
\sum_{t=1}^n K_{b_n}(\bm{X}_t-\bm{x})
\left[
\frac{\bm{Y}_t-\hat{\bm{q}}_n(\bm{u},\bm{x})}
{\|(\bm{Y}_t-\hat{\bm{q}}_n(\bm{u},\bm{x}))\|}
+ \bm{u}
\right]
=0.
\end{equation}

The population quantile $Q(\bm{u},\bm{x})$, defined as the minimizer of the population objective $M^{(p)}_{\bm{u}}$, satisfies the corresponding population version
\begin{equation}
\E_{\,\bm{Y}_i\,|\,\bm{X}_i=\bm{x}}
\left[
\frac{\bm{Y}_i-Q(\bm{u},\bm{x})}
{\|\bm{Y}_i-Q(\bm{u},\bm{x})\|}
+ \bm{u}
\right]
=0.
\end{equation}
\end{theorem}

In implementations one may, instead of $w_t(\bm{q}^{(k)})$, use stabilized weights $w_{t,\vartheta}(\bm{q}) = 1/(\|\bm{Y}_t-\bm{q}\|_{K(\cdot)}+\vartheta)$ with a small fixed $\vartheta>0$ to avoid numerical instability when an iterate approaches a data point. The convergence argument above is stated for $\vartheta=0$ under the separation assumption, while the stabilized version corresponds to minimizing a smooth perturbed objective and enjoys the same monotone descent property.

\begin{corollary}
\label{cor:corollary2}
Assume that for each fixed $\bm{x}\in\mathcal{X}$ the population objective $M^{(p)}_{\bm{u}}(\bm{q})$ is strictly convex in $\bm{q}$ and admits a unique minimizer $Q(\bm{u},\bm{x})$. Then, for each fixed $\bm{x}$, the mapping $\bm{u} \mapsto Q(\bm{u},\bm{x})$ where $\bm{u} \in B^{p-1}$, is injective. In particular, if $\bm{u}_1 \neq \bm{u}_2$, then $Q(\bm{u}_1,\bm{x}) \neq Q(\bm{u}_2,\bm{x})$, and hence the geometric quantiles possess an intrinsic non crossing structure across directions.
\end{corollary}

In the geometric quantile framework, quantiles are indexed by directions $\bm{u} \in B^{p-1}$ rather than by a scalar level. For each fixed $\bm{x}$, the geometric quantile $Q(\bm{u}, \bm{x})$ is defined as the unique minimizer of a strictly convex objective. Strict convexity ensures uniqueness for each fixed $\bm{u}$.

To establish the non crossing property across directions, 
consider the population first order condition satisfied by the geometric quantile:
\[
\E\!\left[
\frac{\bm{Y} - Q(\bm{u},\bm{x})}
     {\|\bm{Y} - Q(\bm{u},\bm{x})\|}
\,\Big|\, \bm{X}=\bm{x}
\right]
= -\bm{u}.
\]
Suppose, for contradiction, that two distinct directions 
$\bm{u}_1 \neq \bm{u}_2$ produce the same quantile 
$q = Q(\bm{u}_1,\bm{x}) = Q(\bm{u}_2,\bm{x})$. 
Then the first order condition would imply
\[
\E\!\left[
\frac{\bm{Y} - q}
     {\|\bm{Y} - q\|}
\,\Big|\, \bm{X}=\bm{x}
\right]
= -\bm{u}_1
\quad \text{and} \quad
\E\!\left[
\frac{\bm{Y} - q}
     {\|\bm{Y} - q\|}
\,\Big|\, \bm{X}=\bm{x}
\right]
= -\bm{u}_2,
\]
which is impossible since $\bm{u}_1 \neq \bm{u}_2$. Therefore, distinct directions must produce distinct quantiles. This establishes mapping is injective, and hence the geometric quantiles possesses an intrinsic non crossing structure.

The following result establishes consistency of our multivariate quantile estimator by showcasing the convergence of the sample estimate, $\hat{q}_n(\bm{u},\bm{x})$, to the population estimate, $Q(\bm{u}, \bm{x})$ in probability hence reinforcing the believe in the accuracy of our estimator as the number of observations increases.

\begin{theorem}
\label{thm:theorem5}
    Under the previous assumptions on quantile estimates and bandwidth, the estimator $\hat{\bm{q}}_n(\bm{u},\bm{x})$ of the population quantile $Q(\bm{u},\bm{x})$, satisfies
    \begin{equation}
        \| \hat{q}_n(\bm{u},\bm{x}) - Q(\bm{u}, \bm{x}) \| = \mathcal{O}_{\P}(\sqrt{p/nb_n^k}),
    \end{equation}
    for all $\bm{x} \in \mathcal{X}$ where $\mathcal{X} \subset \R^k$ is compact and $\int_{\mathcal{X}}\|\bm{u}(s)\|^2ds \leqslant 1$.
\end{theorem}

In establishing the above result we utilize the analytical stability result from \Cref{thm:theorem3}, the strict convexity of the objective function, and rates of convergence for the relevant terms in the Taylor expansion of the score function under the small range dependence in the physical dependence setting $\Lambda_n = \mathcal{O}(n)$.  

Estimation of the conditional quantile function, from \eqref{eq:main_model_multivariate}, gives us the average percentile behavior of the response variable. For $n$ number of observations and $p$ number of dimensions, the conditional quantile of the dependent variable, conditioned on the covariates is given by \Cref{eq:IRLS_equation} where a modified pinball-type loss function for the multivariate case is used. The expression, $\|\bm{Y}_t-\bm{q}\|$ is a norm, and the expression, $\langle \bm{u}, \bm{Y}_t-\bm{q} \rangle_{K(\cdot)}$, describes an euclidean inner product between $\bm{u}$ and $\bm{Y}_t-\bm{q}$ under kernel function, $K(\cdot)$. The expression for $Q(\bm{u})$ is a generalization of the objective function defined by \cite{koenker1978regression}.  

Many of the risk measures are derivative of the quantiles, example, Value-at-Risk and Average Value-at-Risk. We emphasize on the Value-at-Risk measure which has been established as an important financial risk measure post the 2008 financial crisis, (see \cite{ruiz2023direct}. Value-at-Risk (VaR) at confidence level $\alpha \in (0,1)$ is the $\alpha$-quantile of the loss distribution. Suppose $L$ denotes loss, then, 
\begin{equation}
\label{eq:var}
    \text{VaR}_{\alpha} = \inf\{l \in \R: \P(L \leqslant l) \geqslant \alpha\}.
\end{equation}

At level $\alpha$, VaR is the smallest number $l$ such that the probability of the loss exceeding $l$ is at most $1-\alpha$.

\section{A short simulation study}
\label{sec:simulation_studies}

We conduct a brief simulation study that replicates real-world phenomena, particularly a multivariate time series data with temporal dependence among response variables, likely to appear in the world of finance. This study is primarily designed to establish finite sample behavior and empirical consistency. This exercise enables us to examine the generalized stochastic regression framework that we are working with, as well as the accuracy of the estimation procedure. The simulation setup is established by defining the data-generating processes and specifying the number of observations. We consider three covariates ($k = 3$), denoted as $\bm{X}_t = \{ X_{1,t}, X_{2,t}, X_{3,t} \}$, and explore different types of data generation. In this simulation study, we evaluate the estimation accuracy of the multivariate conditional means and multivariate conditional quantiles, which are the main focus of the current study. The estimation of multivariate conditional variance is primarily explored as an auxiliary component required for inference, including asymptotic variance characterization and construction of confidence regions of the mean. The covariates in the simulation study are generated under autoregressive processes of order 1 with common innovations which induce serial correlation and contemporaneous dependence among the variables. The response variables $\bm{Y} = \{Y_1, Y_2\}$ (i.e., $p=2$) are generated following the specifications
\begin{equation*}
	Y_{1t} =  \frac{1}{3} (X_{1t} + X_{2t} + X_{3t})  + e_{1t}, \quad Y_{2t} = b_1X_{1t} + b_2 X_{2t} + b_3 X_{3t} + e_{2t}; \, \text{with } \sum_{i=1}^3 b_i = 1.
\end{equation*}

These two dependent response variables with distinct dependence structure on the covariates helps in exhibiting the strength of our methodology. The error components $e_{1t}, e_{2t}$ for all timepoints are simulated independently using three different population distributions: standard normal, student's $t$ with 3 degrees of freedom, and shifted exponential. These choices are made to ensure that they cover light tail, heavy tail, and skewed error designs, respectively. The sample sizes are taken as $n \in \{100, 500, 1000\}$, representing small, moderate and large sample regimes. 

Following the methodology described in \Cref{sec:main_results}, we apply the estimation techniques for different conditional functions, and compare our estimates with the true values that can be obtained by Monte Carlo sampling from the conditional distributions in our data generating processes. For the quantiles, we consider five different levels: 0.05, 0.1, 0.5, 0.9, and 0.95. At each quantile level, we use IRLS scheme to obtain the geometric multivariate quantiles as explained in \Cref{subsec:Estimation procedure and asymptotic theory for conditional quantiles}. The directional vector $\bm{u} \in B^{p-1}$ is chosen appropriately in this regard. The bandwidth selection for the conditional mean estimation is carried out using a blocked cross-validation procedure for time series data, while for the conditional quantile estimation a rate-consistent bandwidth choice is used. We replicate this exercise for 50 Monte Carlo iterations and compute the root mean square error (RMSE) to assess the accuracy of our methodology. While it is our primary accuracy metric, we also report mean absolute percentage error (MAPE) in a relative scale, which provides us a scale-free measure to identify how the accuracy is improving based on the sample size. In this regard, the baseline is kept at $n=100$ for each case.

The simulation results are tabulated below. In \Cref{tab:simstdy_mcm}, we see that for each error distribution considered, the two metrics display a visibly evident decreasing pattern with increasing sample sizes. This suggests an empirical consistency for the multivariate conditional mean estimates. \Cref{tab:simstdy_mcq}, on the other hand, illustrates the performance of the conditional quantile estimates for different sample sizes, quantile levels and error distributions. Here also, the results suggest empirical consistency for the multivariate conditional quantile estimates for each quantile level. 

\begin{table}[h]
    \centering
    \caption{Simulation results for multivariate conditional mean estimation.}
    \label{tab:simstdy_mcm}
    \begin{tabular}{c|c|c|c}
    \toprule
    Sample size ($n$) & Error distribution & RMSE & Relative MAPE \\
    \midrule
    100 & Normal & 0.259 & 1.000 \\
    500 & Normal & 0.141 & 0.529 \\
    1000& Normal & 0.111 & 0.292 \\
    \midrule
    100 & Exponential & 0.250 & 1.000 \\
    500 & Exponential & 0.135 & 0.212 \\
    1000& Exponential & 0.105 & 0.157 \\
    \midrule
    100 & $t_3$ & 0.362 & 1.000 \\
    500 & $t_3$ & 0.203 & 0.669 \\
    1000& $t_3$ & 0.158 & 0.400 \\
    \bottomrule
    \end{tabular}
\end{table}

\begin{table}[!ht]
\centering
\caption{Simulation results for multivariate conditional quantiles estimation.}
\label{tab:simstdy_mcq}
\begin{tabular}{c c|cc|cc|cc}
\toprule
Sample & Quantile & \multicolumn{2}{c|}{{Normal error}} 
   & \multicolumn{2}{c|}{{Exponential error}} 
   & \multicolumn{2}{c}{{$t_3$ error}} \\
size ($n$) & ($\tau$) & RMSE & Relative MAPE & RMSE & Relative MAPE & RMSE & Relative MAPE \\
\midrule

100  & 0.05 & 0.304 & 1.000 & 0.222 & 1.000 & 0.432 & 1.000 \\
500  & 0.05 & 0.173 & 0.834 & 0.119 & 0.624 & 0.227 & 0.097 \\
1000 & 0.05 & 0.139 & 0.491 & 0.083 & 0.328 & 0.173 & 0.046 \\
\midrule

100  & 0.10 & 0.300 & 1.000 & 0.227 & 1.000 & 0.421 & 1.000 \\
500  & 0.10 & 0.172 & 0.485 & 0.121 & 0.464 & 0.224 & 0.378 \\
1000 & 0.10 & 0.138 & 0.389 & 0.085 & 0.340 & 0.170 & 0.321 \\
\midrule

100  & 0.50 & 0.290 & 1.000 & 0.270 & 1.000 & 0.384 & 1.000 \\
500  & 0.50 & 0.166 & 0.710 & 0.148 & 0.568 & 0.212 & 0.358 \\
1000 & 0.50 & 0.135 & 0.540 & 0.108 & 0.513 & 0.159 & 0.212 \\
\midrule

100  & 0.90 & 0.299 & 1.000 & 0.335 & 1.000 & 0.421 & 1.000 \\
500  & 0.90 & 0.169 & 0.418 & 0.190 & 0.328 & 0.219 & 0.575 \\
1000 & 0.90 & 0.138 & 0.401 & 0.141 & 0.271 & 0.174 & 0.414 \\
\midrule

100  & 0.95 & 0.303 & 1.000 & 0.345 & 1.000 & 0.429 & 1.000 \\
500  & 0.95 & 0.169 & 0.873 & 0.197 & 0.186 & 0.223 & 0.767 \\
1000 & 0.95 & 0.139 & 0.437 & 0.147 & 0.153 & 0.180 & 0.585 \\

\bottomrule
\end{tabular}
\end{table}

\section{Application}
\label{sec:application}

As a practical illustration, we study the intricate relationship between the log-return series of the stock market data of Maersk and Lockheed Martin with geopolitical risk. Using the proposed unified estimation technique, we evaluate the conditional properties in a multivariate sense, specifically various magnitudes of risk through quantiles. It is of the essence here to recall few studies that focused on defense and shipping stock markets. \cite{shackman2021maritime} examined how the global maritime stock prices affect the stock prices of large transportation companies. \cite{bhattacharjee2022analysis} analyzed the volatility of the defense stocks, while the work of \cite{stanivuk2023general} assesed the impact of war in Ukraine on the shipping industry using parametric techniques. Most recently, \cite{sim2024resilient} studied the relationship of global supply chain pressure index and stock prices of global logistic companies.

In our illustration, as covariates we use three types of geopolitical risk indices introduced by \cite{caldara2022measuring}. The authors derive these indices using newspaper coverage of events like war, terror attacks, and other events of global relevance such as election results of major countries. The geopolitical risk index not only captures the actual events of disruptions but also any speculations. More than 25 million news articles from major English newspapers are used in this regard. Roughly speaking, a ratio of the number of articles with specific words which correlates heavily with geopolitical risks to the total number of articles is used to compute the risk index. There are three indices proposed: the first one captures the overall geopolitical risk, the second one is an act based geopolitical risk index that realizes disruptions in the geopolitical relationships, and the third one is a threat based geopolitical risk index that anticipates disruptions in the geopolitical relationships. These three indices are hereafter denoted as GPRD, GPRD-A, and GPRD-T, respectively. As our response variables, we use log-returns of the stock prices of Lockheed Martin, a global defense company ($Y_{\mathrm{LHM}}$), and Maersk, a global shipping company ($Y_{\mathrm{MMA}}$). 

The timeline of the data is the period from 1st April, 2021 to 31st March, 2025. This timeline which is marked as a recovery period from multiple COVID-19 waves, is also known for other globally consequential events like the Russia-Ukraine conflict, Turkey-Syria earthquake, Nagorno-Karabakh conflict escalation, Civil war in Sudan and the Israel-Palestine conflict, just to name a few. A few exploratory plots are presented in \Cref{sec:Plots} in the interest of space. First, \Cref{fig:GPRD_GPRDA_GPRDT_vs_Times} shows the geopolitical risk indices against time: the events of Russia-Ukraine conflict (RUC) in February 2022 and Israel-Palestine conflict (IPC) in October 2023 are particularly visible via peaks in all three time series plots, and they will be of particular interest to us. On the other hand, \Cref{fig:sp_lr_vol_LHMnMMA} shows the stock prices, log-returns, and volatility of the two stocks. The top two panels show a consistent increasing trend in LHM whereas MMA observes an increasing trend from April 2021 to January 2022, but a decreasing trend afterward. The average log-returns for both stocks generally stay close to zero, with occasional fluctuations, as displayed in the middle two panels. Finally, the volatility series displayed in the bottom two panels are based on 5-day rolling standard deviation and are representatives of the risk associated with these stocks. The LHM stock shows two periods of high volatility (2022-23 and 2024-25) and two periods of low volatility (2021-22 and 2023-24). The MMA stock, comparatively speaking, observes high volatility throughout the period of study with two exceptional peaks during the RUC and in the early months of 2024, possibly caused by the IPC.

For the main analysis, we apply our proposed techniques of analyzing multivariate conditional mean, variance, and quantiles. We discuss these results in the same order with other analyses wherever suited for clarity and coherence. For implementation, we use cross-validation techniques to obtain optimal bandwidth values in each case. The estimated conditional mean along with the confidence band, quantiles at $\tau=0.05, 0.50, 0.95$, generalized variance i.e. the determinant of the covariance matrix, and the value-at-risk (VaR) at 95\% level are plotted for the entire timeline in \Cref{fig: lr+CB_gv_quntl_vs_time}. The left-side panels in the top and middle correspond to LHM while the right-side panels show the results for MMA. In both cases, the estimated mean returns remain near zero which is a usual behavior, but minor disturbances can be seen during the RUC and IPC in both the plots. This disturbance around IPC is more loud in the case of MMA, possibly due to maritime security concerns and threats of blockade of important sea routes. These fluctuations are more prominent in the estimated quantiles, presented in \Cref{fig: lr+CB_gv_quntl_vs_time_d} and \Cref{fig: lr+CB_gv_quntl_vs_time_e}. Specifically, we notice that the 5th quantile suffers more disturbances post the IPC in both the plots. It is also imperative to point out that our approach ensures avoidance of quantile crossing. Next, turn attention to the bottom two plots which can be used to talk about risks associated with the two stocks. The generalized variance reflects the combined uncertainty in the two assets, and we observe huge spikes around the RUC. The VaR at 95\% level is another risk measure which implies that with 95\% confidence, loss should not exceed the estimated value. For ease of understanding, we plot it on weekly basis. It helps us infer that the periods of pre-RUC and post-IPC show great deviations. We can see that the risk persists during these periods and help us to detect the crisis. We can also see that the RUC impacted the major defense and shipping sectors more dominantly than any other geopolitical event during the period of study.

\begin{figure}[!ht]
    \centering
    \begin{subfigure}{0.43\textwidth}
        \centering
        \includegraphics[width=\linewidth]{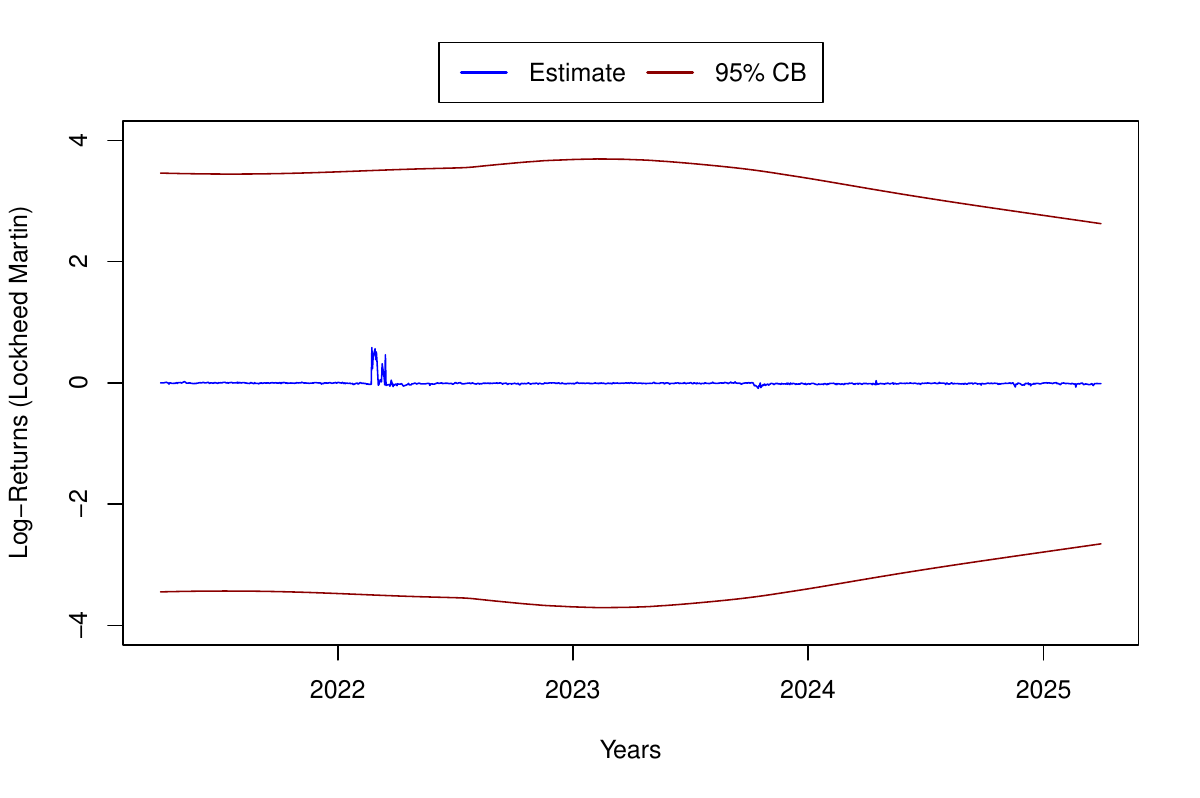}
        \caption{Estimated conditional mean with the corresponding confidence band}
        \label{fig: lr+CB_gv_quntl_vs_time_a}
    \end{subfigure}
    \hfill
    \begin{subfigure}{0.43\textwidth}
        \centering
        \includegraphics[width=\linewidth]{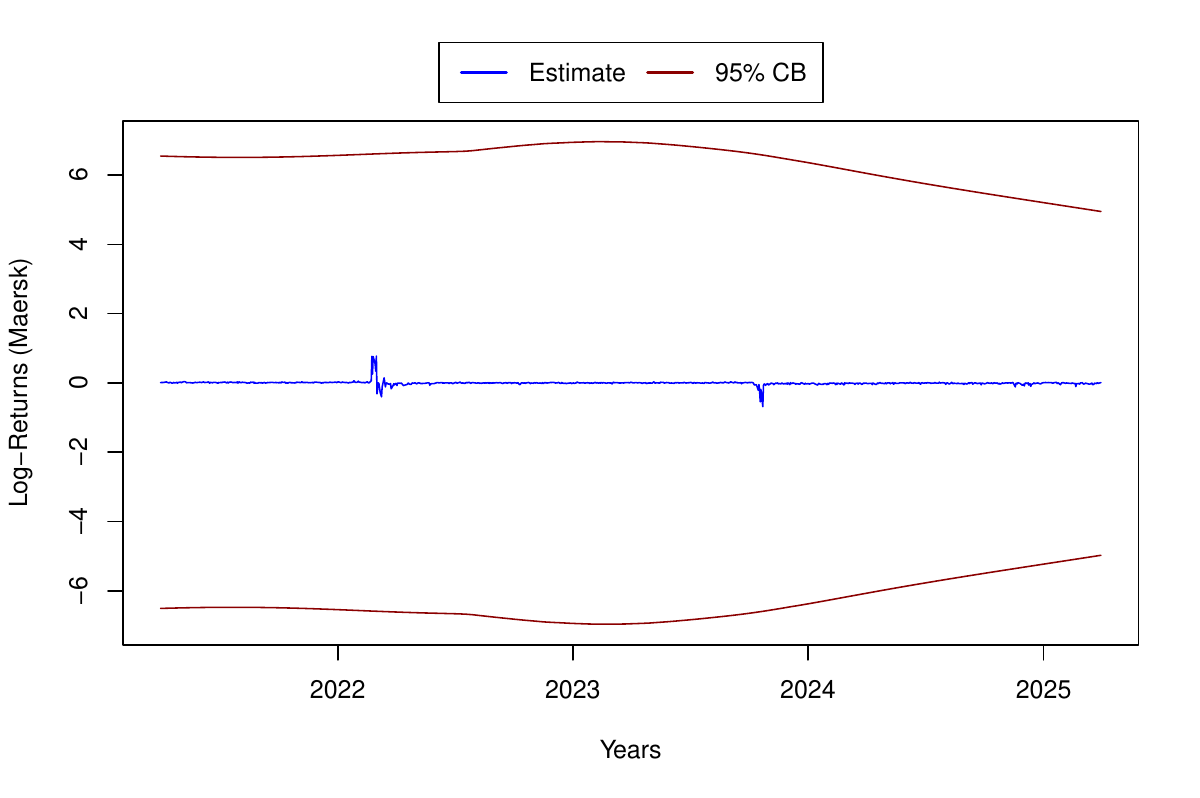}
        \caption{Estimated conditional mean with the corresponding confidence band}
        \label{fig: lr+CB_gv_quntl_vs_time_b}
    \end{subfigure}

    \vspace{0.6cm}

    \begin{subfigure}{0.43\textwidth}
        \centering
        \includegraphics[width=\linewidth]{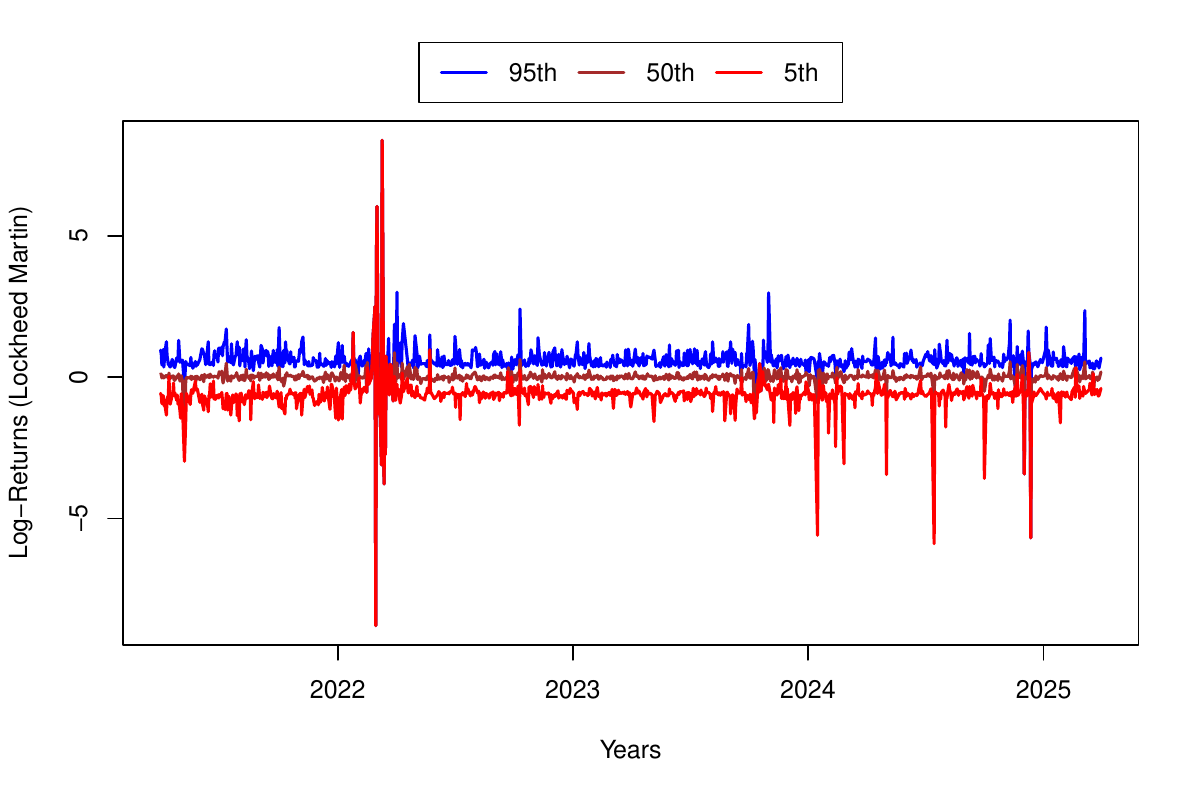}
        \caption{Estimated quantiles at three different levels}
        \label{fig: lr+CB_gv_quntl_vs_time_d}
    \end{subfigure}
    \hfill
    \begin{subfigure}{0.43\textwidth}
        \centering
        \includegraphics[width=\linewidth]{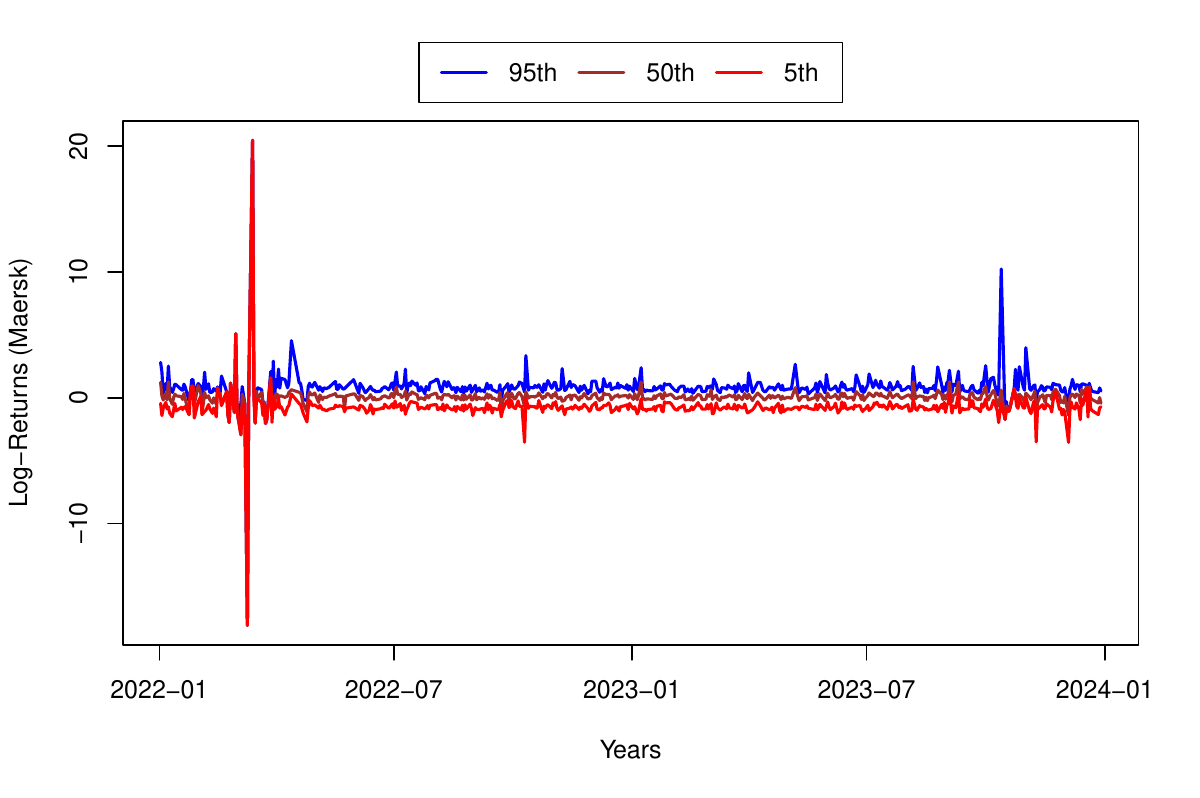}
        \caption{Estimated quantiles at three different levels}
        \label{fig: lr+CB_gv_quntl_vs_time_e}
    \end{subfigure}

    \vspace{0.6cm}

    \begin{subfigure}{0.43\textwidth}
        \centering
        \includegraphics[width=\linewidth]{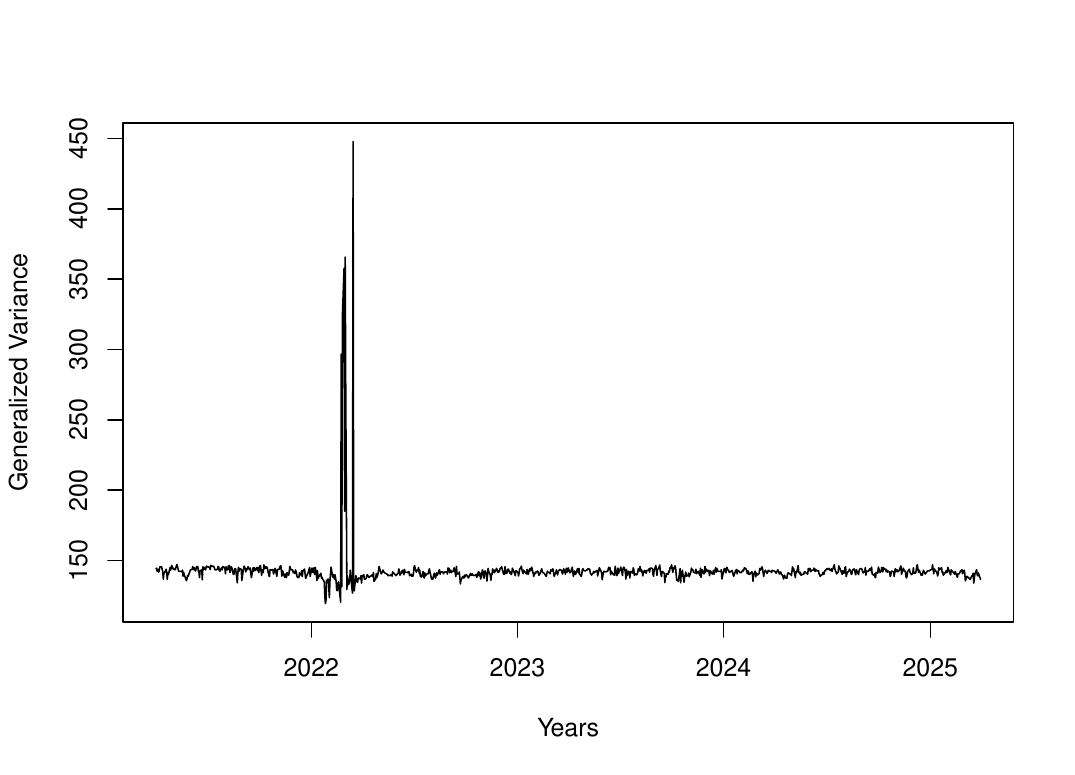}
        \caption{Estimated generalized variance}
        \label{fig: lr+CB_gv_quntl_vs_time_c}
    \end{subfigure}
    \hfill
    \begin{subfigure}{0.43\textwidth}
        \centering
        \includegraphics[width=\linewidth]{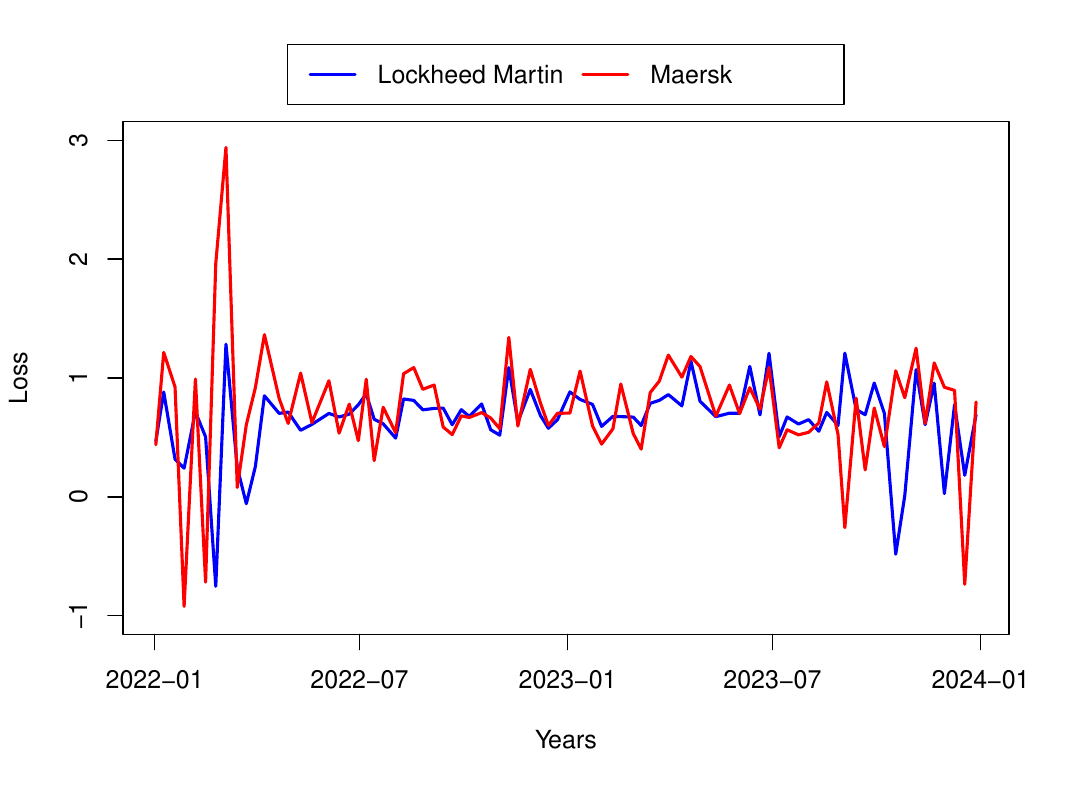}
        \caption{Estimated value-at-risk (VaR) at 95\% level}
        \label{fig:var_vs_time}
    \end{subfigure}

    \caption{Overview of estimated conditional mean, quantiles, volatility and risk, plotted against times.}
    \label{fig: lr+CB_gv_quntl_vs_time}
\end{figure}

Next, in figures \ref{fig:Est_mean_vs_gpr}, \ref{fig:Generalized Variance vs GPRDs}, and \ref{fig:quntvsgprd_a_t_lhm_mma}, we illustrate the conditional mean, volatility, and quantiles as functions of the three different geopolitical risk indices. These visualizations help us assess the effect of the three covariates on different properties of the two assets, while keeping the other covariates constant at the mean levels. From the first figure, one can observe that for each of the risk index, as the value of it increases, the average log-returns decreases for both the stocks. This is an expected behavior. Interestingly, we observe varying effects in the volatility patterns corresponding to the three indices, as presented in the three panels of \Cref{fig:Generalized Variance vs GPRDs}. For GPRD and GPRD-T, the generalized variance shows decreasing trend as the respective risk indices increase whereas for GPRD-A, the estimate stays almost indifferent. This behavior could be attributed to the fact that GPRD-A describes kinetic geopolitical events, events whose consequences are easy to realize. In such a case, the uncertainty as a function of adversity does not change much whereas this is not the case with the other two indices. Increased value of those two risk indices only ascertain the adversity, hence we see a downward slope for the volatility.

\begin{figure}[!ht]
    \centering
    \begin{subfigure}{0.43\textwidth}
        \centering
        \includegraphics[width=\linewidth]{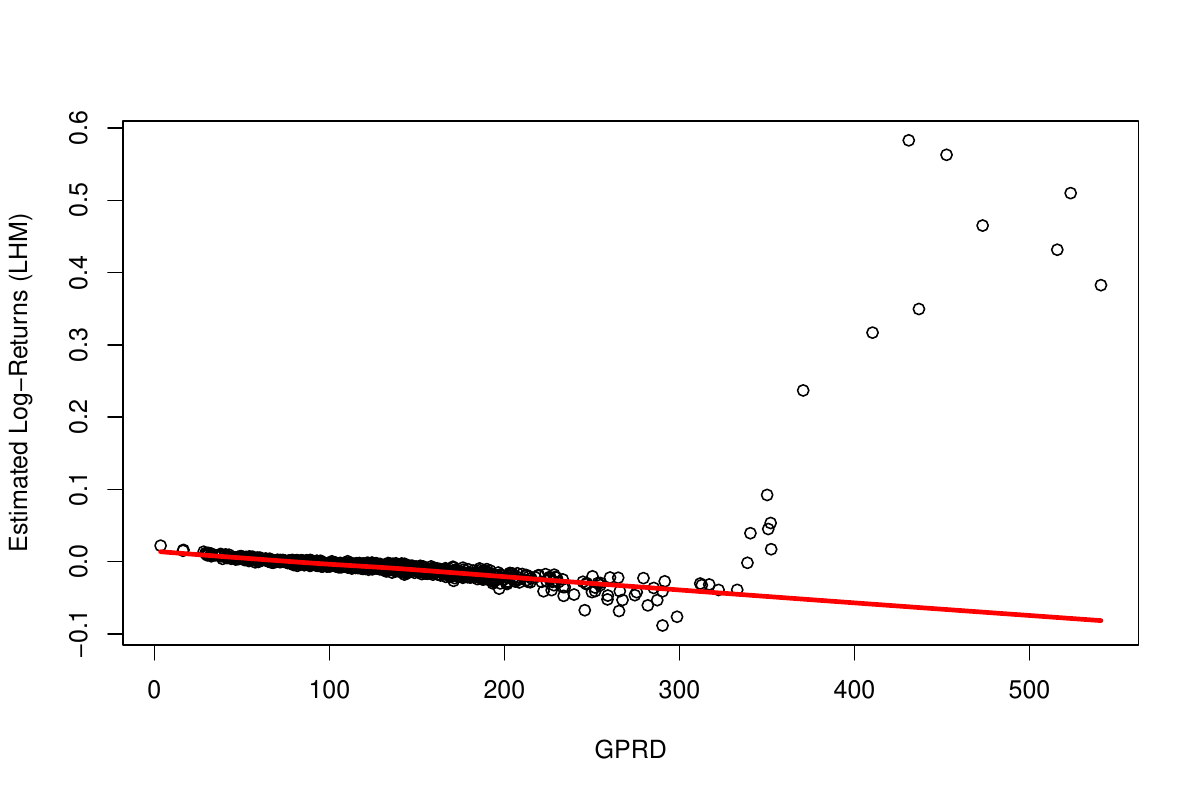}
    \end{subfigure}
    \hfill
    \begin{subfigure}{0.43\textwidth}
        \centering
        \includegraphics[width=\linewidth]{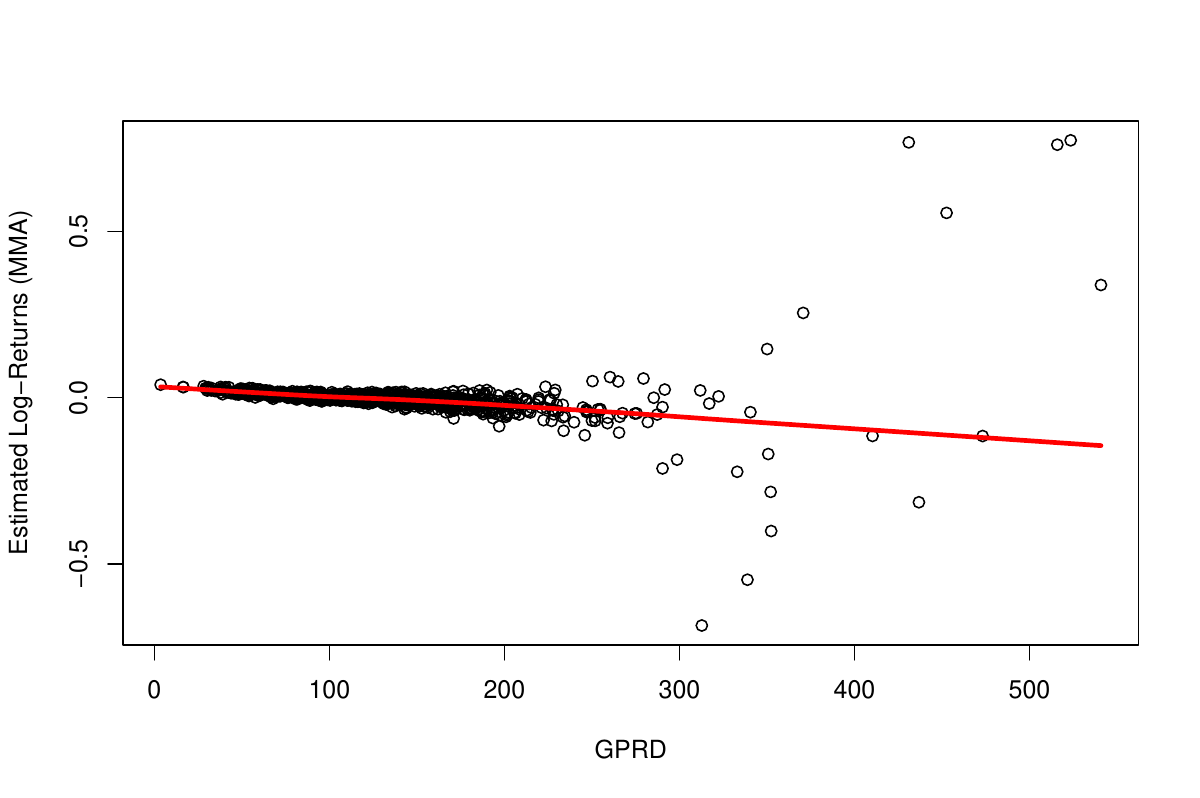}
    \end{subfigure}

    \begin{subfigure}{0.43\textwidth}
        \centering
        \includegraphics[width=\linewidth]{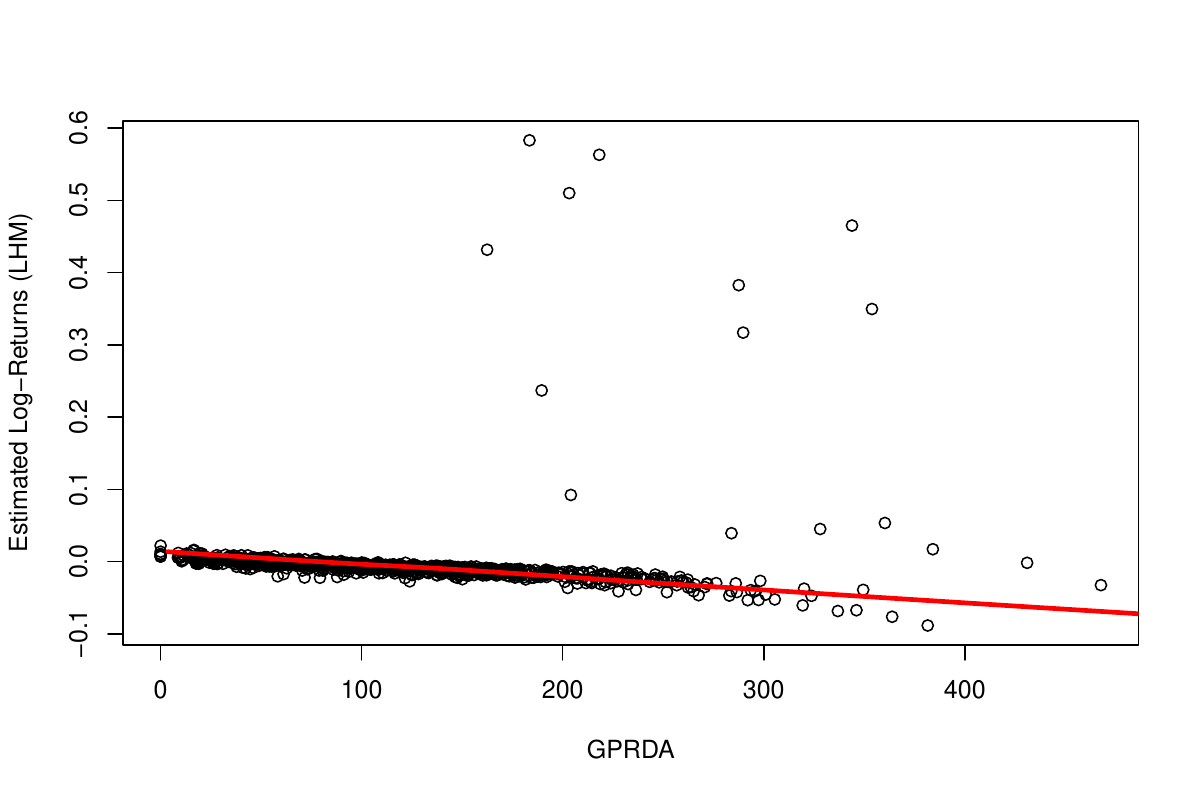}
    \end{subfigure}
    \hfill
    \begin{subfigure}{0.43\textwidth}
        \centering
        \includegraphics[width=\linewidth]{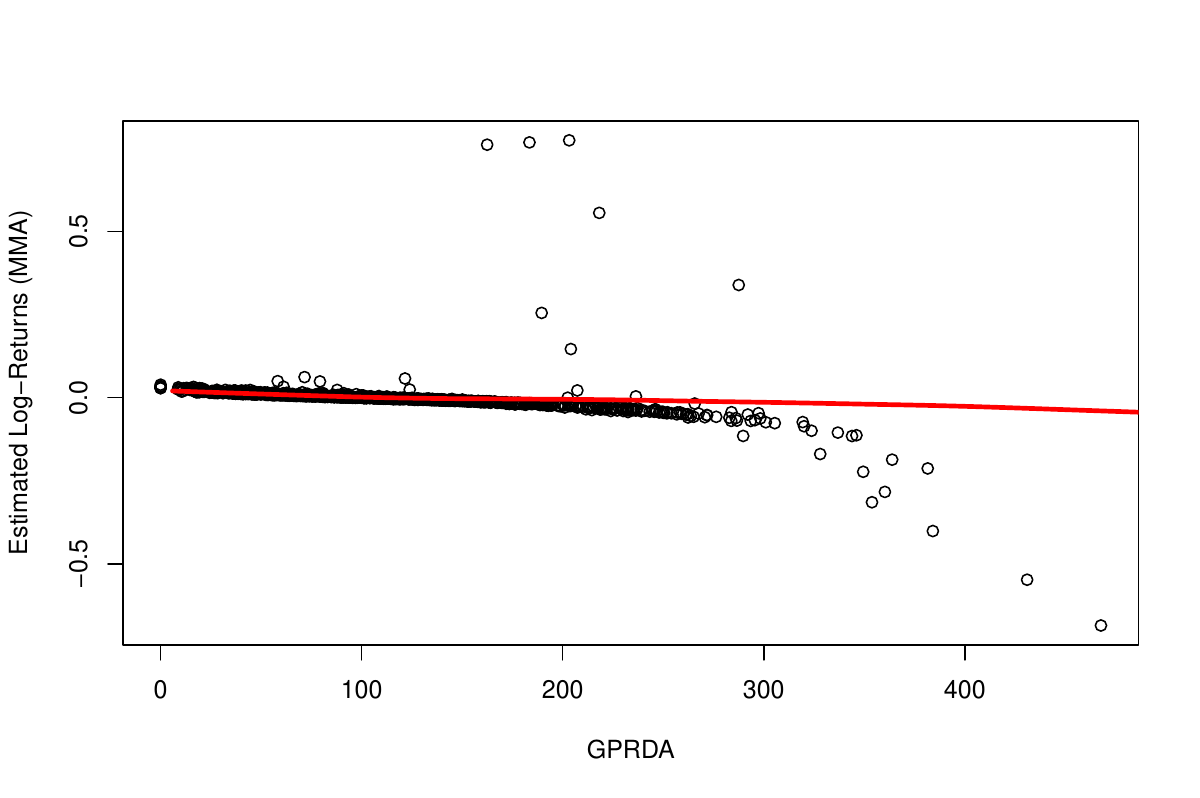}
    \end{subfigure}

    \begin{subfigure}{0.43\textwidth}
        \centering
        \includegraphics[width=\linewidth]{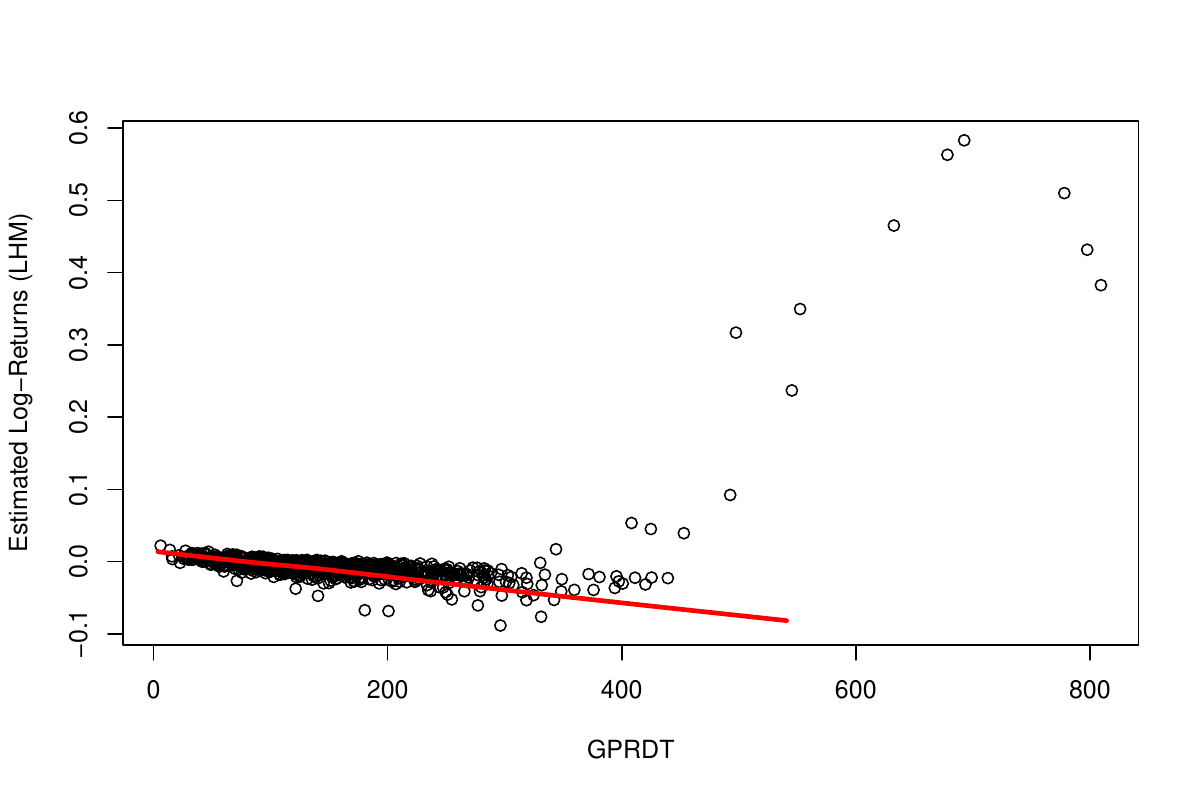}
    \end{subfigure}
    \hfill
    \begin{subfigure}{0.43\textwidth}
        \centering
        \includegraphics[width=\linewidth]{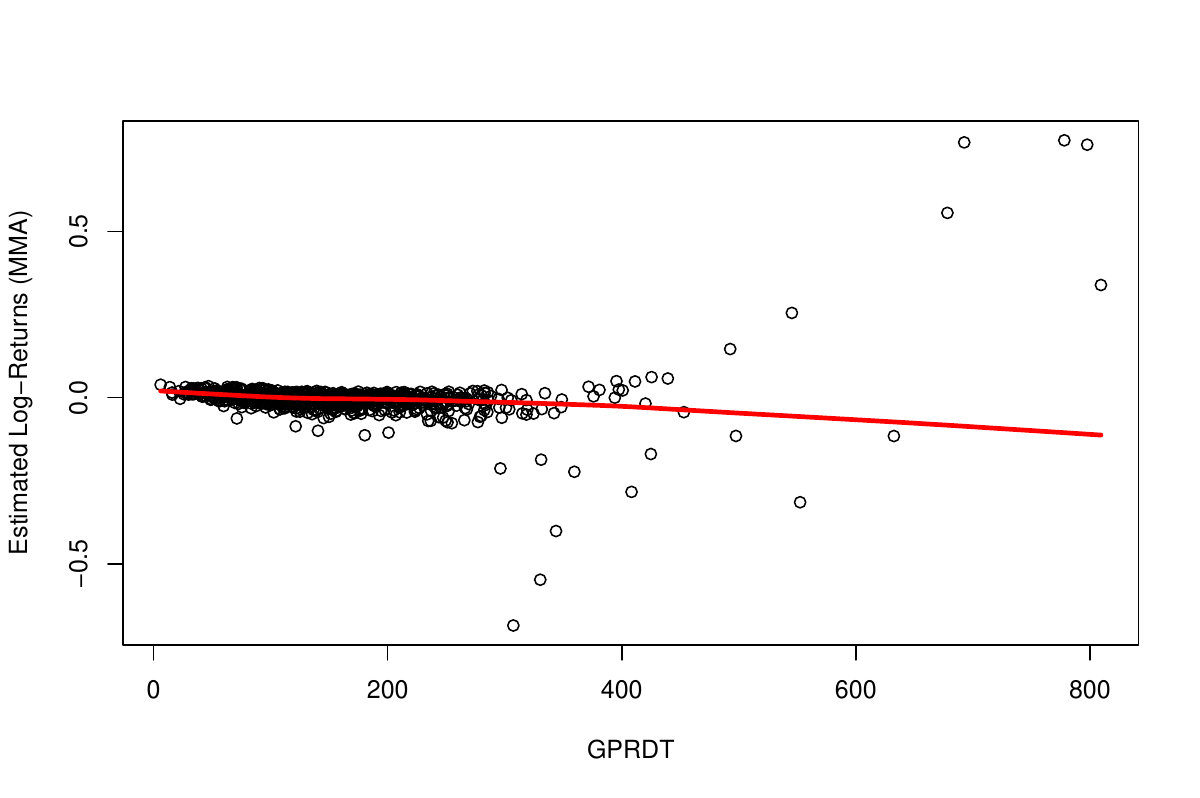}
    \end{subfigure}
    
    \caption{Estimated conditional mean of log-return for LHM (left panels) and MMA (right panels) against the three types of geopolitical risk indices. The red lines are used for smoothing purposes.}
    \label{fig:Est_mean_vs_gpr}
\end{figure}

\begin{figure}[!ht]
\centering
  \begin{subfigure}{.3\textwidth}
    \includegraphics[width=1\linewidth]{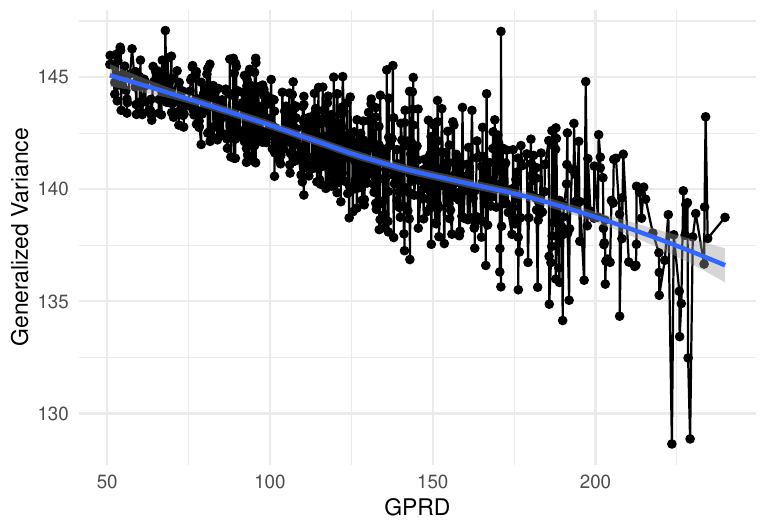}
  \end{subfigure}%
  \begin{subfigure}{.3\textwidth}
    \includegraphics[width=1\linewidth]{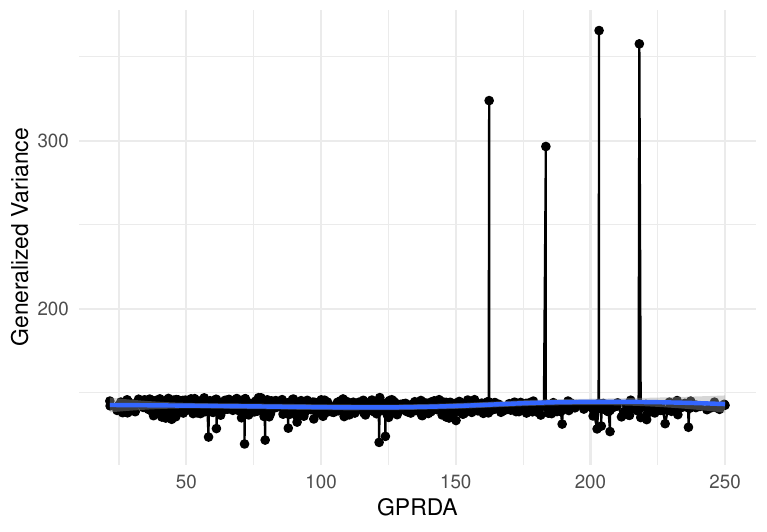}
  \end{subfigure}%
   \begin{subfigure}{.3\textwidth}
    \includegraphics[width=1\linewidth]{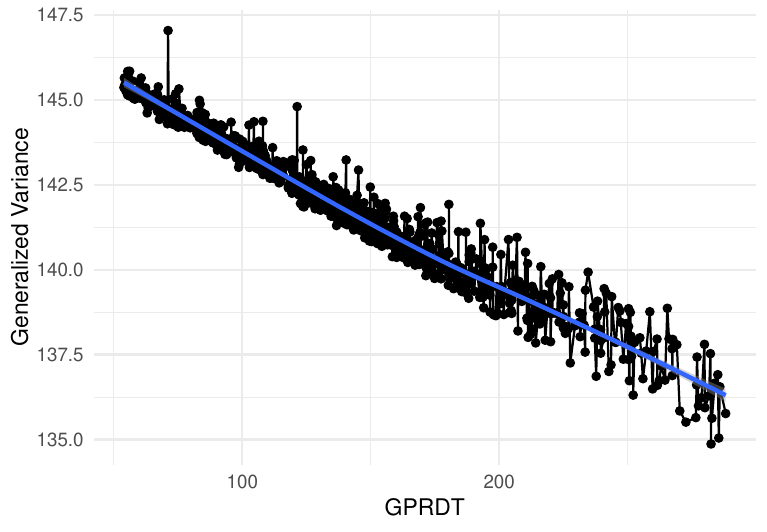}
  \end{subfigure}
    \caption{Estimated generalized variance of log-return for the two assets (LHM and MMA) against the three types of geopolitical risk indices. The blue lines are used for smoothing purposes.}
\label{fig:Generalized Variance vs GPRDs}
\end{figure}

We turn attention to \Cref{fig:quntvsgprd_a_t_lhm_mma} which presents the estimated quantiles for the two assets against the three geopolitical risk indices. In case of LHM, we see that the 5th quantile increases as the GPRD value increases, but an opposite behavior is seen when a higher quantile, that is, the 95th quantile is analyzed. The median behaves more or less similar at all levels of the GPRD values. A different pattern is visible for MMA. There, albeit the 95th quantile behaves similarly, the 5th quantile displays a cubic polynomial behavior where a local maximum is reached around 130 GPRD but a local minima is attained before it. The median behaves similar to the 95th quantile, the values go down but the change is small as the GPRD increases. In quantiles against GPRD-A plots, we see that the lower quantile of LHM shows a slight quadratic behavior, peaking near 100 and then going down, the median behave more or less similarly throughout the domain, but the higher quantile seems to decrease as the GPRD-A value increases. The effects on the quantiles of MMA are more prominent in comparison. There, the lower quantile attains a local maxima at around 140. The median, and the 95th quantile behave similarly, which can be thought of as a cubic polynomial with a local minimum and a local maximum at 100 and 140. Finally, with respect to GPRD-T, we broadly see a similar pattern as in the case of GPRD. All of these results are consistent with our earlier findings which depict that the impact is more profound in MMA.

\begin{figure}[!]
    \centering
    \begin{subfigure}{0.43\textwidth}
        \centering
        \includegraphics[width=\linewidth]{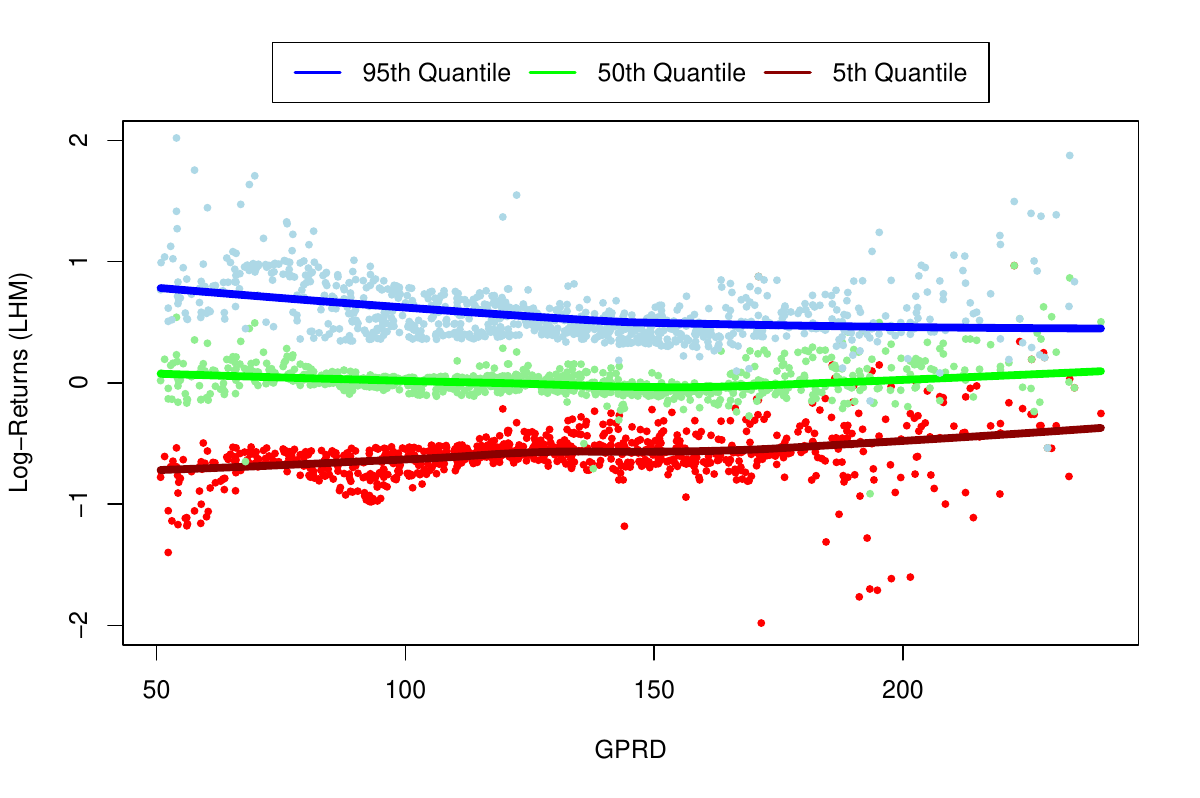}
        \label{fig:quntvsgprd_lhm}
    \end{subfigure}
    \hfill
    \begin{subfigure}{0.43\textwidth}
        \centering
        \includegraphics[width=\linewidth]{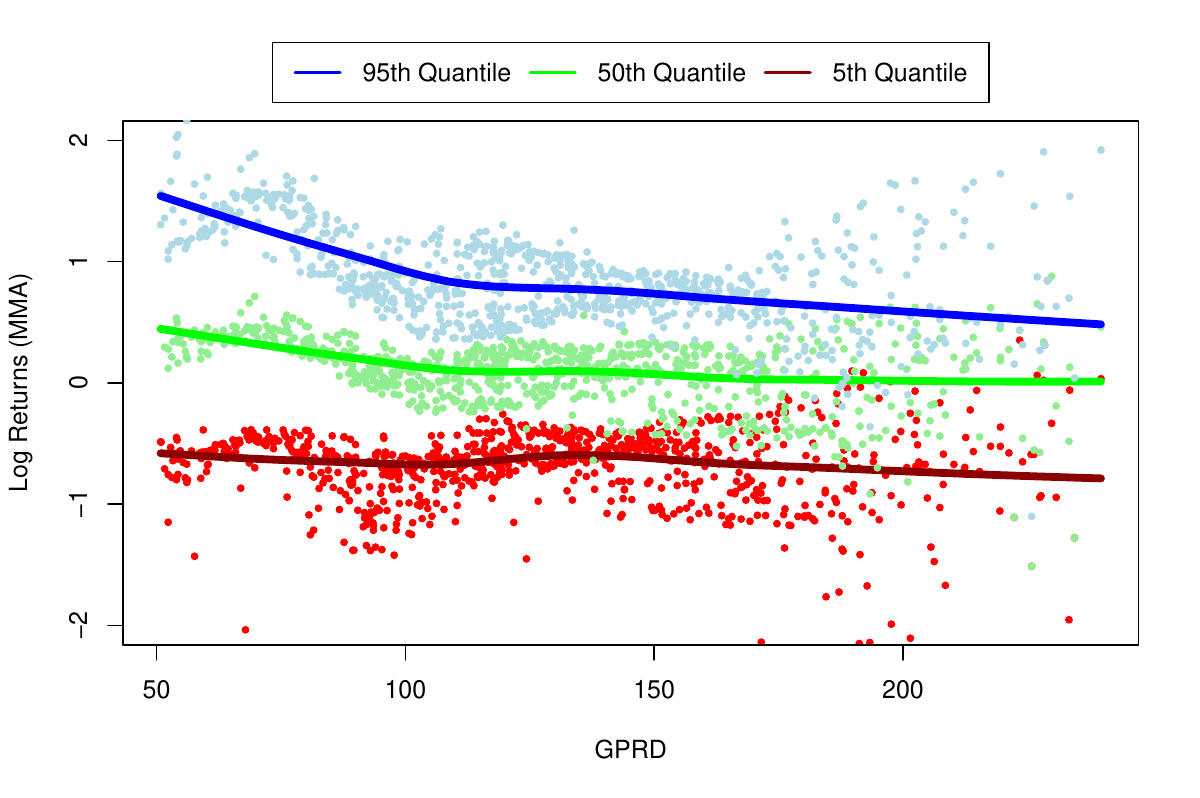}
        \label{fig:quntvsgprd_mma}
    \end{subfigure}
    \begin{subfigure}{0.43\textwidth}
        \centering
        \includegraphics[width=\linewidth]{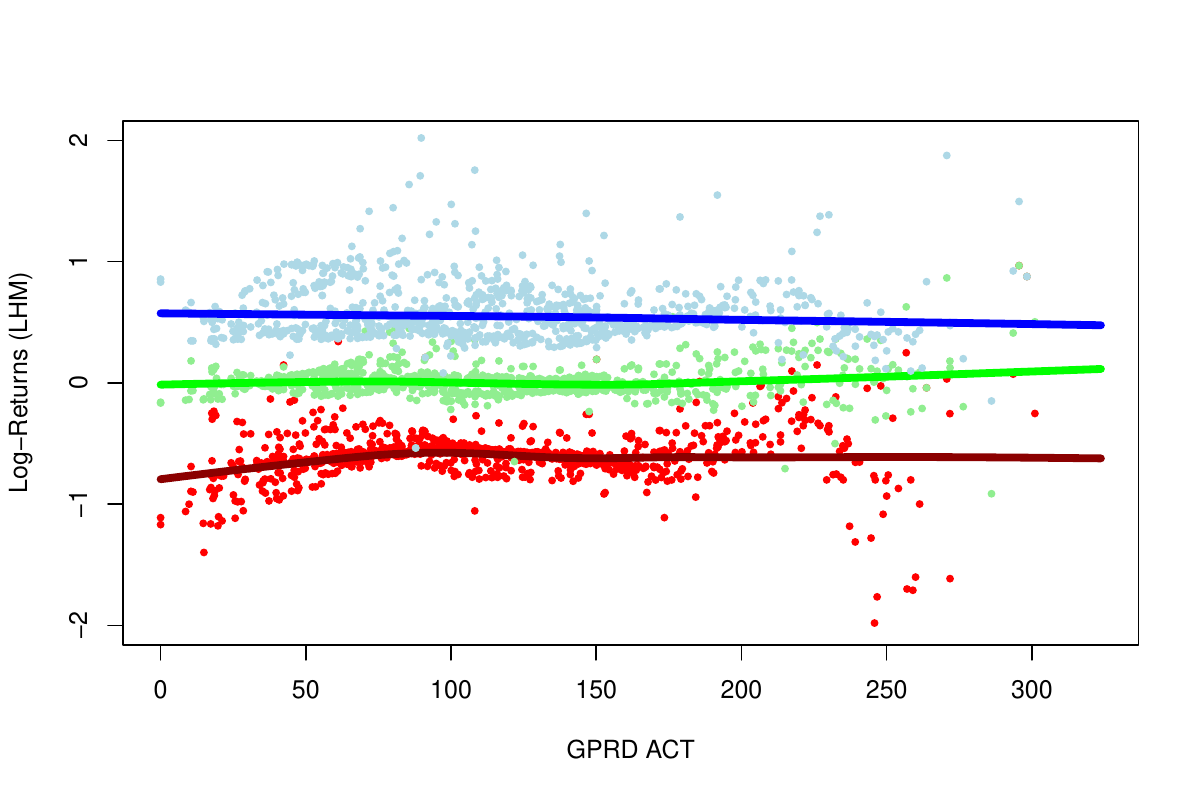}
        \label{fig:quntvsgprda_lhm}
    \end{subfigure}
    \hfill
    \begin{subfigure}{0.43\textwidth}
        \centering
        \includegraphics[width=\linewidth]{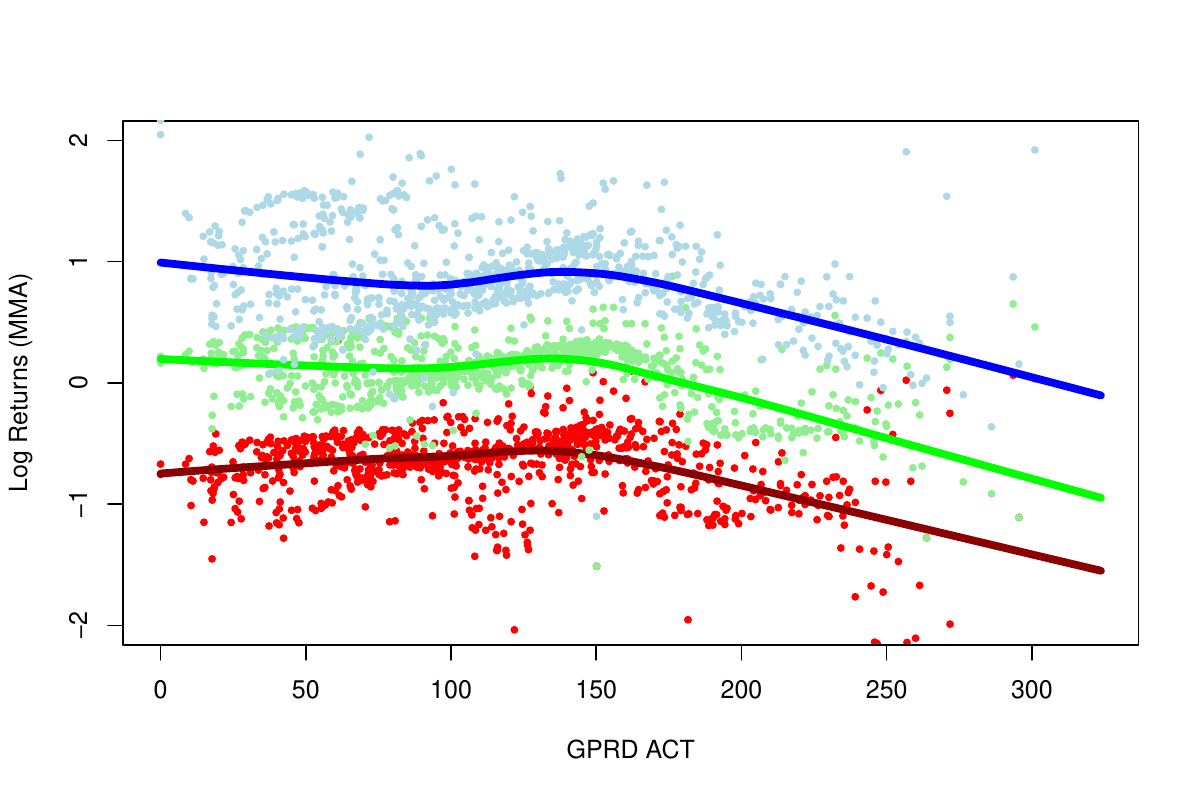}
        \label{fig:quntvsgprda_mma}
    \end{subfigure}
    \begin{subfigure}{0.43\textwidth}
        \centering
        \includegraphics[width=\linewidth]{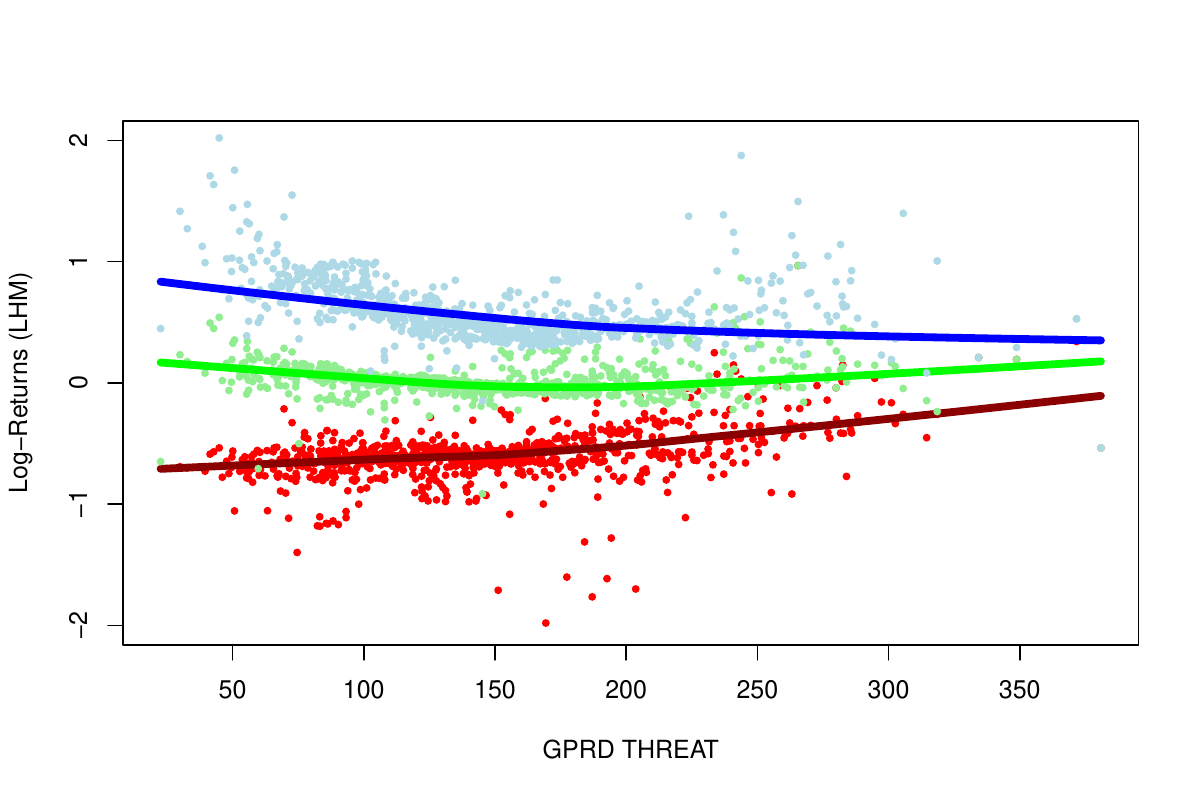}
        \label{fig:quntvsgprdt_lhm}
    \end{subfigure}
    \hfill
    \begin{subfigure}{0.43\textwidth}
        \centering
        \includegraphics[width=\linewidth]{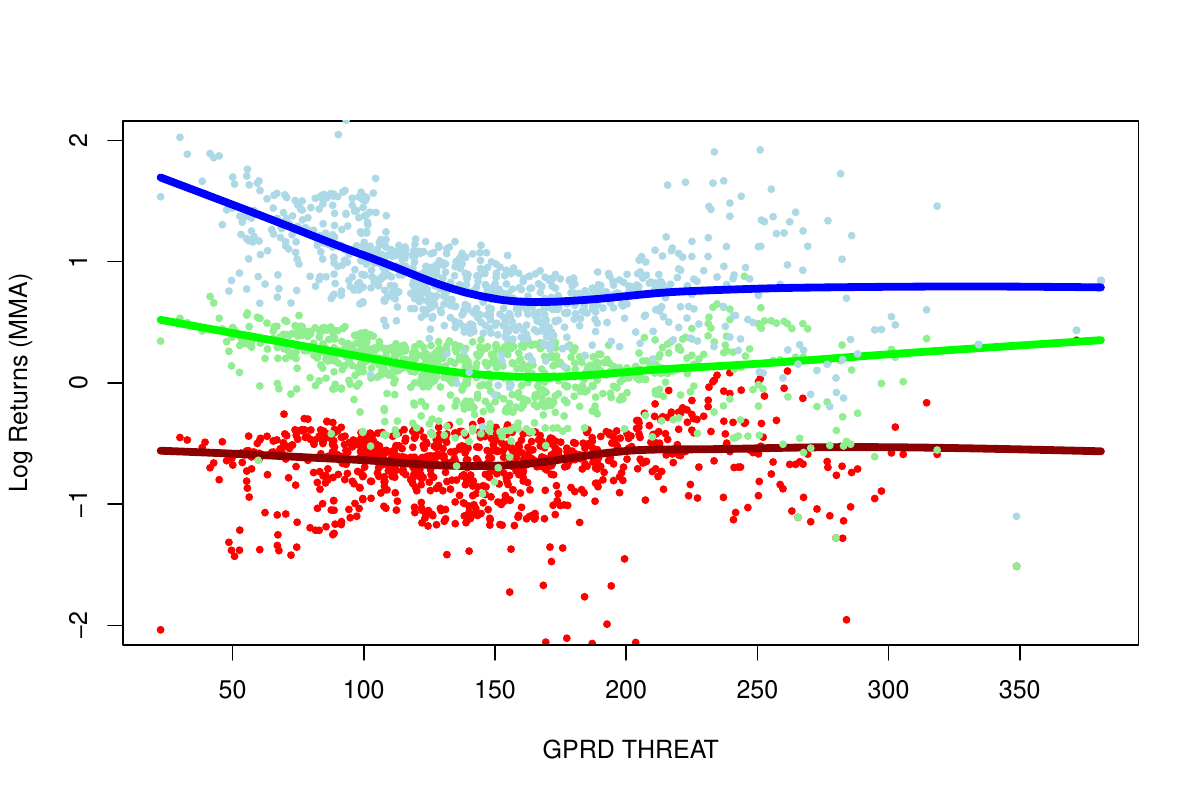}
        \label{fig:quntvsgprdt_mma}
    \end{subfigure}
    \caption{Estimates quantiles at $\tau=0.05, 0.50, 0.95$ for LHM (left panels) and MMA (right panels) against the three types of geopolitical risk indices.}
\label{fig:quntvsgprd_a_t_lhm_mma}
\end{figure}

As a final exploration, we add \Cref{fig:Var_Cov_GPR_Time} in \Cref{sec:Plots}, where we plot the components of the estimated covariance matrix against the three risk indices and time. The variance components behave largely akin to the generalized variance against the risks, that is, with increase of the risk, the value of the variance goes down. An exception is noted in the covariance between the two assets against GPRD-T, where we see an upward trend -- small in size but persistent. The effects of the RUC and the IPC are directly visible in the last three plots. There is a hint of a subtle structural break in the variances plot against time during the IPC but it is not as loud as the RUC. When we look at the plot of covariance against time, we can easily see the two big spikes for the two events. A few other smaller spikes can also be seen in the same direction which can be attributed to further escalations in the IPC. Overall, these analyses help us observe how the proposed framework can be applied to understand the impact of geopolitical risk on the stock prices of the defense and shipping industries. The considered timeline has been a turbulent period in the early twenty-first century so far, and our study offer interesting findings about the period.

\section{Conclusions}
\label{sec:conclusions}

In this paper, we developed a unified nonparametric framework for the statistical analysis of multivariate time series. Using Nadaraya-Watson kernel based estimators, we studied the multivariate conditional mean, variance–covariance matrix, and geometric quantiles, establishing consistency and convergence properties under general dependence structures. In particular, we derived asymptotic results for the conditional mean, proved consistency of the conditional variance–covariance estimator, and established consistency and non-crossing properties for the multivariate conditional quantiles. The proposed methodology was applied to examine the impact of geopolitical risk on globally exposed industries, focusing on Lockheed Martin and Maersk. Our empirical findings reveal three key insights. First, the conditional mean dynamics of both stocks were significantly affected during the Russia–Ukraine and Israel-Palestine conflicts, with narrower confidence intervals for Lockheed Martin suggesting relatively greater stability. Second, the generalized variance exhibited nonlinear behavior with respect to geopolitical risk, decreasing up to a threshold and increasing thereafter, with pronounced temporal spikes during the Russia–Ukraine conflict. Third, the conditional quantile analysis showed asymmetric tail responses: lower quantiles increased and upper quantiles decreased as geopolitical risk intensified, indicating shifting distributional risk profiles.

Beyond these findings, the framework opens several promising avenues for future research. The kernel based methodology can be extended through hybrid approaches that integrate machine learning tools, such as neural network–assisted smoothing or adaptive bandwidth selection, to enhance finite-sample performance and scalability in higher dimensions. Further theoretical work may explore uniform convergence results, inference for conditional quantile processes, and extensions to high-dimensional or functional covariate settings. Additionally, the framework can be adapted to other asset classes, systemic risk measurement, or macro-financial stress testing, where modeling joint dynamics of location, scale, and tail behavior is crucial. We believe that combining rigorous nonparametric theory with modern data-driven techniques represents a fruitful direction for advancing multivariate time series analysis.

\section*{Declaration of interests}

The authors declare no competing interests.

\section*{Funding}

This research did not receive any specific grant from funding agencies in the public, commercial, or not-for-profit sectors.

\bibliography{references}

\newpage
\setcounter{equation}{0}
\renewcommand\theequation{A.\arabic{equation}}

\appendix

\section{Proofs}
\label{sec:Proofs}

\begin{proof}[Proof of \Cref{thm:theorem1} $(i)$]
We establish the consistency of the multivariate conditional mean estimator in the stochastic regression model \eqref{eq:main_model_multivariate}, under the physical dependence setting for $\{\bm{X}_i\}$ and the bandwidth conditions $b_n\to 0$ and $n b_n^k\to\infty$ with $\bm{x} \in \R^k$ being fixed throughout. Recall from \eqref{eq:multivariate_mu_expression} and \eqref{eq:weight_definition} that the local constant estimator of the conditional mean can be expressed as a ratio, as shown below:
\begin{equation}
\label{eq:muhat_in_ratio}
    \widehat{\mu}_n(\bm{x}) = \frac{\sum_{i=1}^n
K_{b_n}(\bm{x}-\bm{X}_i)\bm{Y}_i}{\sum_{i=1}^n
K_{b_n}(\bm{x}-\bm{X}_i)} = \frac{\widehat{g}_n(\bm{x})}{\widehat{f}_n(\bm{x})},
\end{equation}
where the numerator, $\widehat{g}_n(\bm{x})$, and the denominator, $\widehat{f}_n(\bm{x})$, are defined as shown below,
\begin{equation}
\label{eq:f_and_g_def}
\widehat{g}_n(\bm{x}) = \frac{1}{n b_n^k} \sum_{i=1}^n K\!\left(\frac{\bm{x}-\bm{X}_i}{b_n}\right)\bm{Y}_i,
\qquad
\widehat{f}_n(\bm{x}) = \frac{1}{n b_n^k} \sum_{i=1}^n K\!\left(\frac{\bm{x}-\bm{X}_i}{b_n}\right) = \frac{1}{n} \sum_{i=1}^nZ_{i,n}(\bm{x}).
\end{equation}
We show $\widehat{f}_n(\bm{x}) \inprob f_{\bm{X}}(\bm{x})$ and $\widehat{g}_n(\bm{x}) \inprob \mu(\bm{x}) f_{\bm{X}}(\bm{x})$, since $f_{\bm{X}}(\bm{x})>0$ and Slutsky’s theorem, $
\widehat{\mu}_n(x) \inprob \mu(x)$.

By stationarity and the change of variables $\bm{u}=(\bm{x}-\bm{z})/b_n$,
\begin{equation*}
    \E[\widehat{f}_n(\bm{x})] = \int K(\bm{u}) f_{\bm{X}}(\bm{x}-b_n \bm{u})\,d\bm{u} \to f_{\bm{X}}(\bm{x})
\end{equation*}

by continuity of $f_{\bm{X}}$ at $\bm{x}$, and dominated convergence. Under the short range physical dependence condition stated in \Cref{sec:main_results}, kernel weighted averages satisfy a weak law of large numbers, in particular, the covariance summability implied by the physical dependence measure yields, $\text{Var}[n^{-1}\sum Z_{i,n}(\bm{x})] \leqslant C\E[Z_{0,n}(\bm{x})^2] / n$, from this one can bound the variance of the numerator:
\begin{equation*}
    \text{Var}[\widehat{f}_n(\bm{x})] = \frac{1}{n^2} \sum_{s,t = 1}^n \text{Cov}[Z_{s,n}(\bm{x}), Z_{t,n}(\bm{x})] \leqslant \frac{C}{n}\E[Z_{0,n}(\bm{x})^2]
\end{equation*}
where $C$ is a constant that depends on the dependence measure then the right hand side expression becomes
\begin{equation*}
    \frac{C}{n}\E[Z_{0,n}(\bm{x})^2] = \frac{C}{n}\int b_n^{-2k} K\left(\frac{\bm{x}-\bm{z}}{b_n}\right)^2f_{\bm{X}}(\bm{z})d\bm{z} =\frac{C}{n} b_n^{-k} \int K(\bm{u})^2 f_{\bm{X}}(\bm{x}-b_n\bm{u})d\bm{u} = \mathcal{O}\left(\frac{1}{nb_n^k}\right),
\end{equation*}
under boundedness of $f_{\bm{X}}$ near $\bm{x}$ and $K \in \mathcal{L}^2$ and $n b_n^k \to \infty$, variance vanishes, therefore $\widehat{f}_n(x) \xrightarrow{P} f_{\bm{X}}(x)$.

We use the stochastic regression model \eqref{eq:main_model_multivariate} to decompose the numerator, $\widehat{g}_n(\bm{x})$, as a summation of $A_n(\bm{x})$ and $B_n(\bm{x})$ where they are defined as
\begin{equation*}
    A_n(\bm{x}) = \frac{1}{n b_n^k} \sum_{i=1}^n K\left(\frac{\bm{x}-\bm{X}_i}{b_n}\right) \mu(\bm{X}_i),
\qquad
    B_n(\bm{x}) = \frac{1}{n b_n^k} \sum_{i=1}^n K\left(\frac{\bm{x}-\bm{X}_i}{b_n}\right) \Sigma(\bm{X}_i)\bm{e}_i.
\end{equation*}

By stationarity and change of change of variables, $\bm{u}=(\bm{x}-\bm{z})/b_n$,
\begin{equation*}
    \E[A_n(\bm{x})] = \int K(\bm{u}) \mu(\bm{x}-b_n\bm{u}) f_{\bm{X}}(\bm{x}-b_n\bm{u})d\bm{u} \to \mu(\bm{x})f_{\bm{X}}(\bm{x}) \int K(\bm{u}) d\bm{u} = \mu(\bm{x})f_{\bm{X}}(\bm{x}),
\end{equation*}
under continuity of $\mu$ and $f_{\bm{X}}$ at $\bm{x}$ and dominated convergence. A variance bound, identical to the previous result, is obtained, under local boundedness of $\mu$ near $\bm{x}$ and with $K$ localizing to a shrinking neighborhood of $\bm{x}$, hence, $\text{Var}[A_n(x)] = \mathcal{O}((n b_n^k)^{-1})$ and $A_n(x)\inprob \mu(x)f_{\bm{X}}(x)$.

For $B_n(\bm{x})$, $\bm{e}_i$ is independent of the covariate sigma field, $\mathcal{F}$, and have zero mean, so, $\E[B_n(\bm{x})|\mathcal{F}] = 0$. The conditional variance of $B_n(\bm{x})$, $\text{Var}[B_n(\bm{x})|\mathcal{F}]$, is less than or equal to, $\E[\|B_n(\bm{x})\|^2]$. Now, under stationarity, boundedness of $K$ and local boundedness of $\Sigma$, we obtain
\begin{equation*}
    \E[\|B_n(\bm{x})\|^2] \leqslant \frac{C}{n^2b_n^{2k}} n \E \left[ K\left( \frac{\bm{x} - \bm{X}_0}{b_n} \right)^2 \right] = \mathcal{O}\left(\frac{1}{nb_n^k}\right).
\end{equation*}

By change of variable in integration above, we obtain rate, hence, $B_n(x)\inprob 0$ and combining in the results for $A_n(\bm{x})$ and $B_n(\bm{x})$, we obtain, $\widehat{g}_n(\bm{x}) = A_n(\bm{x}) + B_n(\bm{x}) \inprob \mu(\bm{x})f_{\bm{X}}(\bm{x})$. From the results on $\widehat{g}_n(\bm{x})$ and $\widehat{f}_n(\bm{x})$ and Slutsky's theorem, we obtain, $\widehat{\mu}(\bm{x}) \inprob \mu(\bm{x})$. This completes the proof of \Cref{thm:theorem1} $(i)$.
\end{proof}

\begin{proof}[Proof of \Cref{thm:theorem1} $(ii)$]

We now strengthen the consistency result in part $(i)$ to an asymptotic normality statement. Throughout the proof, we fix $\bm{x} \in \R^k$ and avoid $\bm{x}$ in the notation where it is not ambiguous. As before, we assume that $K$ is a symmetric kernel with compact support (or sufficiently fast decay) and finite second moments, and that $f_{\bm{X}}$ and $\mu$ are sufficiently smooth in a neighborhood of $\bm{x}$ with $g(\bm{x}) = \mu(\bm{x})f_{\bm{X}}(\bm{x})$.

We first derive a more precise bias expansion for $\widehat{\mu}_{b_n}(\bm{x})$. Using stationarity and a change of variables, $\bm{u} = (\bm{z}-\bm{x})/b_n$, together with a multivariate Taylor expansion of $f_{\bm{X}}$ around $\bm{x}$, we obtain
\begin{equation*}
    \E[\widehat{f}_{n}(\bm{x})]
    = \int K(\bm{u}) f_{\bm{X}}(\bm{x} + b_n \bm{u}) d\bm{u}
    = f_{\bm{X}}(\bm{x}) + b_n^2 \psi_K \nabla^2 f_{\bm{X}}(\bm{x}) + \mathcal{O}(b_n^3),
\end{equation*}
Analogously, using stationarity, conditional expectation, change of variables and symmetry of kernel, together with a multivariate Taylor expansion of $\widehat{g}_n(\bm{x})$ around $\bm{x}$, we obtain
\begin{equation*}
\begin{split}
        \E[\widehat{g}_n(\bm{x})] &= \frac{1}{b_n^k}\E \left[ K\left( \frac{\bm{x}-\bm{X}_i}{b_n} \right) \bm{Y}_i \right] = \frac{1}{b_n^k} \E \left[ \E \left[ K\left( \frac{\bm{x}-\bm{X}_i}{b_n} \right) \bm{Y}_i \Bigg| \bm{X}_i \right] \right] = \frac{1}{b_n^k}\E \left[ K\left( \frac{\bm{x}-\bm{X}_i}{b_n} \right) \mu(\bm{X}_i) \right] \\
        &= \frac{1}{b_n^k} \int K \left( \frac{\bm{x} - \bm{z}}{b_n}\right) \mu(\bm{z}) f_{\bm{X}}(\bm{z}) d\bm{z} = \frac{1}{b_n^k} \int K \left( \frac{\bm{x} - \bm{z}}{b_n}\right) g(\bm{z}) d\bm{z} = \int K(\bm{u}) g(\bm{x} + b_n \bm{u}) d\bm{u}
\end{split}
\end{equation*}
Hence, we obtain
\begin{equation*}
    \E[\widehat{g}_{n}(\bm{x})]
    = g(\bm{x}) + b_n^2 \psi_K \nabla^2 g(\bm{x}) + \mathcal{O}(b_n^3),
\end{equation*}
Since  $g(\bm{x}) = \mu(\bm{x}) f_{\bm{X}}(\bm{x})$, we obtain $\nabla^2 g(\bm{x})
    = f_{\bm{X}}(\bm{x}) \nabla^2 \mu(\bm{x})
      + 2 \nabla \mu(\bm{x})^{\intercal} \nabla f_{\bm{X}}(\bm{x})
      + \mu(\bm{x}) \nabla^2 f_{\bm{X}}(\bm{x})$, further,
\begin{equation*}
    \frac{\nabla^2 g(\bm{x})}{f_{\bm{X}}(\bm{x})}
    - \mu(\bm{x}) \frac{\nabla^2 f_{\bm{X}}(\bm{x})}{f_{\bm{X}}(\bm{x})}
    = \nabla^2 \mu(\bm{x})
      + 2 \nabla \mu(\bm{x}) \frac{\nabla f_{\bm{X}}(\bm{x})}{f_{\bm{X}}(\bm{x})}
    = \rho_{\mu}(\bm{x}).
\end{equation*}

We now expand the ratio in \Cref{{eq:muhat_in_ratio}} around the deterministic point $(g(\bm{x}), f_{\bm{X}}(\bm{x}))$. Suppose, we have, $u$ and $v$, and we want to approximate $u/v$ when $u$ and $v$ are close to true values, $u_0$ and $v_0$, let, $u = u_0 + \Delta u$ and $v = v_0 + \Delta v$. Then $u/v = u_0/v_0 + \Delta u / v_0 - u_0 \Delta v / v_0^2 + \mathcal{O}(\Delta u\Delta v)$. This yields
\begin{equation*}
    \widehat{\mu}_{b_n}(\bm{x}) - \mu(\bm{x})
    = \frac{\widehat{g}_{n}(\bm{x}) - g(\bm{x})}{f_{\bm{X}}(\bm{x})}
     - \frac{\mu(\bm{x})}{f_{\bm{X}}(\bm{x})}
       \big(\widehat{f}_{n}(\bm{x}) - f_{\bm{X}}(\bm{x})\big) + R_n(\bm{x}),
\end{equation*}

Let $\widehat{\mu}_{b_n}(\bm{x}) - \mu(\bm{x}) = L_n(\bm{x}) + R_n(\bm{x})$, so, $\widehat{\mu}_{b_n}(\bm{x}) - \mu(\bm{x})
    = (L_n(\bm{x}) - \E [L_n(\bm{x})]) + \E [L_n(\bm{x})] + R_n(\bm{x})$. Then
\begin{equation*}
    \widehat{\mu}_{b_n}(\bm{x}) - \mu(\bm{x})
    = \frac{\widehat{g}_{n}(\bm{x}) - \E[\widehat{g}(\bm{x})]}{f_{\bm{X}}(\bm{x})}
     - \frac{\mu(\bm{x})}{f_{\bm{X}}(\bm{x})}
       \big(\widehat{f}_{n}(\bm{x}) - \E[\widehat{f}_{n}(\bm{x})]\big) \\
     + \Big\{\E[\widehat{\mu}_{b_n}(\bm{x})] - \mu(\bm{x}) - \E[R_n] \Big\}
      + \widetilde{R}_n(\bm{x}),
\end{equation*}
where $\widetilde{R}_n(\bm{x})$ collects the higher order terms involving products of
$(\widehat{g}_{n}(\bm{x}) - g(\bm{x}))$ and $(\widehat{f}_{n}(\bm{x}) - f_{\bm{X}}(\bm{x}))$ like $R_n$ but they also have deterministic component, $(g-\E[\widehat{g}_n])/f - \mu(f-\E[ \widehat{f}])/f$. Let, $r_n(\bm{x}) = \widetilde{R}_n(\bm{x}) - \E[R_n(\bm{x})]$. Using the variance calculations from the proof of part $(i)$, we know that $\widehat{g}_{n}(\bm{x}) - g(\bm{x})
    = \mathcal{O}_{\P}((nb_n^k)^{1/2})$ and $\widehat{f}_{b_n}(\bm{x}) - f(\bm{x})
    = \mathcal{O}_{\P}((nb_n^k)^{1/2})$, thus, $\widetilde{R}_n(\bm{x})$, $R_n(\bm{x})$ and $r_n(\bm{x})$ are of order $\mathcal{O}_{\P}((n b_n^k)^{-1})$, and therefore $\sqrt{n b_n^k} \widetilde{R}_n(\bm{x}) \inprob 0$, $\sqrt{n b_n^k} R_n(\bm{x}) \inprob 0$ and $\sqrt{n b_n^k} r_n(\bm{x}) \inprob 0$. Using debiasing of $\widehat{\mu}_{b_n}(\bm{x})$ estimate, we obtain
\begin{equation*}
    \widehat{\mu}_{b_n}(\bm{x}) - \mu(\bm{x}) - b_n^2 \psi_K \rho_{\mu}(\bm{x})
    = \frac{\widehat{g}_{n}(\bm{x}) - g(\bm{x})}{f_{\bm{X}}(\bm{x})}
     - \frac{\mu(\bm{x})}{f_{\bm{X}}(\bm{x})}
       \big(\widehat{f}_{n}(\bm{x}) - f_{\bm{X}}(\bm{x})\big)
     + r_n(\bm{x}),
\end{equation*}
where $b_n^2 \psi_K \rho_{\mu}(\bm{x})$ is the leading bias term. Under the bandwidth condition $ \sqrt{n b_n^k} b_n^2 \to 0$ under the bandwidth condition, $nb_n^{k+4} \to 0$ so the bias term is negligible at the $\sqrt{n b_n^k}$ scale once it has been subtracted explicitly.

We rearrange the $L_n(\bm{x})$, $(\widehat{g}-g) / f - \mu (\widehat{f} - f) / f = (\widehat{g} - \mu \widehat{f}) / f - (g - \mu f) / f$, the second term is equal to zero since $g = \mu f$, so we have, $(\widehat{g}-g) / f - \mu (\widehat{f} - f) / f = (\widehat{g} - \mu \widehat{f}) / f$, now we can plug-in the definitions of $\widehat{g}$ and $\widehat{f}$ from \eqref{eq:f_and_g_def}. We write $\bm{Y}_i - \mu(\bm{x})$ as $(\bm{Y}_i - \mu(\bm{X}_i)) + (\mu(\bm{X}_i)\ - \mu(\bm{x}))$. The first sum has conditionally zero mean for the regression error. The second sum has small local difference due to bounded support (the distance between $\bm{X}_i$ and $\bm{x}$ is of order $\mathcal{O}(b_n)$) and symmetry of kernel (only $b_n^2$ order components are left that goes to zero after scaling by $\sqrt{nb_n^k}$) and smoothness of the mean function ($\mu(\bm{X}_i) - \mu(\bm{x}) = \mathcal{O}(b_n)$). Therefore, the second sum gets absorbed into $o_P(1)$ term after scaling, so we have,
\begin{equation*}
        \sqrt{n b_n^k}
    \left[ \frac{\widehat{g}_{n}(\bm{x}) - g(\bm{x})}{f_{\bm{X}}(\bm{x})}
     - \frac{\mu(\bm{x})}{f_{\bm{X}}(\bm{x})}
       \big(\widehat{f}_{n}(\bm{x}) - f_{\bm{X}}(\bm{x})\big) \right]
    = \frac{\sqrt{nb_n^k}}{nb_n^kf_{\bm{X}}(\bm{x})} \sum_{i=1}^n K\left( \frac{\bm{x} - \bm{X}_i}{b_n}\right) (\bm{Y}_i - \mu(\bm{X}_i)) + o_P(1).
\end{equation*}

For the stochastic part, it is convenient to write the leading term as a normalized sum. Define
\begin{equation*}
    \bm{\xi}_{n,i}(\bm{x})
    = \frac{1}{b_n^{k/2} f_{\bm{X}}(\bm{x})}
      K\!\left(\frac{\bm{x} - \bm{X}_i}{b_n}\right)
      \Big\{\bm{Y}_i - \mu(\bm{X}_i)\Big\},
\end{equation*}
so that, up to terms that vanish after multiplication by $\sqrt{n b_n^k}$,
\begin{equation*}
    \sqrt{n b_n^k}
    \Big[\widehat{\mu}_{b_n}(\bm{x}) - \mu(\bm{x})
         - b_n^2 \psi_K \rho_{\mu}(\bm{x})\Big]
    = \frac{1}{\sqrt{n}} \sum_{i=1}^n \bm{\xi}_{n,i}(\bm{x}) + o_P(1).
\end{equation*}
The sequence $\{\bm{\xi}_{n,i}(\bm{x})\}_{i=1}^n$ forms a triangular array of mean zero random vectors (with negligible bias, $\mathcal{O}(b_n^2)$, already considered). Under our stationarity and smoothness assumptions, with $\mathbb{V}_e = \E[\bm{e}_i\bm{e}^\top_i]$ and taking conditional expectation and change of variables, standard kernel calculations give 
\begin{equation*}
    \text{Var}\!\left(
    \frac{1}{\sqrt{n}} \sum_{i=1}^n \bm{\xi}_{n,i}(\bm{x})
    \right)
    \to \frac{\phi_K}{f_{\bm{X}}(\bm{x})} \Sigma^{1/2}(\bm{x}) \mathbb{V}_e \Sigma^{1/2}(\bm{x})^{\top}.
\end{equation*}

Dependence affects the correlation, but because the kernel localizes to a shrinking neighborhood around $\bm{x}$ and the dependence is controlled by $\Lambda_n$, the leading variance is still the same as the diagonal elements.

Because we are working with a dependent process, we need a central limit theorem that accommodates both dependence and a triangular structure. This is achieved by first constructing a martingale approximation using the functional dependence measure in \eqref{eq:functional_dependence_measure}. Following \cite{zhao2008confidence}, the short range dependence condition $\Theta_{\infty} < \infty$ together with the definition of $\Lambda_n$ in \eqref{eq:srd_lrd} imply that there exists a $\R^p$ valued martingale difference array $\{\bm{D}_{n,i}(\bm{x}),\mathcal{F}_i\}$ such that for some constant $C$ independent of $n$,
\begin{equation*}
    \max_{1 \leqslant m \leqslant n}
    \Big\|
    \sum_{i=1}^m \big(\bm{\xi}_{n,i}(\bm{x}) - \bm{D}_{n,i}(\bm{x})\big)
    \Big\|_2
    \leqslant C\,\Lambda_n
    \left(\frac{b_n^{3}}{n} + \frac{1}{n^2}\right).
\end{equation*}
 By the bandwidth and dependence condition, as in \eqref{eq:functional_dependence_measure},
\begin{equation*}
    \frac{1}{\sqrt{n}} \sum_{i=1}^n
    \big(\bm{\xi}_{n,i}(\bm{x}) - \bm{D}_{n,i}(\bm{x})\big)
    \xrightarrow{P} 0.
\end{equation*}
Hence, the asymptotic distribution of $\sum_{i=1}^n \bm{\xi}_{n,i}(\bm{x}) / \sqrt{n}$ is the same as that of $\sum_{i=1}^n \bm{D}_{n,i}(\bm{x}) / \sqrt{n}$.

We now apply the multivariate martingale central limit theorem of \cite{helland1982central} (Theorem 3.3) to the martingale difference array $\{\bm{D}_{n,i}(\bm{x}),\mathcal{F}_i\}$. The boundedness of the kernel $K$, together with the assumption $\E[\|\bm{Y}_0\|^2] < \infty$, ensures the conditional Lindeberg condition. The conditional variance condition follows from the variance calculation above and from the short range dependence condition, which guarantees that the conditional and unconditional covariances coincide asymptotically. In particular, we obtain
\begin{equation*}
    \frac{\sqrt{n b_n^k f_{\bm{X}}(\bm{x})}}{\sqrt{\phi_K}}
    (\Sigma^{1/2}(\bm{x}) \mathbb{V}_e \Sigma^{1/2}(\bm{x})^{\top})^{-1/2}
    \Big[\widehat{\mu}_{b_n}(\bm{x}) - \mu(\bm{x})
         - b_n^2 \psi_K \rho_{\mu}(\bm{x})\Big]
    \xrightarrow{d} N_p(\bm{0}, \bm{I}_p).
\end{equation*}

Finally, from part $(i)$ we know that $\widehat{f}_{n}(\bm{x}) \xrightarrow{P} f_{\bm{X}}(\bm{x})$. Therefore, by Slutsky's theorem we can replace $f_{\bm{X}}(\bm{x})$ with $\widehat{f}_{n}(\bm{x})$ in the scaling factor, which yields
\begin{equation}
    \frac{\sqrt{n b_n^k \widehat{f}_{n}(\bm{x})}}{\sqrt{\phi_K}}
    (\Sigma^{1/2}(\bm{x}) \mathbb{V}_e \Sigma^{1/2}(\bm{x})^{\top})^{-1/2}
    \Big[\widehat{\mu}_{b_n}(\bm{x}) - \mu(\bm{x})
         - b_n^2 \psi_K \rho_{\mu}(\bm{x})\Big]
    \xrightarrow{d} N_p(\bm{0}, \bm{I}_p),
\end{equation}
as claimed in \eqref{eq:thm1}. We assume $\mathbb{V}_e = \I$, then, $(\Sigma^{1/2}(\bm{x}) \mathbb{V}_e \Sigma^{1/2 \top}(\bm{x}))^{-1/2} = \Sigma^{1/2}$, so our result becomes,
\begin{equation}
    \frac{\sqrt{n b_n^k \widehat{f}_{n}(\bm{x})}}{\sqrt{\phi_K}}
    \Sigma^{1/2}(\bm{x})
    \Big[\widehat{\mu}_{b_n}(\bm{x}) - \mu(\bm{x})
         - b_n^2 \psi_K \rho_{\mu}(\bm{x})\Big]
    \xrightarrow{d} N_p(\bm{0}, \bm{I}_p),
\end{equation}

This completes the proof of \Cref{thm:theorem1} $(ii)$.
\end{proof}

\begin{proof}[Proof of \Cref{thm:theorem2}]

In order to obtain consistency of the sample estimate of the variance matrix, we require a local uniform consistency of the multivariate mean estimate and this can be obtained under the assumptions of \Cref{thm:theorem1} and regularity assumption on kernel, $K \in \mathcal{K}$, is compactly supported otherwise is sufficiently fast decaying and is Lipschitz, i.e, $\mid K(\bm{u}) - K(\bm{v}) \leqslant L\| \bm{u} - \bm{v}\|\mid \; \forall \; \bm{u}, \bm{v}$ and $\E[\|\bm{Y}_0\|^2] < \infty$.

\begin{theorem}
\label{thm:theorem6}
    Assume $\mu(\bm{u})$, $f_{\bm{X}}(\bm{u}) \in \mathcal{C}^2$ on $\bm{x}^{\epsilon}$, with $f_{\bm{X}}(\bm{u}) \geqslant c > 0 \; \forall \; \bm{u} \in \bm{x}^{\epsilon}$ and $\E[\|\bm{Y}_0\|^{2+\delta}] < \infty$. The kernel function, $K \in \mathcal{K}$, is compactly supported otherwise is sufficiently fast decaying and is Lipschitz, i.e, $\mid K(\bm{u}) - K(\bm{v}) \leqslant L\| \bm{u} - \bm{v}\|\mid \; \forall \; \bm{u}, \bm{v}$. Further, $ b_n^2 + 1/nb_n^k + \log n/nb_n^k \to 0$ then for some $\epsilon > 0$, $\sup_{\|\bm{u}-\bm{x}\| \leqslant \epsilon} \|\widehat{\mu}_n(\bm{u})-\mu(\bm{u})\| \inprob 0$, as, $\; n \to \infty$.
\end{theorem}

\begin{proof}[Proof of \Cref{thm:theorem6}]
    Recall that the multivariate Nadaraya--Watson estimator of the conditional mean is given by
    \begin{equation*}
        \widehat{\mu}_n(\bm{u})
        = \frac{N_n(\bm{u})}{D_n(\bm{u})}, \qquad
        N_n(\bm{u}) = \frac{1}{b_n^k}\sum_{t=1}^n K(\bm{u} - \bm{X}_t)\bm{Y}_t, \quad
        D_n(\bm{u}) = \frac{1}{b_n^k}\sum_{t=1}^n K(\bm{u} - \bm{X}_t),
    \end{equation*}
    where $K$ is a bounded Lipschitz kernel with compact (or sufficiently rapidly decaying) support and the usual moment conditions. We introduce the rescaled versions
    \begin{equation*}
        \widetilde{D}_n(\bm{u}) = \frac{D_n(\bm{u})}{n},
        \qquad
        \widetilde{N}_n(\bm{u}) = \frac{N_n(\bm{u})}{n},
    \end{equation*}
    so that $\widehat{\mu}_n(\bm{u}) = \widetilde{N}_n(\bm{u}) / \widetilde{D}_n(\bm{u})$.
    
    We first study $\widetilde{D}_n(\bm{u})$. For fixed $\bm{u}$, by stationarity and the definition of $f_{\bm{X}}$,
    \begin{equation*}
        \E[\widetilde{D}_n(\bm{u})]
        = \frac{1}{b_n^k} \int K\left(\frac{\bm{u} - \bm{x}}{b_n}\right) f_{\bm{X}}(\bm{x}) \, d\bm{x}
        = \int K(\bm{v}) f_{\bm{X}}(\bm{u} - b_n \bm{v}) \, d\bm{v},
    \end{equation*}
    where we use the change of variables $\bm{v} = (\bm{u} - \bm{x})/b_n$. Since $f_{\bm{X}} \in \mathcal{C}^2(\bm{x}^{\epsilon})$ with bounded second derivatives and $K$ has compact support and integrates to one with zero first moments, a Taylor expansion of $f_{\bm{X}}(\bm{u} - b_n \bm{v})$ around $\bm{u}$ together with the kernel moment conditions yields
    \begin{equation}
        \label{eq:D_bias}
        \E[\widetilde{D}_n(\bm{u})]
        = f_{\bm{X}}(\bm{u}) + \mathcal{O}(b_n^2),
    \end{equation}
    where the remainder term is of order $b_n^2$ uniformly over $\{\bm{u} : \|\bm{u} - \bm{x}\| \leqslant \epsilon\}$ for some $\epsilon > 0$. Thus the bias of $\widetilde{D}_n(\bm{u})$ is uniformly $\mathcal{O}(b_n^2)$ on $\bm{x}^{\epsilon}$.
    
    For the variance, write
    \begin{equation*}
        \widetilde{D}_n(\bm{u})
        = \frac{1}{n b_n^k} \sum_{t=1}^n K_{b_n}(\bm{u} - \bm{X}_t)
        = \frac{1}{n} \sum_{t=1}^n Z_{t,n}(\bm{u}),
        \qquad
        Z_{t,n}(\bm{u}) = b_n^{-k} K\!\left( \frac{\bm{u} - \bm{X}_t}{b_n} \right).
    \end{equation*}
    
    Under the weak dependence assumption, $\Lambda_n = \mathcal{O}(n)$, standard covariance bounds imply
    \begin{equation*}
        \text{Var}(\widetilde{D}_n(\bm{u}))
        = \frac{1}{n^2} \sum_{s,t=1}^n \Cov(Z_{s,n}(\bm{u}), Z_{t,n}(\bm{u}))
        = \mathcal{O}\!\left( \frac{1}{n} \E[Z_{0,n}(\bm{u})^2] \right).
    \end{equation*}
    A direct calculation shows that
    \begin{equation*}
        \E[Z_{0,n}(\bm{u})^2]
        = \int b_n^{-2k} K\!\left( \frac{\bm{u} - \bm{x}}{b_n} \right)^2 f_{\bm{X}}(\bm{x}) \, d\bm{x}
        = \mathcal{O}\!\left( \frac{1}{b_n^k} \right),
    \end{equation*}
    uniformly over $\|\bm{u} - \bm{x}\| \leqslant \epsilon$, since $K$ and $f_{\bm{X}}$ are bounded and $K$ has compact support. Hence
    \begin{equation}
        \label{eq:D_var}
        \text{Var}(\widetilde{D}_n(\bm{u}))
        = \mathcal{O}\!\left( \frac{1}{n b_n^k} \right)
    \end{equation}
    uniformly over $\|\bm{u} - \bm{x}\| \leqslant \epsilon$. In particular, for each fixed $\bm{u}$, $\widetilde{D}_n(\bm{u}) - \E[\widetilde{D}_n(\bm{u})] \inprob 0$ as $n \to \infty$.
    
    We now upgrade this pointwise convergence to local uniform convergence over $\bm{x}^{\epsilon}$. Let $\beta_n > 0$ be a sequence such that $\beta_n \to 0$, $\beta_n / b_n^{k+1} \to 0$ and $\beta_n^{-k} / (n b_n^k) \to 0$ as $n \to \infty$. The bandwidth condition
    \begin{equation*}
        b_n + \frac{1}{n b_n^k} + \frac{\log n}{n b_n^k} \to 0
    \end{equation*}
    guarantees that such a choice of $\beta_n$ exists, suppose, $\beta^{-k}_n = (nb_n^k)^{1/2}$. Construct a finite grid $\mathcal{U}_n = \{\bm{u}_{n,1},\dots,$ $\bm{u}_{n,M_n}\}$ $\subset \bm{x}^{\epsilon}$ such that for every $\bm{u}$ with $\|\bm{u} - \bm{x}\| \leqslant \epsilon$ there exists $\bm{u}_{n,j}$ with $\|\bm{u} - \bm{u}_{n,j}\| \leqslant \beta_n$. Then $M_n = \mathcal{O}(\beta_n^{-k})$.

    On the grid points, an exponential inequality for weakly dependent sequences, defined by the physical dependence measure, $\Lambda_n$, applied to the bounded variables $Z_{t,n}(\bm{u}_{n,j}) / n$ together with the variance bound in \eqref{eq:D_var}, Chebyshev, and a union bound over $j = 1,\dots,M_n$ implies that for any fixed $\eta > 0$,
    \begin{equation}
        \label{eq:D_grid}
        \max_{1 \leqslant j \leqslant M_n}
        \big| \widetilde{D}_n(\bm{u}_{n,j}) - \E[\widetilde{D}_n(\bm{u}_{n,j})] \big|
        \inprob 0.
    \end{equation}
    
    To control points between the grid, we use the Lipschitz continuity of $K$. For any $\bm{u},\bm{v} \in \bm{x}^{\epsilon}$ we have
    \begin{equation*}
        \big|K_{b_n}(\bm{u} - \bm{X}_t) - K_{b_n}(\bm{v} - \bm{X}_t)\big|
        = \frac{1}{b_n^k}\left| K\!\left( \frac{\bm{u} - \bm{X}_t}{b_n} \right)
               - K\!\left( \frac{\bm{v} - \bm{X}_t}{b_n} \right) \right|
        \leqslant \frac{1}{b_n^k} \frac{L}{b_n} \|\bm{u} - \bm{v}\|,
    \end{equation*}
    where $L$ is the Lipschitz constant of $K$. Consequently,
    \begin{equation}
        \label{eq:D_Lip}
        \big|\widetilde{D}_n(\bm{u}) - \widetilde{D}_n(\bm{v})\big|
        \leqslant \frac{L}{n b_n^k} \sum_{t=1}^n \frac{\|\bm{u} - \bm{v}\|}{b_n}
        = \frac{L}{b_n^{k+1}} \|\bm{u} - \bm{v}\|.
    \end{equation}
    An analogous bound holds for the expectations $\E[\widetilde{D}_n(\bm{u})]$ by the smoothness of $f_{\bm{X}}$ and the compact support of $K$. Therefore, for any $\bm{u}$ and its nearest grid point $\bm{u}_{n,j}$ with $\|\bm{u} - \bm{u}_{n,j}\| \leqslant \beta_n$,
    \begin{equation*}
        \big|\widetilde{D}_n(\bm{u}) - \widetilde{D}_n(\bm{u}_{n,j})\big|
        + \big|\E[\widetilde{D}_n(\bm{u})] - \E[\widetilde{D}_n(\bm{u}_{n,j})]\big|
        \leqslant C \frac{\beta_n}{b_n^{k+1}} \to 0,
    \end{equation*}
    for some constant $C > 0$ since $|\E[\widetilde{D}_n(\bm{u})]-\widetilde{D}_n(\bm{v})]| = \big|\int K(\bm{z})[f_{\bm{X}}(\bm{u}-b_n\bm{z}) - f_{\bm{X}}(\bm{v}-b_n\bm{z})]d\bm{z} \big| \leqslant \|\nabla f_{\bm{X}}\|_{\infty} \|\bm{u} - \bm{v}\| \int |K(\bm{z})|d\bm{z}$. Combining this with \eqref{eq:D_bias} and \eqref{eq:D_grid} we obtain
    \begin{equation}
        \label{eq:D_uniform}
        \sup_{\|\bm{u} - \bm{x}\| \leqslant \epsilon}
        \big| \widetilde{D}_n(\bm{u}) - f_{\bm{X}}(\bm{u}) \big|
        \inprob 0.
    \end{equation}
    Since $f_{\bm{X}}(\bm{u}) \geqslant c > 0$ for all $\bm{u} \in \bm{x}^{\epsilon}$, \eqref{eq:D_uniform} implies that there exists $c_0 > 0$ such that
    \begin{equation}
        \label{eq:D_lower}
        \inf_{\|\bm{u} - \bm{x}\| \leqslant \epsilon} \widetilde{D}_n(\bm{u})
        \inprob c_0, \qquad c_0 \geqslant c/2 > 0.
    \end{equation}
    
    We now analyze the numerator. By definition,
    \begin{equation*}
        \widetilde{N}_n(\bm{u})
        = \frac{1}{n b_n^k} \sum_{t=1}^n K_{b_n}(\bm{u} - \bm{X}_t)\bm{Y}_t
        = \frac{1}{n} \sum_{t=1}^n H_{t,n}(\bm{u}),
        \qquad
        H_{t,n}(\bm{u}) = b_n^{-k} K\!\left( \frac{\bm{u} - \bm{X}_t}{b_n} \right)\bm{Y}_t.
    \end{equation*}
    Using the definition of the regression function $\mu(\bm{u}) = \E(\bm{Y}_t \mid \bm{X}_t = \bm{u})$, we obtain
    \begin{equation*}
        \E[\widetilde{N}_n(\bm{u})]
        = \frac{1}{b_n^k} \int K_{b_n}(\bm{u} - \bm{x}) \mu(\bm{x}) f_{\bm{X}}(\bm{x}) \, d\bm{x}
        = \int K(\bm{v}) \mu(\bm{u} - b_n \bm{v}) f_{\bm{X}}(\bm{u} - b_n \bm{v}) \, d\bm{v}.
    \end{equation*}
    The product $\mu(\cdot) f_{\bm{X}}(\cdot)$ belongs to $\mathcal{C}^2(\bm{x}^{\epsilon})$ with bounded second derivatives, and $K$ satisfies the same moment conditions as before. A Taylor expansion of $\mu(\bm{u} - b_n \bm{v}) f_{\bm{X}}(\bm{u} - b_n \bm{v})$ around $\bm{u}$ therefore yields
    \begin{equation}
        \label{eq:N_bias}
        \E[\widetilde{N}_n(\bm{u})]
        = \mu(\bm{u}) f_{\bm{X}}(\bm{u}) + \mathcal{O}(b_n^2),
    \end{equation}
    uniformly over $\|\bm{u} - \bm{x}\| \leqslant \epsilon$.
    
    For the variance, the same weak dependence arguments as before give
    \begin{equation*}
        \text{Var}(\widetilde{N}_n(\bm{u}))
        = \mathcal{O}\!\left( \frac{1}{n} \E\|H_{0,n}(\bm{u})\|^2 \right),
    \end{equation*}
    where
    \begin{equation*}
        \E\|H_{0,n}(\bm{u})\|^2
        \leqslant C_1 \int b_n^{-2k} K\!\left( \frac{\bm{u} - \bm{x}}{b_n} \right)^2 \E\|\bm{Y}_0\|^2 f_{\bm{X}}(\bm{x}) \, d\bm{x}
        = \mathcal{O}\!\left( \frac{1}{b_n^k} \right),
    \end{equation*}
    since $\E\|\bm{Y}_0\|^2 < \infty$, $K$ and $f_{\bm{X}}$ are bounded and $K$ has compact support. Hence,
\begin{equation}
    \label{eq:N_var}
    \text{Var}\!\left(\widetilde{N}_n(\bm{u})\right)
    = \mathcal{O}\!\left( \frac{1}{n b_n^k} \right)
\end{equation}
uniformly over $\|\bm{u} - \bm{x}\| \leqslant \epsilon$. For each grid point $\bm{u}_{n,j}$, Chebyshev's inequality together with 
\eqref{eq:N_var} implies 
\begin{equation*}
    \P\!\left(
\big\|\widetilde{N}_n(\bm{u}_{n,j}) -\E[\widetilde{N}_n(\bm{u}_{n,j})]\big\| > \eta \right) \leqslant \frac{C}{\eta^2 n b_n^k},
\end{equation*}
uniformly in $j$. Applying a union bound over $j = 1, \dots, M_n$, we obtain
\begin{equation*}
    \P\!\left( \max_{1 \leqslant j \leqslant M_n} \big\|\widetilde{N}_n(\bm{u}_{n,j})
      - \E[\widetilde{N}_n(\bm{u}_{n,j})]\big\| > \eta \right) \leqslant \frac{C M_n}{\eta^2 n b_n^k}.
\end{equation*}

Since $M_n = \mathcal{O}(\beta_n^{-k})$ and 
$\beta_n^{-k}/(n b_n^k) \to 0$, the right hand side converges to zero. Hence,
\begin{equation*}
    \max_{1 \leqslant j \leqslant M_n} \big\|\widetilde{N}_n(\bm{u}_{n,j}) - \E[\widetilde{N}_n(\bm{u}_{n,j})]\big\| \inprob 0.
\end{equation*}

To extend this convergence from the discrete grid to the entire 
neighborhood $\{\|\bm{u}-\bm{x}\| \leqslant \epsilon\}$, 
we use the Lipschitz continuity of the kernel. 
For any $\bm{u}$ in the neighborhood, there exists a grid point 
$\bm{u}_{n,j}$ with $\|\bm{u}-\bm{u}_{n,j}\| \leqslant \beta_n$, and by kernel smoothness, $\big\|\widetilde{N}_n(\bm{u}) - \widetilde{N}_n(\bm{u}_{n,j})\big\|$ is uniformly negligible as $\beta_n \to 0$. Combining this interpolation step with the grid convergence above yields
\begin{equation}
    \label{eq:N_uniform}
    \sup_{\|\bm{u} - \bm{x}\| \leqslant \epsilon}
    \big\| \widetilde{N}_n(\bm{u})
          - \E[\widetilde{N}_n(\bm{u})] \big\|
    \inprob 0.
\end{equation}

Combining \eqref{eq:N_bias} and \eqref{eq:N_uniform} gives
\begin{equation}
    \label{eq:N_limit}
    \sup_{\|\bm{u} - \bm{x}\| \leqslant \epsilon}
    \big\| \widetilde{N}_n(\bm{u})
          - \mu(\bm{u}) f_{\bm{X}}(\bm{u}) \big\|
    \inprob 0.
\end{equation}    
    
    Finally, we merge the results for numerator and denominator to obtain the desired local uniform convergence of the estimator. For any $\bm{u}$ with $\|\bm{u} - \bm{x}\| \leqslant \epsilon$,
    \begin{equation*}
        \widehat{\mu}_n(\bm{u}) - \mu(\bm{u})
        = \frac{\widetilde{N}_n(\bm{u})}{\widetilde{D}_n(\bm{u})}
          - \frac{\mu(\bm{u}) f_{\bm{X}}(\bm{u})}{f_{\bm{X}}(\bm{u})}
        = A_n(\bm{u}) + B_n(\bm{u}),
    \end{equation*}
    where
    \begin{equation*}
        A_n(\bm{u})
        = \frac{\widetilde{N}_n(\bm{u}) - \mu(\bm{u}) f_{\bm{X}}(\bm{u})}
               {\widetilde{D}_n(\bm{u})},
        \qquad
        B_n(\bm{u})
        = \mu(\bm{u}) \left(
          \frac{f_{\bm{X}}(\bm{u}) - \widetilde{D}_n(\bm{u})}
               {f_{\bm{X}}(\bm{u}) \widetilde{D}_n(\bm{u})}
        \right).
    \end{equation*}
    Taking supremum over $\|\bm{u} - \bm{x}\| \leqslant \epsilon$ and using the boundedness of $\mu$ and $f_{\bm{X}}$ on the compact set $\{\bm{u} : \|\bm{u} - \bm{x}\| \leqslant \epsilon\}$, since $\mu \in \mathcal{C}^2(\bm{u}^{\epsilon})$ thus $\sup_{\bm{u}\in \bm{x}^{\epsilon}}\|\mu (\bm{u})\| < \infty$, together with \eqref{eq:D_uniform}, \eqref{eq:D_lower} and \eqref{eq:N_limit}, we obtain
    \begin{equation*}
        \sup_{\|\bm{u} - \bm{x}\| \leqslant \epsilon}
        \big\| \widehat{\mu}_n(\bm{u}) - \mu(\bm{u}) \big\|
        \inprob 0.
    \end{equation*}
    This proves the local uniform consistency of the multivariate Nadaraya-Watson conditional mean estimator.
\end{proof}

With the result of local uniform convergence of the multivariate mean estimator we proceed to establish the consistency of the sample variance matrix estimate. Let $\bm{x}\in\R^k$ be fixed. Recall that the sample Nadaraya Watson type estimator of the conditional variance matrix is given by
\begin{equation*}
    \widehat{\Sigma}(\bm{x}) = \frac{\sum_{t=1}^{n}K_{b_n}(\bm{x}-\bm{X}_t)(\bm{Y}_t-\widehat{\mu}_n(\bm{X}_t))(\bm{Y}_t-\widehat{\mu}_n(\bm{X}_t))^\top
}{\sum_{t=1}^{n}K_{b_n}(\bm{x}-\bm{X}_t)}
\in \R^{p \times p},
\end{equation*}
where $D_n(\bm{x})=\sum_{t=1}^n K_{b_n}(\bm{x}-\bm{X}_t)$ denotes the denominator. In order to establish the consistency of $\widehat{\Sigma}(\bm{x})$, we first consider the oracle estimator that uses the true multivariate conditional mean $\mu(\bm{X}_t)$,
\begin{equation*}
    \widetilde{\Sigma}(\bm{x}) = \frac{\sum_{t=1}^{n}K_{b_n}(\bm{x}-\bm{X}_t)(\bm{Y}_t-\mu(\bm{X}_t))(\bm{Y}_t-\mu(\bm{X}_t))^\top}{\sum_{t=1}^{n}K_{b_n}(\bm{x}-\bm{X}_t)}.
\end{equation*}

Define $Z_t^{(ij)} = (Y_t^{(i)}-\mu^{(i)}(\bm{X}_t))(Y_t^{(j)}-\mu^{(j)}(\bm{X}_t))$ for $1 \leqslant i$, $j \leqslant p$. Then the $(i,j)$th element of $\widetilde{\Sigma}(\bm{x})$ is
\begin{equation*}
    \widetilde{\Sigma}^{(ij)}(\bm{x}) = \frac{\sum_{t=1}^n K_{b_n}(\bm{x}-\bm{X}_t) Z_t^{(ij)}}{\sum_{t=1}^n K_{b_n}(\bm{x}-\bm{X}_t)}.
\end{equation*}

Define the target function $\Sigma^{(ij)}(\bm{u}) 
= \E[Z_0^{(ij)} \mid \bm{X}_0=\bm{u}]$, and denote the corresponding matrix by $\Sigma(\bm{u})$. By the smoothness assumption, $\Sigma^{(ij)}(\cdot) \in \mathcal{C}^2(\bm{x}^{\epsilon})$, and under the stated moment condition we have $\E |Z_0^{(ij)}|^2 < \infty$. Therefore, by the same Nadaraya--Watson consistency argument used for the multivariate conditional mean under the SRD functional dependence condition $\Theta_\infty < \infty$ and the kernel and bandwidth assumptions, we obtain for each fixed $i$ and $j$,
\begin{equation*}
    \widetilde{\Sigma}^{(ij)}(\bm{x}) \inprob (\Sigma^{(ij))}(\bm{x}),
\qquad n \to \infty.
\end{equation*}

Since $p$ is fixed, element-wise convergence implies convergence in matrix norm, and hence
\begin{equation*}
        \widetilde{\Sigma}(\bm{x}) - \Sigma(\bm{x}) \inprob 0.
\end{equation*}

Next, we compare the sample estimator $\widehat{\Sigma}_n(\bm{x})$ with the oracle estimator $\widetilde{\Sigma}_n(\bm{x})$. Define $\Delta_t = \widehat{\mu}_n(\bm{X}_t) - \mu(\bm{X}_t)$. Then, $(\bm{Y}_t-\widehat{\mu}_n(\bm{X}_t))(\bm{Y}_t-\widehat{\mu}_n(\bm{X}_t))^\top
-(\bm{Y}_t-\mu(\bm{X}_t))(\bm{Y}_t-\mu(\bm{X}_t))^\top
 =
-(\bm{Y}_t-\mu(\bm{X}_t))\Delta_t^\top
-\Delta_t(\bm{Y}_t-\mu(\bm{X}_t))^\top
+\Delta_t \Delta_t^\top
$. Let $T_{1t}=(\bm{Y}_t-\mu(\bm{X}_t))\Delta_t^\top$ and $T_{2t}=\Delta_t \Delta_t^\top$. Then,
\[
\widehat{\Sigma}_n(\bm{x}) - \widetilde{\Sigma}_n(\bm{x})
=
\frac{
\sum_{t=1}^n K_{b_n}(\bm{x}-\bm{X}_t)
\left[-T_{1t}-T_{1t}^\top+T_{2t}\right]
}{
D_n(\bm{x})
}.
\]
Taking a submultiplicative matrix norm and using $\|T_{1t}^\top\|=\|T_{1t}\|$, we obtain
\[
\|\widehat{\Sigma}_n(\bm{x}) - \widetilde{\Sigma}_n(\bm{x})\|
\le
\frac{1}{D_n(\bm{x})}
\sum_{t=1}^n K_{b_n}(\bm{x}-\bm{X}_t)
\left[2\|T_{1t}\|+\|T_{2t}\|\right].
\]

Since the kernel is compactly supported, there exists $C_K>0$ such that $K(\bm{v})=0$ whenever $\|\bm{v}\|>C_K$. Hence, $K_{b_n}(\bm{x}-\bm{X}_t)\neq 0$ implies $\|\bm{X}_t-\bm{x}\|\leqslant C_K b_n$. Fix $\epsilon>0$ and take $n$ sufficiently large so that $C_K b_n \leqslant \epsilon$. Then, on the support of the kernel, $\|\Delta_t\|
= \|\widehat{\mu}_n(\bm{X}_t)-\mu(\bm{X}_t)\|
\leqslant \sup_{\|\bm{u}-\bm{x}\|\leqslant \epsilon} \|\widehat{\mu}_n(\bm{u})-\mu(\bm{u})\| =: A_n$, and by the local uniform consistency of the conditional mean estimator, $A_n \inprob 0$. 

For the term $T_{2t}$, $\|T_{2t}\|=\|\Delta_t \Delta_t^\top\|\leqslant \|\Delta_t\|^2$, and therefore on the kernel support $\|T_{2t}\|\leqslant A_n^2$. It follows,
\begin{equation*}
    \frac{1}{D_n(\bm{x})} \sum_{t=1}^n K_{b_n}(\bm{x}-\bm{X}_t)\|T_{2t}\| \leqslant A_n^2\frac{1}{D_n(\bm{x})} \sum_{t=1}^n K_{b_n}(\bm{x}-\bm{X}_t)=A_n^2 \inprob 0.
\end{equation*}

For the term $T_{1t}$, by submultiplicativity of the norm, $\|T_{1t}\|
\le
\|\bm{Y}_t-\mu(\bm{X}_t)\| \|\Delta_t\|
$. On the kernel support, $\|\Delta_t\| \leqslant A_n$, hence
\begin{equation*}
    \frac{1}{D_n(\bm{x})} \sum_{t=1}^n K_{b_n}(\bm{x}-\bm{X}_t)\|T_{1t}\| \leqslant A_n \frac{1}{D_n(\bm{x})} \sum_{t=1}^n K_{b_n}(\bm{x}-\bm{X}_t) \|\bm{Y}_t-\mu(\bm{X}_t)\|.
\end{equation*}

The ratio on the right-hand side is the Nadaraya--Watson estimator of the scalar regression function $m(\bm{x}) = \E\big[\|\bm{Y}_0-\mu(\bm{X}_0)\|\mid \bm{X}_0=\bm{x}\big]$, which is finite under $\E\|\bm{Y}_0\|^2<\infty$. By the same consistency argument as before, this ratio converges in probability to $m(\bm{x})$. Since $A_n \inprob 0$, Slutsky's theorem implies that
\begin{equation*}
    \frac{1}{D_n(\bm{x})}
\sum_{t=1}^n K_{b_n}(\bm{x}-\bm{X}_t)\|T_{1t}\|
\inprob 0.
\end{equation*}

Combining the bounds for $T_{1t}$ and $T_{2t}$, and using $\frac{D_n(\bm{x})}{n b_n^k}
\inprob f_{\bm{X}}(\bm{x})>0$, so that the denominator is bounded away from zero with probability tending to one, we obtain $\|\widehat{\Sigma}_n(\bm{x}) - \widetilde{\Sigma}_n(\bm{x})\|
\inprob 0$. 

Finally, by the triangle inequality, $\|\widehat{\Sigma}_n(\bm{x})-\Sigma(\bm{x})\|
\le
\|\widehat{\Sigma}_n(\bm{x})-\widetilde{\Sigma}_n(\bm{x})\|
+
\|\widetilde{\Sigma}_n(\bm{x})-\Sigma(\bm{x})\|
$. The first term converges to zero in probability by the above argument, and the second term converges to zero in probability by the oracle consistency. Hence, $\widehat{\Sigma}_n(\bm{x}) - \Sigma(\bm{x}) \inprob 0$, which completes the proof.
\end{proof}

\begin{proof}[Proof of \Cref{thm:theorem3}]

We show that the iteratively reweighted least squares algorithm converges to the unique minimizer of the empirical quantile objective. The argument relies on constructing a quadratic surrogate function that majorizes the objective at each iteration and whose minimizer can be written in closed form. By repeatedly minimizing these surrogates, we generate a sequence of iterates whose objective values decrease monotonically and whose limit equals the true minimizer due to strict convexity.

Recall that for any $\bm{q} \in \R^p$, the empirical objective function is defined as
\begin{equation}
    M^{(p)}_{\bm{u},n}(\bm{q}) = \frac{1}{n} \sum_{t=1}^n \Big( \| \bm{Y}_t - \bm{q} \|_{K(\cdot)} + \langle \bm{u},\, \bm{Y}_t - \bm{q} \rangle_{K(\cdot)} \Big),
    \label{eq:lemma31_obj}
\end{equation}
where $\|v\|_{K(\cdot)} = \| K(\cdot)v \|$ and $\langle a, b \rangle_{K(\cdot)} = a^\intercal K(\cdot)b$. Since $K(\cdot)$ is scalar and nonnegative, the objective in \eqref{eq:lemma31_obj} is a positively weighted geometric quantile criterion. Under the standard nondegeneracy condition that the set $\{\bm{Y}_t:K(\cdot)>0\}$ is not contained in a single affine line (in particular, not all kernel weighted observations are collinear), the function $M_{\bm{u},n}^{(p)}(\bm{q})$ is strictly convex in $\bm{q}$. The linear term in $\bm{u}$ does not affect convexity. Hence the global minimizer exists and is unique.

We now describe the surrogate construction. Let $\bm{q}^{(k)}$ denote the current iterate. For any positive numbers $a$ and $b$, $(a-b)^2 \geqslant 0$, so, $a^2 + b^2 \geqslant 2ab$, which implies, $a^2/2b + b/2 \geqslant a$. The equality holds if and only if $a = b$. Applying this inequality to $a = \| \bm{Y}_t - \bm{q} \|_{K(\cdot)}$ and $b = \| \bm{Y}_t - \bm{q}^{(k)} \|_{K(\cdot)}$, we obtain the upper bound
\begin{equation}
    \|\bm{Y}_t - \bm{q}\|_{K(\cdot)} \leqslant \frac{1}{2} w_t(\bm{q}^{(k)}) \|\bm{Y}_t-\bm{q}\|_{K(\cdot)}^2 + \frac{1}{2} w_t(\bm{q}^{(k)})^{-1},
    \label{eq:lemma31_single_maj}
\end{equation}
where the weights are defined as $ w_t(\bm{q}^{(k)})=\frac{1}{\|\bm{Y}_t-\bm{q}^{(k)}\|_{K(\cdot)}}. $ Summing \eqref{eq:lemma31_single_maj} over $t$ and adding the linear term from \eqref{eq:lemma31_obj}, we obtain the surrogate
\begin{equation}
    \widetilde M^{(p)}_{\bm{u},n}(\bm{q} \mid \bm{q}^{(k)}) = \frac{1}{n} \sum_{t=1}^n \left[ \frac{1}{2} w_t(\bm{q}^{(k)}) \|\bm{Y}_t-\bm{q}\|_{K(\cdot)}^2 + \langle \bm{u}, \bm{Y}_t-\bm{q} \rangle_{K(\cdot)} \right] + C^{(k)},
    \label{eq:lemma31_surrogate}
\end{equation}
where $C^{(k)}$ collects all terms independent of $\bm{q}$. By construction, the surrogate majorizes the objective:
\begin{equation*}
    M^{(p)}_{\bm{u},n}(\bm{q}) \leqslant \widetilde M^{(p)}_{\bm{u},n}(\bm{q} \mid \bm{q}^{(k)}) \quad \forall \; \bm{q},
\end{equation*}
and equality holds at $\bm{q}^{(k)}$, that is, $\widetilde M^{(p)}_{\bm{u},n}(\bm{q}^{(k)} \mid \bm{q}^{(k)}) = M^{(p)}_{\bm{u},n}(\bm{q}^{(k)})$. The surrogate function in \eqref{eq:lemma31_surrogate} is a strictly convex quadratic function in $\bm{q}$, hence its minimizer can be computed explicitly. Expanding the squared term, $\|\bm{Y}_t-\bm{q}\|_{K(\cdot)}^2 = (\bm{Y}_t-\bm{q})^\intercal K(\cdot)^2 (\bm{Y}_t-\bm{q})$, and differentiating the surrogate with respect to $\bm{q}$, the first order optimality condition yields 
\begin{equation*}
    \sum_{t=1}^n w_t(\bm{q}^{(k)})K(\cdot)^2 (\bm{Y}_t-\bm{q}) + \frac{1}{2} \sum_{t=1}^n K(\cdot)\bm{u} = 0.
\end{equation*}
Solving for $\bm{q}$ gives the update expression
\begin{equation}
    \bm{q}^{(k+1)} = 
    \left[ \sum_{t=1}^n w_t(\bm{q}^{(k)}) K(\cdot)^2 \right]^{-1}
    \left(
        \frac{1}{2}\sum_{t=1}^n K(\cdot)\bm{u} 
        + \sum_{t=1}^n w_t(\bm{q}^{(k)}) K(\cdot)^2 \bm{Y}_t
    \right),
    \label{eq:lemma31_update}
\end{equation}
which is exactly the update rule stated in the lemma.

Since $\bm{q}^{(k+1)}$ minimizes the surrogate, we have $    \widetilde M^{(p)}_{\bm{u},n}(\bm{q}^{(k+1)} \mid \bm{q}^{(k)}) \leqslant \widetilde M^{(p)}_{\bm{u},n}(\bm{q}^{(k)} \mid \bm{q}^{(k)})$. Using the majorization relation, $    M^{(p)}_{\bm{u},n}(\bm{q}^{(k+1)}) \leqslant \widetilde M^{(p)}_{\bm{u},n}(\bm{q}^{(k+1)} \mid \bm{q}^{(k)}) \leqslant \widetilde{M}_{\bm{u},n}^{(p)}(\bm{q}^{k} \mid \bm{q}^k) = M^{(p)}_{\bm{u},n}(\bm{q}^{(k)})$, which shows that the sequence of objective values is non increasing. Since the objective is a weighted sum of norms and linear terms with non-negative kernel weights, it is bounded from below. Hence the sequence $\{M_{\bm{u},n}^{(p)}(\bm{q}^{(k)})\}_{k \geqslant 0}$ converges to a finite limit. Moreover, the iterates remain in the level set determined by the initial value. Indeed, $M_{\bm{u},n}^{(p)}(\bm{q}^{(k)}) \leqslant M_{\bm{u},n}^{(p)}(\bm{q}^{(0)}) \, \forall \, k$, so each $\bm{q}^{(k)}$ belongs to $\{ \bm{q} \in \mathbb{R}^p : M_{\bm{u},n}^{(p)}(\bm{q}) \leqslant M_{\bm{u},n}^{(p)}(\bm{q}^{(0)})\}$. 

We now show that this level set is bounded. Using the inequality $\langle \bm{u}, \bm{v} \rangle \geqslant - \|\bm{u}\|\,\|\bm{v}\|$, together with $\|\bm{u}\| < 1$, we obtain $\|\bm{Y}_t - \bm{q}\| + \langle \bm{u}, \bm{Y}_t - \bm{q} \rangle \geqslant (1 - \|\bm{u}\|)\|\bm{Y}_t - \bm{q}\|$. Multiplying by the non-negative kernel weights and summing over $t$, it follows that $M_{\bm{u},n}^{(p)}(\bm{q}) \geqslant c_1 \|\bm{q}\| - c_2$, for some constants $c_1 > 0$ and $c_2 < \infty$ independent of $\bm{q}$. Hence the objective grows at least linearly as $\|\bm{q}\| \to \infty$, and therefore every level set is bounded. In particular, the sequence $\{\bm{q}^{(k)}\}$ is bounded.

Since the sequence is bounded, it admits at least one accumulation point. Let $\bar{\bm{q}}$ be such a point and consider a subsequence $\bm{q}^{(k_j)} \to \bar{\bm{q}}$. By the assumption of the theorem,
$\min_{1\leqslant t\leqslant n}\|\bm{Y}_t-\bar{\bm{q}}\|_{K(\cdot)}\geqslant \delta$.
Since $\bm{q}\mapsto \|\bm{Y}_t-\bm{q}\|_{K(\cdot)}$ is continuous, there exists a neighborhood
$\mathcal{N}(\bar{\bm{q}})$ such that $\|\bm{Y}_t-\bm{q}\|_{K(\cdot)}\geqslant \delta/2$ for all
$\bm{q}\in\mathcal{N}(\bar{\bm{q}})$ and all $t=1,\dots,n$.
Consequently the weights $w_t(\bm{q})=1/\|\bm{Y}_t-\bm{q}\|_{K(\cdot)}$ are continuous on
$\mathcal{N}(\bar{\bm{q}})$, and therefore the update map $T(\bm{q})$ defined by \eqref{eq:lemma31_update}
is continuous at $\bar{\bm{q}}$. Since $\bm{q}^{(k_j+1)}=T(\bm{q}^{(k_j)})$ and $T$ is continuous at $\bar{\bm{q}}$, we have $\bm{q}^{(k_j+1)}\to T(\bar{\bm{q}})$.
But $\{\bm{q}^{(k_j+1)}\}$ is also a subsequence of the iterates, hence it converges to the same accumulation point $\bar{\bm{q}}$.
Therefore $T(\bar{\bm{q}})=\bar{\bm{q}}$, so $\bar{\bm{q}}$ is a fixed point of the surrogate minimization step. By construction, any such fixed point satisfies the stationary condition $\nabla_{\bm{q}} M^{(p)}_{\bm{u},n}(\bar{\bm{q}})=\bm{0}$. Thus every accumulation point is a stationary point of the objective. Because $M_{\bm{u},n}^{(p)}$ is strictly convex, it admits a unique stationary point, which is its global minimizer. Hence the entire sequence $\{\bm{q}^{(k)}\}$ converges to this unique minimizer.

We have established the algorithmic stability and we now are going to check for the analytical stability by checking a score function which similar to the pinball loss function in its application.

We fix the number of response variables $p \in \N$. All constants may depend on $p$ but not on $n$. For a given covariate value $\bm{x} \in \R^k$ and direction $\bm{u} \in B^{p-1}$, recall that the sample objective function is
\begin{equation*}
  M^{(p)}_{\bm{u},n}(q)
  = \frac{1}{n} \sum_{t=1}^n \Phi_{K}\bigl(\bm{u}, \bm{Y}_t - \bm{q}\bigr),
\end{equation*}
where, for each $t \in \Gamma$ and $\bm{q} \in \R^p$,
\begin{equation*}
  \Phi_{K}\bigl(\bm{u}, \bm{Y}_t - \bm{q}\bigr)
  := \bigl\|\bm{Y}_t - \bm{q}\bigr\|_{K(\cdot)}
  + \bigl\langle \bm{u}, \bm{Y}_t - \bm{q} \bigr\rangle_{K(\cdot)}.
\end{equation*}

As defined earlier, $\bigl\|\bm{Y}_t - \bm{q}\bigr\|_{K(\cdot)} = \bigl\|K(\cdot)(\bm{Y}_t - \bm{q})\bigr\|$. Since, $K(\cdot) \geqslant 0$, we can write, $\bigl\|\bm{Y}_t - \bm{q}\bigr\|_{K(\cdot)} = \big\|K(\cdot)\big\|\bigl\|(\bm{Y}_t - \bm{q})\bigr\|$. The euclidean norm of $(\bm{Y}_t - \bm{q})$ is strictly convex with given $t$ and $p$ in $\bm{q}$. The inner product is linear so the sum of convex and linear results in a convex function.

The sample criterion $M^{(p)}_{\bm{u},n}$ is therefore a finite sum of strictly convex functions and is itself strictly convex in $\bm{q}$. We restrict the optimization to the open convex set $U_p \subset \R^p$ and that the true minimizer does not coincide with any of the observed points $\bm{Y}_t$.

Let $\widehat{\bm{q}}_n(\bm{u},\bm{x})$ be the unique minimizer of $M^{(p)}_{\bm{u},n}(\bm{q})$ over $U_p$. By strict convexity, $\widehat{\bm{q}}_n$ is characterized by the vanishing of the directional derivative in every direction $\bm{v} \in \R^p$. More precisely, for every $\bm{v} \in \R^p$
\begin{equation}
  \label{eq:directional-ineq}
  \lim_{\varepsilon \to 0}
  \frac{1}{n \varepsilon} \sum_{i=1}^n
  \Bigl[
    \Phi_{K}\bigl(\bm{u}, \bm{Y}_t - \widehat{\bm{q}}_n - \varepsilon \bm{v}\bigr)
    - \Phi_{K}\bigl(\bm{u}, \bm{Y}_t - \widehat{\bm{q}}_n\bigr)
  \Bigr]
  \geqslant 0.
\end{equation}
We now compute this directional derivative explicitly. Using the definition of $\Phi_{K}$, we can rewrite the increment as
\begin{equation*}
  \Phi_{K}\bigl(\bm{u}, \bm{Y}_t - \widehat{\bm{q}}_n - \varepsilon \bm{v}\bigr)
  - \Phi_{K}\bigl(\bm{u}, \bm{Y}_t - \widehat{\bm{q}}_n\bigr)
  =
  \bigl\|\bm{Y}_t - \widehat{\bm{q}}_n - \varepsilon \bm{v}\bigr\|_{K(\cdot)}
  - \bigl\|\bm{Y}_t - \widehat{\bm{q}}_n\bigr\|_{K(\cdot)}
  - \varepsilon \bigl\langle \bm{u}, \bm{v} \bigr\rangle_{K(\cdot)}.
\end{equation*}
Using the standard derivative of the norm, we obtain, for each $t$ and any $\bm{v} \in \R^p$,
\begin{equation*}
  \lim_{\varepsilon \to 0}
  \frac{\bigl\|\bm{Y}_t - \widehat{\bm{q}}_n - \varepsilon \bm{v}\bigr\|_{K(\cdot)}
        - \bigl\|\bm{Y}_t - \widehat{\bm{q}}_n\bigr\|_{K(\cdot)}
       }{\varepsilon}
  =
  - \frac{\bigl\langle \bm{Y}_t - \widehat{\bm{q}}_n, \bm{v} \bigr\rangle_{K(\cdot)}}
         {\bigl\|\bm{Y}_t - \widehat{\bm{q}}_n\bigr\|_{K(\cdot)}}.
\end{equation*}
Therefore, taking limits in \eqref{eq:directional-ineq}, using linearity of the sum and multiplying $-1$ both sides, we obtain
\begin{equation*}
  \frac{1}{n} \sum_{t=1}^n
  \left[
    \frac{\bigl\langle \bm{Y}_t - \widehat{\bm{q}}_n, \bm{v} \bigr\rangle_{K(\cdot)}}
            {\bigl\|\bm{Y}_t - \widehat{\bm{q}}_n\bigr\|_{K(\cdot)}}
    + \bigl\langle \bm{u}, \bm{v} \bigr\rangle_{K(\cdot)}
  \right]
  \leqslant 0,
  \qquad \forall \; \bm{v} \in \R^p.
\end{equation*}
Applying the same argument with direction $-\bm{v}$ instead of $\bm{v}$ yields the opposite inequality, and therefore
\begin{equation*}
  \frac{1}{n} \sum_{t=1}^n
  \left[
    \frac{\bigl\langle \bm{Y}_t - \widehat{\bm{q}}_n, \bm{v} \bigr\rangle_{K(\cdot)}}
            {\bigl\|\bm{Y}_t - \widehat{\bm{q}}_n\bigr\|_{K(\cdot)}}
    + \bigl\langle \bm{u}, \bm{v} \bigr\rangle_{K(\cdot)}
  \right]
  = 0,
  \qquad \forall \; \bm{v} \in \R^p.
\end{equation*}
Since the inner product $\langle \cdot,\cdot \rangle_{K(\cdot)}$ is nondegenerate, we obtain the sample estimating equation
\begin{equation*}
  \frac{1}{n} \sum_{i=1}^n
  K(\cdot)
  \left[
    \frac{K(\cdot)\bigl(\bm{Y}_t - \widehat{\bm{q}}_n(\bm{u},\bm{x})\bigr)}
         {\bigl\|K(\cdot)\bigl(\bm{Y}_t - \widehat{\bm{q}}_n(\bm{u},\bm{x})\bigl\|}
    + \bm{u}
  \right]
  = 0.
\end{equation*}

We now derive the corresponding population version. Recall that, for fixed $\bm{u},\bm{x}$ and $p$, the theoretical conditional $\bm{u}$-quantile $\bm{Q}(\bm{u},\bm{x})$ is defined as the unique minimizer of
\begin{equation*}
  M^{(p)}_{\bm{u}}(\bm{q})
  := \E\!\left[
    \bigl\|\bm{Y}_t - \bm{q}\bigr\|
    + \bigl\langle \bm{u}, \bm{Y}_t - \bm{q} \bigr\rangle
    \,\bigg|\, \bm{X}_t = \bm{x}
  \right],
\end{equation*}
where the expectation is taken with respect to the conditional distribution of $\bm{Y}_t$ given $\bm{X}_t=\bm{x}$. The same convexity arguments as above apply at the population level: By the strict convexity of the norm and linearity of the inner product, the function, $M^{(p)}_{u}(\bm{q})$, is strictly convex, and $\bm{Q}(u,x)$ is characterized by the vanishing of its directional derivatives. For any direction $\bm{v} \in \R^p$, we have
\begin{equation*}
  \lim_{\varepsilon \to 0}
  \frac{M^{(p)}_{\bm{u}}\bigl(\bm{Q}(\bm{u},\bm{x}) + \varepsilon \bm{v}\bigr)
        - M^{(p)}_{\bm{u}}\bigl(\bm{Q}(\bm{u},\bm{x})\bigr)
       }{\varepsilon}
  = \E\!\left[
    \
      \frac{\langle \bm{Y}_t - \bm{Q}(\bm{u},\bm{x}), \bm{v} \rangle}{\bigl\|\bm{Y}_t - \bm{Q}(\bm{u},\bm{x})\bigr\|}
      + \langle \bm{u},
      \bm{v} \rangle
    \bigg| \bm{X}(t) = \bm{x}
  \right]
  = 0.
\end{equation*}
This holds for all $\bm{v} \in \R^p$ and inner product is nondegenerate so we obtain the population estimating equation
\begin{equation*}
  \E\!\left[
    \frac{\bm{Y}_t - \bm{Q}(u,x)}{\bigl\|\bm{Y}_t - \bm{Q}(u,x)\bigr\|}
    + u
    \,\bigg|\, X(t) = x
  \right]
  = 0.
\end{equation*}

We have shown that, for fixed $p$, the minimizer $\widehat{\bm{q}}_n(u,x)$ and $\bm{Q}(u,\bm{x})$ satisfy the desired sample and population estimating equations. This completes the proof.
\end{proof}

\begin{proof}[Proof of \Cref{thm:theorem5}]

We establish the consistency and rate of convergence of the nonparametric multivariate conditional quantile estimator when the response dimension $p$ is fixed. For notational ease we fix $(\bm{u},\bm{x})$ and write
\begin{equation*}
    \widehat{\bm{q}}_n \equiv \widehat{\bm{q}}_n(\bm{u},\bm{x}), 
    \qquad 
    \bm{q}_0 \equiv \bm{Q}(\bm{u},\bm{x}).
\end{equation*}
The empirical and population objective functions are
\begin{equation*}
    M^{(p)}_{u,n}(\bm{q}) 
      = \frac{1}{nb_n^k}\sum_{t=1}^n 
      \{
        \|\bm{Y}_t-\bm{q}\|_{K(\cdot)}
        + \langle \bm{u},\bm{Y}_t-\bm{q} \rangle_{K(\cdot)}
      \};
     \quad
    M^{(p)}_{u}(\bm{q})
      = \E\!\left[
         \|\bm{Y}_t-\bm{q}\|
         + \langle \bm{u},\bm{Y}_t-\bm{q}\rangle
         \,|\,
         \bm{X}_t=\bm{x}
      \right],
\end{equation*}
with $\|\bm{v}\|_{K(\cdot)}=\|K(\cdot)\bm{v}\|$ and $\langle \bm{a},\bm{b}\rangle_{K(\cdot)}=\bm{a}^\intercal K(\cdot)\bm{b}$.

Under the previously stated nondegeneracy condition that the kernel weighted observations are not contained in a single affine subspace, the population objective $M^{(p)}_u$ is strictly convex on an open convex set $U_p \subset \R^p$ and admits a unique minimizer $\bm{q}_0$.

For fixed $\bm{q}\in U_p$, by the law of large numbers,
\begin{equation*}
    M^{(p)}_{u,n}(\bm{q})
    = \frac{1}{n}\sum_{t=1}^n \varphi(\bm{Y}_t,\bm{X}_t;\bm{q})
    \longrightarrow
    M^{(p)}_u(\bm{q})
    \qquad\text{a.s.},
    \label{eq:pointwiseLLN}
\end{equation*}
where
\begin{equation*}
    \varphi(\bm{Y}_t,\bm{X}_t;\bm{q})
    := \big\|\bm{Y}_t-\bm{q}\big\|_{K(\cdot)}
       + \langle \bm{u},\bm{Y}_t-\bm{q}\rangle_{K(\cdot)}.
\end{equation*}
On any compact convex $\mathcal{Q}_p\subset U_p$ containing $\bm{q}_0$,
\begin{equation}
    \sup_{\bm{q}\in \mathcal{Q}_p}
    \big|M^{(p)}_{u,n}(\bm{q}) - M^{(p)}_u(\bm{q})\big|
    \inprob 0,
    \label{eq:ULLN}
\end{equation}
because the map $\bm{q}\mapsto \varphi(\cdot;\bm{q})$ is Lipschitz on $\mathcal{Q}_p$ with integrable envelope. Since $M^{(p)}_u$ is strictly convex, there exists $\eta(\epsilon)>0$ such that
\begin{equation}
    \inf_{\|\bm{q}-\bm{q}_0\|\geqslant \epsilon}
    \big(M^{(p)}_u(\bm{q}) - M^{(p)}_u(\bm{q}_0)\big)
    \geqslant \eta(\epsilon).
    \label{eq:well-separated}
\end{equation}
Combining \eqref{eq:ULLN} and \eqref{eq:well-separated} yields
\begin{equation}
    \widehat{\bm{q}}_n \inprob \bm{q}_0.
    \label{eq:consistencyQhat}
\end{equation}

The minimizer satisfies from the previous result,
\begin{equation*}
    \frac{1}{n}\sum_{t=1}^n \bm{\psi}_{t}(\widehat{\bm{q}}_n)=\bm{0},
    \label{eq:estimatingEq}
\end{equation*}
where
\begin{equation*}
    \bm{\psi}_{t}(\bm{q})
    := K(\cdot)\!\left[
      \frac{K(\cdot)\{\bm{Y}_t-\bm{q}\}}
           {\|K(\cdot)\{\bm{Y}_t-\bm{q}\}\|}
      + \bm{u}
    \right].
\end{equation*}
Moreover $\E[\bm{\psi}_{t}(\bm{q}_0)]=\bm{0}$.  
Let $\dot{\bm{\psi}}_{t}(\bm{q})$ denote the Jacobian. A Taylor expansion around $\bm{q}_0$ gives
\begin{equation}
    \bm{\psi}_{t}(\widehat{\bm{q}}_n)
    = \bm{\psi}_{t}(\bm{q}_0)
      + \dot{\bm{\psi}}_{t}(\bm{q}_0)(\widehat{\bm{q}}_n-\bm{q}_0)
      + \bm{r}_{t,n},
    \label{eq:taylorPsi}
\end{equation}
where
\begin{equation}
    \|\bm{r}_{t,n}\|
    \leqslant C \|\widehat{\bm{q}}_n - \bm{q}_0\|^2.
    \label{eq:remainderBound}
\end{equation}
Averaging \eqref{eq:taylorPsi} yields
\begin{equation}
    \bm{0}
    = \frac{1}{n}\sum_{t=1}^n \bm{\psi}_{t}(\bm{q}_0)
      + \left\{
        \frac{1}{n}\sum_{t=1}^n \dot{\bm{\psi}}_{t}(\bm{q}_0)
      \right\}(\widehat{\bm{q}}_n - \bm{q}_0)
      + \frac{1}{n}\sum_{t=1}^n \bm{r}_{t,n}.
    \label{eq:avgExpansion}
\end{equation}
Define the population Jacobian
\begin{equation*}
    \bm{\Psi}(\bm{q}_0)
    := \E[\dot{\bm{\psi}}_{t}(\bm{q}_0)].
\end{equation*}
By strict convexity, $\bm{\Psi}(\bm{q}_0)$ is positive definite. A law of large numbers implies
\begin{equation*}
    \frac{1}{n}\sum_{t=1}^n \dot{\bm{\psi}}_{t}(\bm{q}_0)
    \inprob \bm{\Psi}(\bm{q}_0),
    \label{eq:jacobianLLN}
\end{equation*}
so the inverse exists with probability approaching one.  
Since $\bm{\psi}_{t}(\bm{q}_0)$ has finite second moments and under small range dependence,
\begin{equation}
    \left\|
      \frac{1}{n}\sum_{t=1}^n \bm{\psi}_{t}(\bm{q}_0)
    \right\|
    = \mathcal{O}p(\sqrt{p/nb_n^k}).
    \label{eq:psiOp}
\end{equation}
By \eqref{eq:remainderBound} and \eqref{eq:consistencyQhat},
\begin{equation}
    \left\|
      \frac{1}{n}\sum_{t=1}^n \bm{r}_{t,n}
    \right\|
    = o_{p}(\sqrt{p/nb_n^k}).
    \label{eq:remainderAvg}
\end{equation}

Pre-multiplying \eqref{eq:avgExpansion} by the inverse of $n^{-1}\sum_{t=1}^n \dot{\bm{\psi}}_{t}(\bm{q}_0)$ yields
\begin{equation*}
    \|\widehat{\bm{q}}_n - \bm{q}_0\|
    \leqslant 
    C_1
    \left\|
      \frac{1}{n}\sum_{t=1}^n \bm{\psi}_{t}(\bm{q}_0)
    \right\|
    +
    C_2
    \left\|
      \frac{1}{n}\sum_{t=1}^n \bm{r}_{t,n}
    \right\|.
\end{equation*}
By \eqref{eq:psiOp} and \eqref{eq:remainderAvg},
\begin{equation*}
    \|\widehat{\bm{q}}_n - \bm{q}_0\|
    = \mathcal{O}p(\sqrt{p/nb_n^k}).
\end{equation*}

This completes the proof.
\end{proof}


\setcounter{figure}{0}
\renewcommand\thefigure{B.\arabic{figure}}

\newpage

\section{Additional plots from the real data analysis}
\label{sec:Plots}

\begin{figure}[h]
\centering
  \begin{subfigure}{.33\textwidth}
    \includegraphics[width=.9\linewidth]{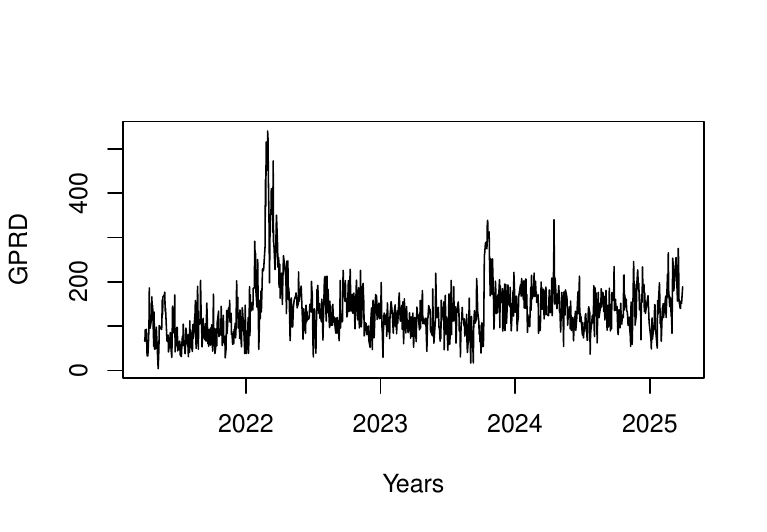}
  \end{subfigure}%
  \begin{subfigure}{.33\textwidth}
    \includegraphics[width=.9\linewidth]{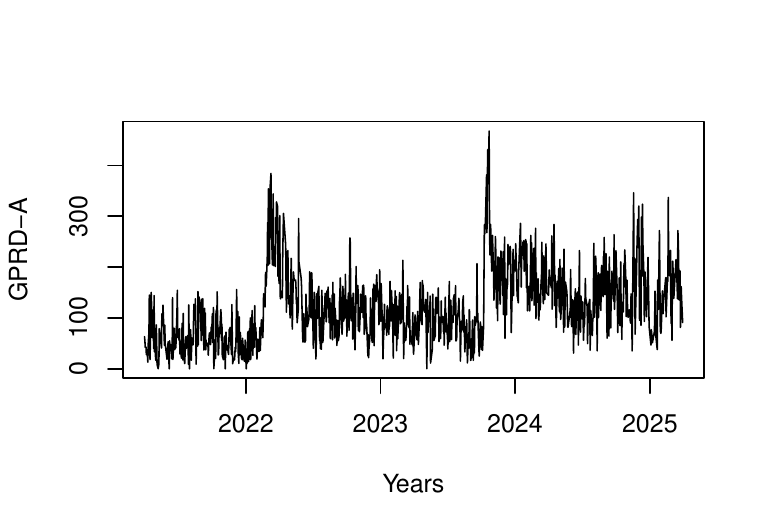}
  \end{subfigure}%
   \begin{subfigure}{.33\textwidth}
    \includegraphics[width=.9\linewidth]{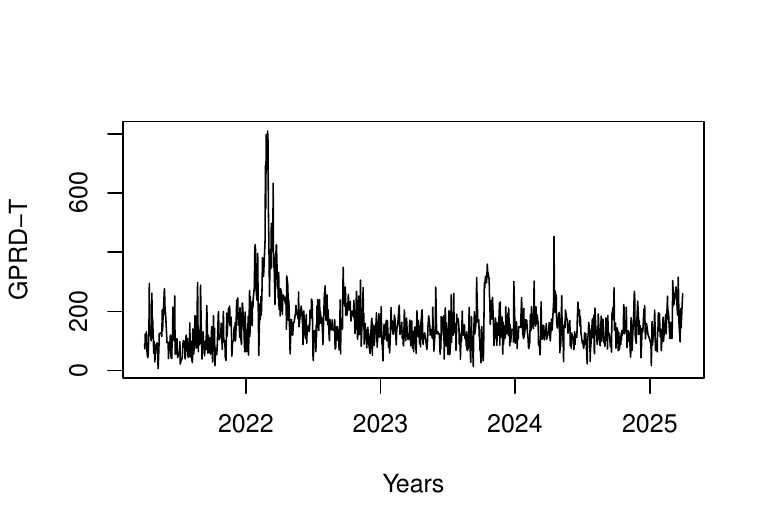}
  \end{subfigure}
\caption{Geopolitical Risk Indices (GPRD, GPRD-A, GPRD-T) against Time}
\label{fig:GPRD_GPRDA_GPRDT_vs_Times}
\end{figure}

\begin{figure}[!ht]
    \centering
    \begin{subfigure}{0.4\textwidth}
        \centering
        \includegraphics[width=\linewidth]{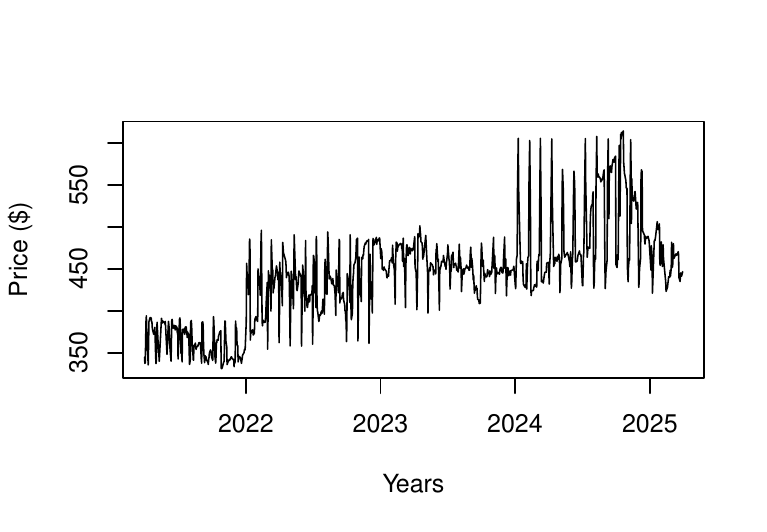}
        \label{fig:sp_lr_vol_LHMnMMA a}
    \end{subfigure}
    \hfill
    \begin{subfigure}{0.4\textwidth}
        \centering
        \includegraphics[width=\linewidth]{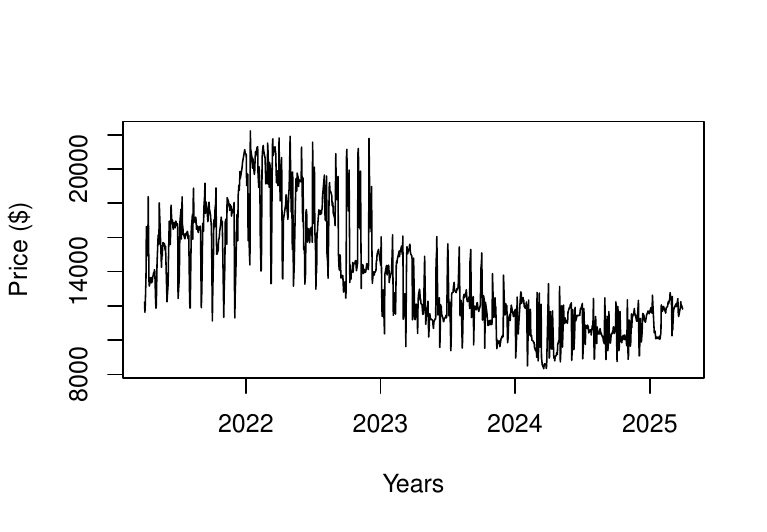}
        \label{fig:sp_lr_vol_LHMnMMA b}
    \end{subfigure}
    \begin{subfigure}{0.4\textwidth}
        \centering
        \includegraphics[width=\linewidth]{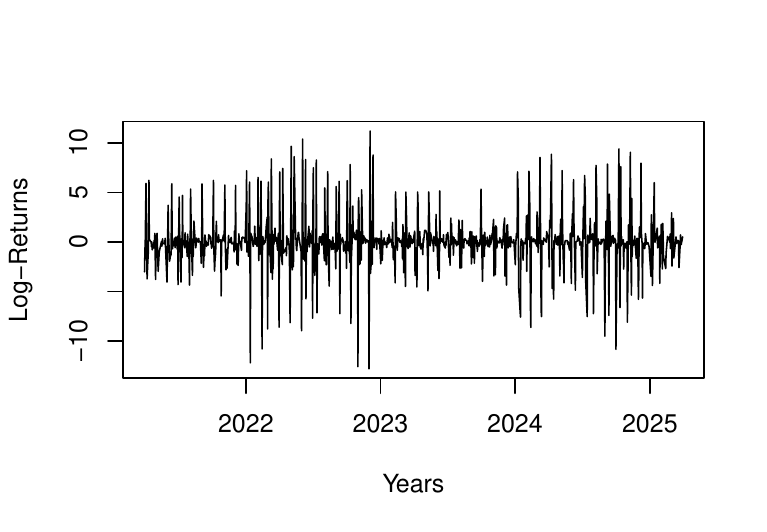}
        \label{fig:sp_lr_vol_LHMnMMA c}
    \end{subfigure}
    \hfill
    \begin{subfigure}{0.4\textwidth}
        \centering
        \includegraphics[width=\linewidth]{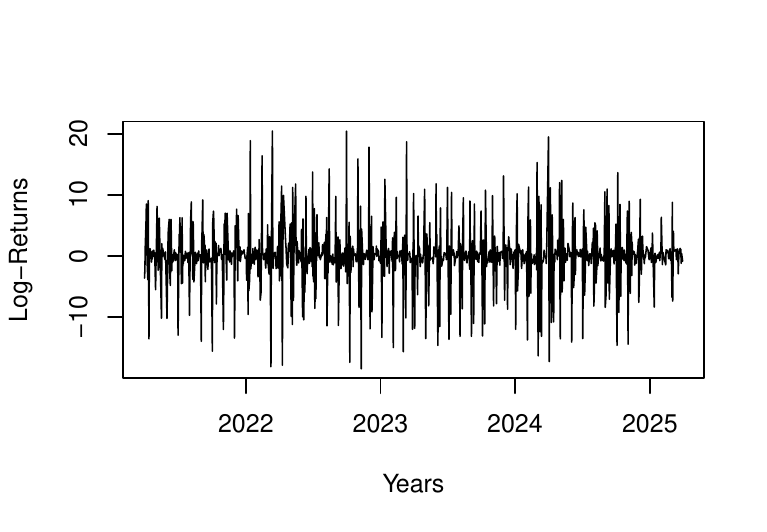}
        \label{fig:sp_lr_vol_LHMnMMA d}
    \end{subfigure}
    \begin{subfigure}{0.41\textwidth}
        \centering
        \includegraphics[width=\linewidth]{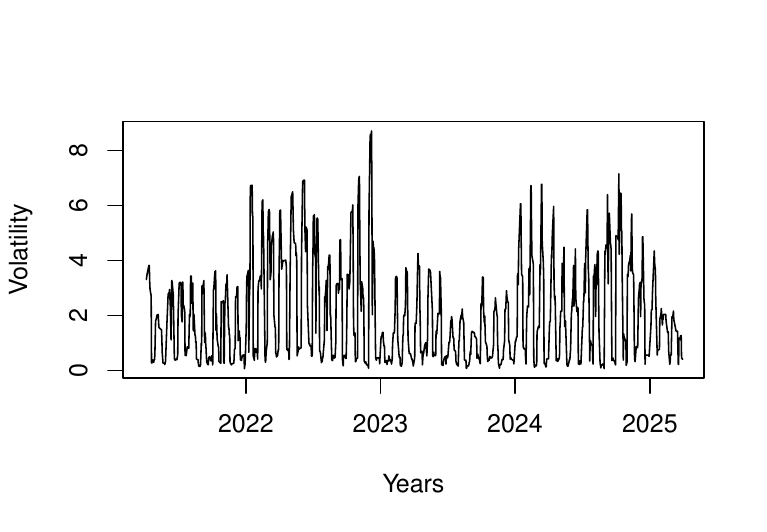}
        \label{fig:sp_lr_vol_LHMnMMA e}
    \end{subfigure}
    \hfill
    \begin{subfigure}{0.41\textwidth}
        \centering
        \includegraphics[width=\linewidth]{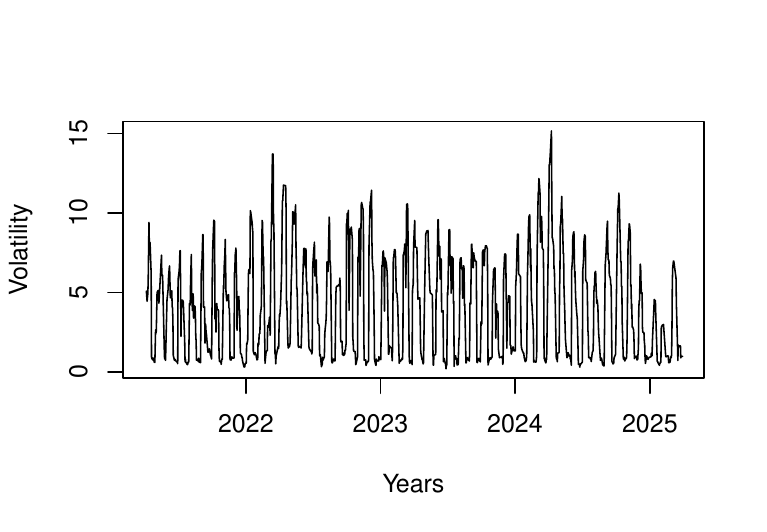}
        \label{fig:sp_lr_vol_LHMnMMA f}
    \end{subfigure}
    \caption{Stock prices, log-returns and volatility of LHM (left) and MMA (right) against time}
\label{fig:sp_lr_vol_LHMnMMA}
\end{figure}

\begin{figure}[!]
    \centering
    \begin{subfigure}[t]{0.32\textwidth}
        \centering
        \includegraphics[width=\textwidth]{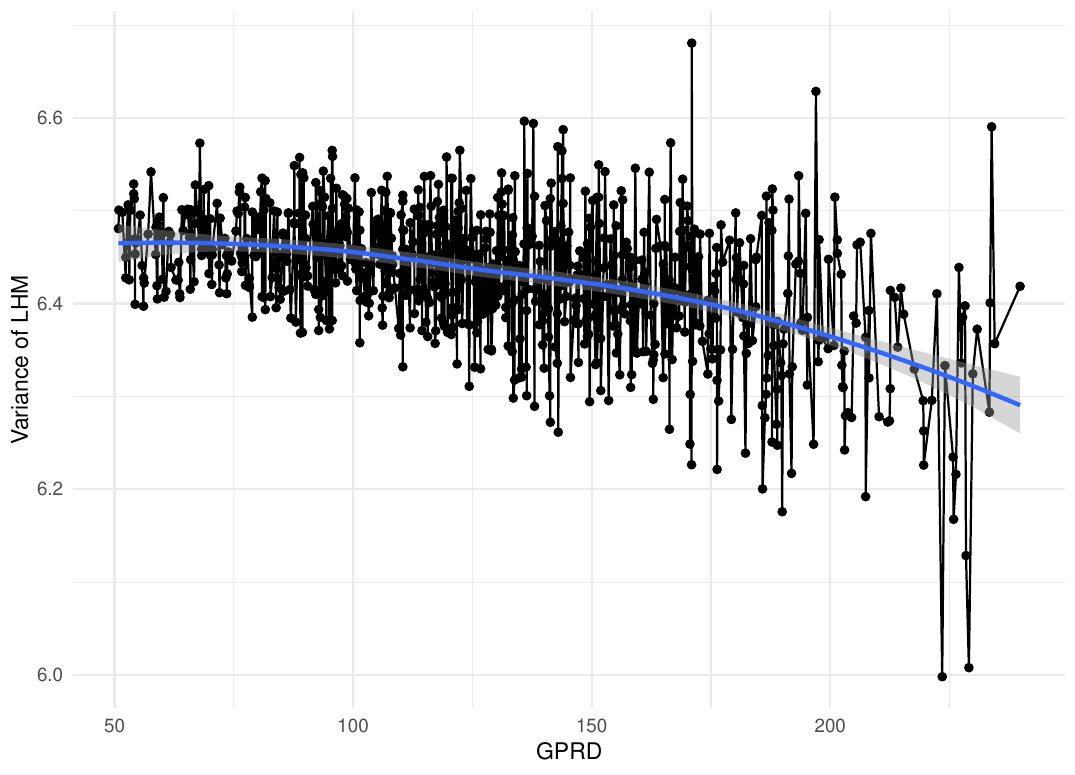}
        \caption{Var(LHM) - GPRD}
        \label{fig:Var_Cov_GPR_Time a}
        \vspace{0.3cm}
        
        \includegraphics[width=\textwidth]{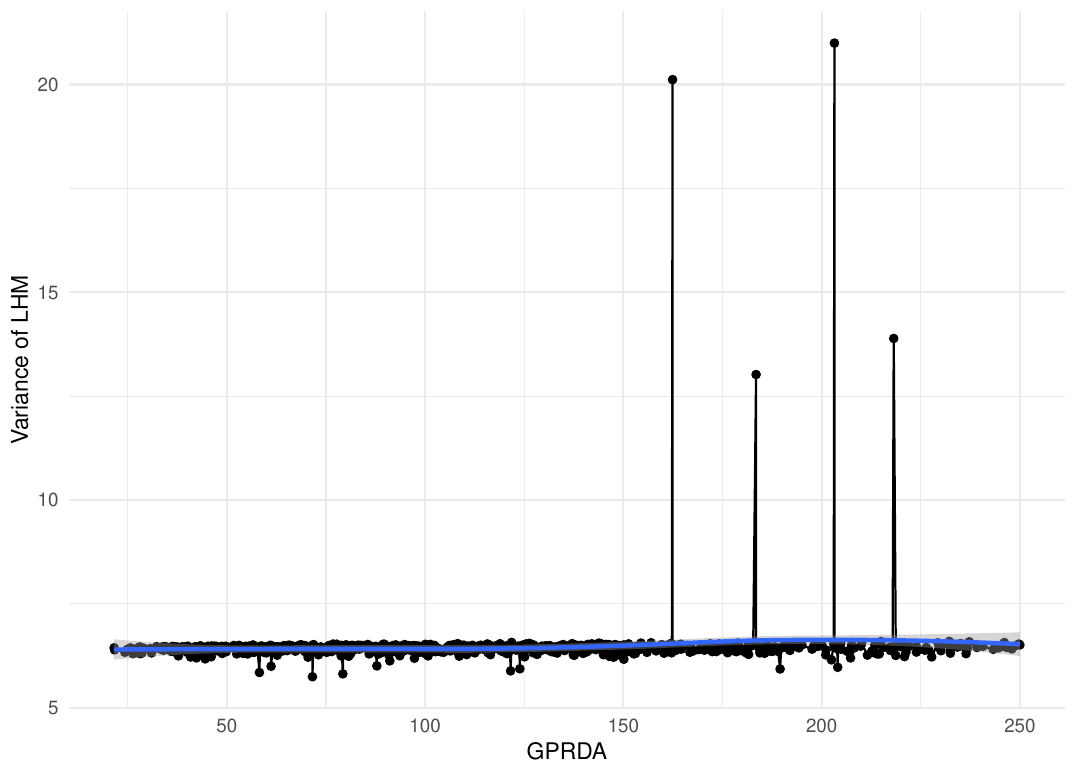}
        \caption{Var(LHM) - GPRDA}
        \label{fig:Var_Cov_GPR_Time b}
        \vspace{0.3cm}
        
        \includegraphics[width=\textwidth]{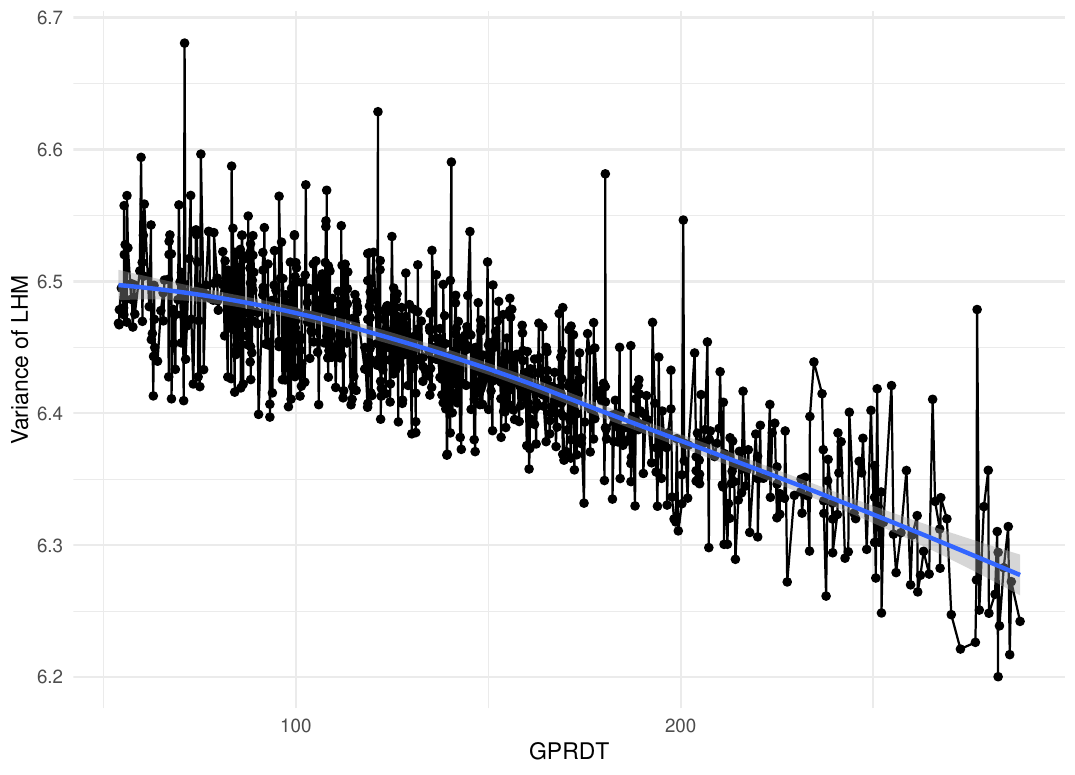}
        \caption{Var(LHM) - GPRDT}
        \label{fig:Var_Cov_GPR_Time c}
        \vspace{0.3cm}
        
        \includegraphics[width=\textwidth]{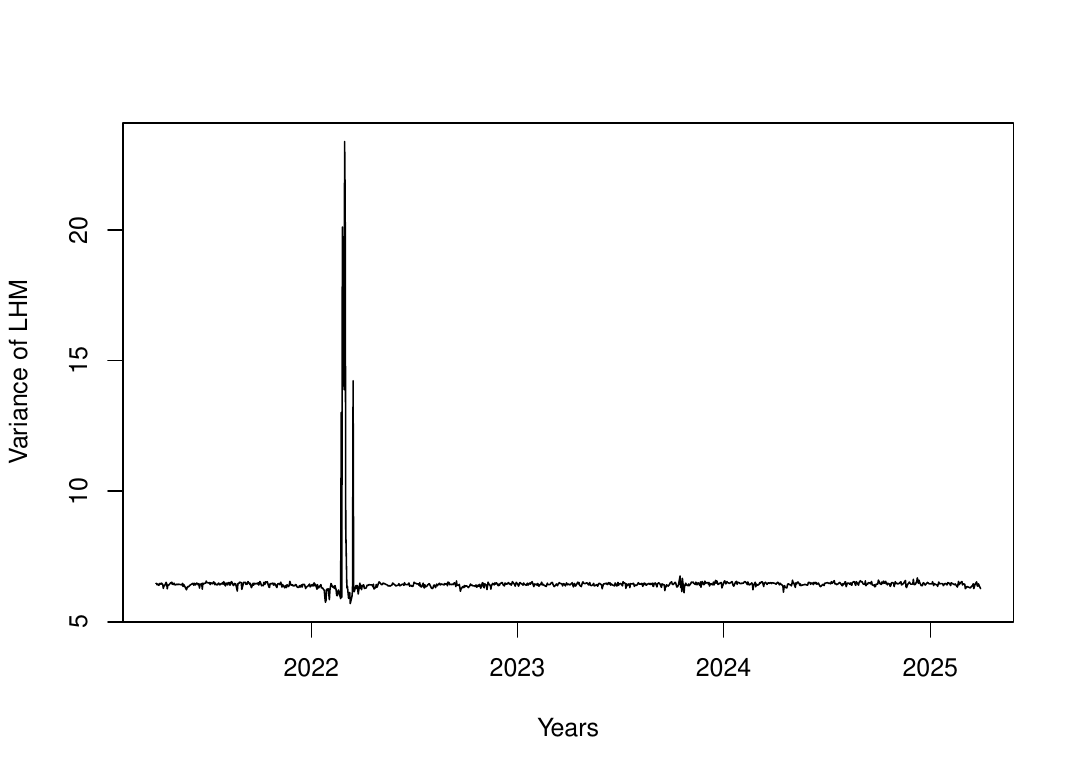}
        \caption{Var(LHM) - Time}
        \label{fig:Var_Cov_GPR_Time d}
    \end{subfigure}
    \hfill
    \begin{subfigure}[t]{0.32\textwidth}
        \centering
        \includegraphics[width=\textwidth]{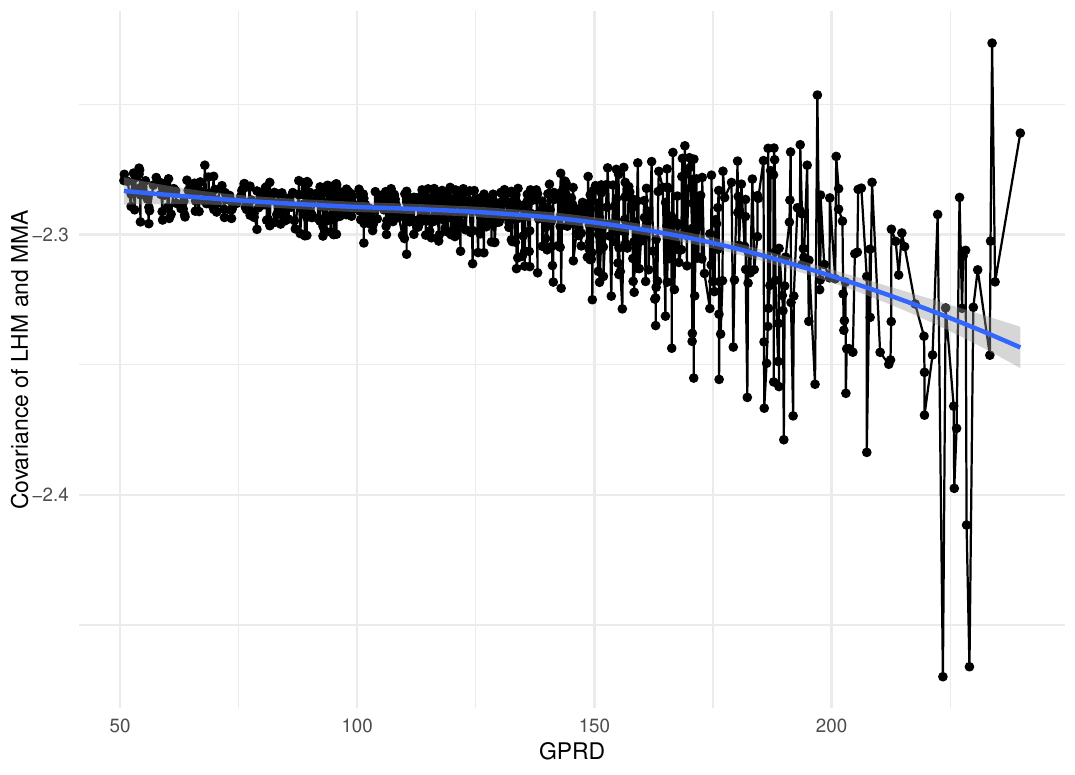}
        \caption{Cov(LHM, MMA) - GPRD}
        \label{fig:Var_Cov_GPR_Time e}
        \vspace{0.3cm}
        
        \includegraphics[width=\textwidth]{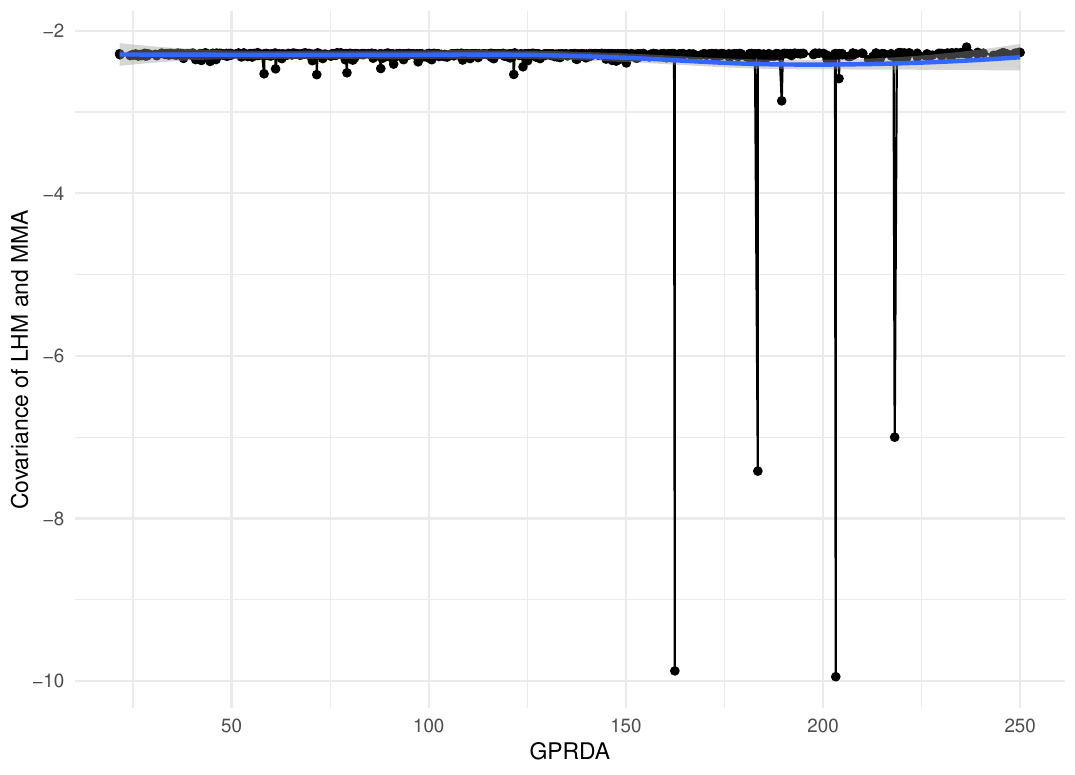}
        \caption{Cov(LHM, MMA) - GPRDA}
        \label{fig:Var_Cov_GPR_Time f}
        \vspace{0.3cm}
        
        \includegraphics[width=\textwidth]{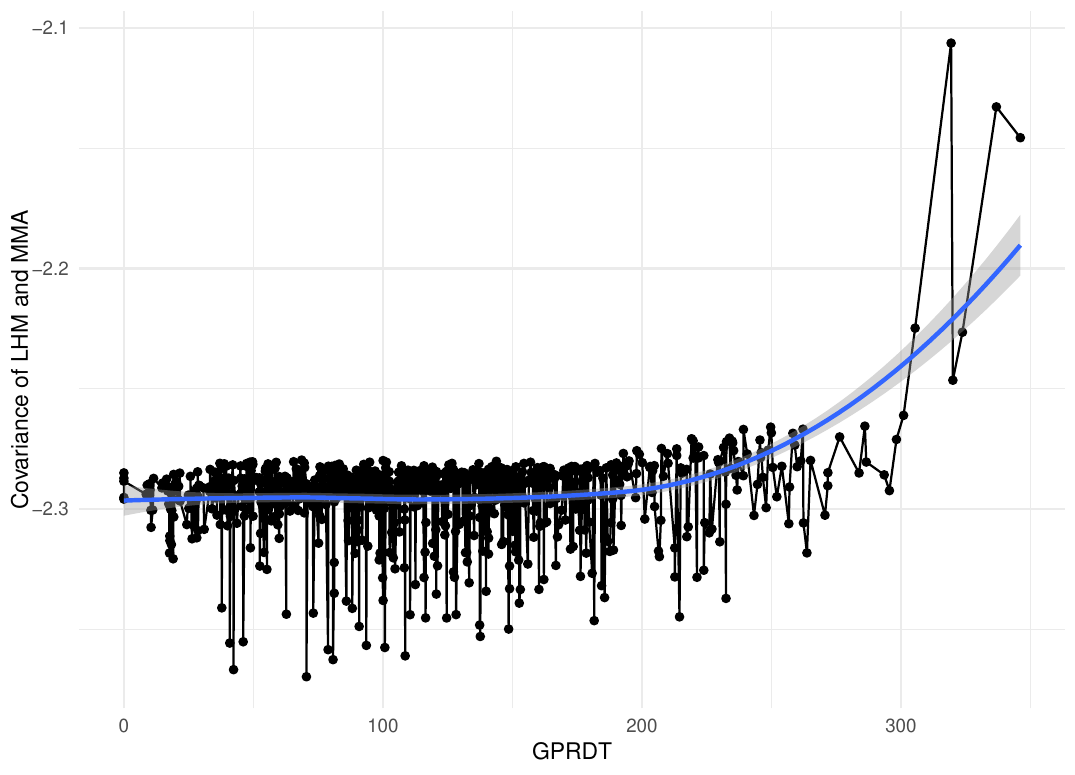}
        \caption{Cov(LHM, MMA) - GPRDT}
        \label{fig:Var_Cov_GPR_Time g} 
        \vspace{0.3cm}
        
        \includegraphics[width=\textwidth]{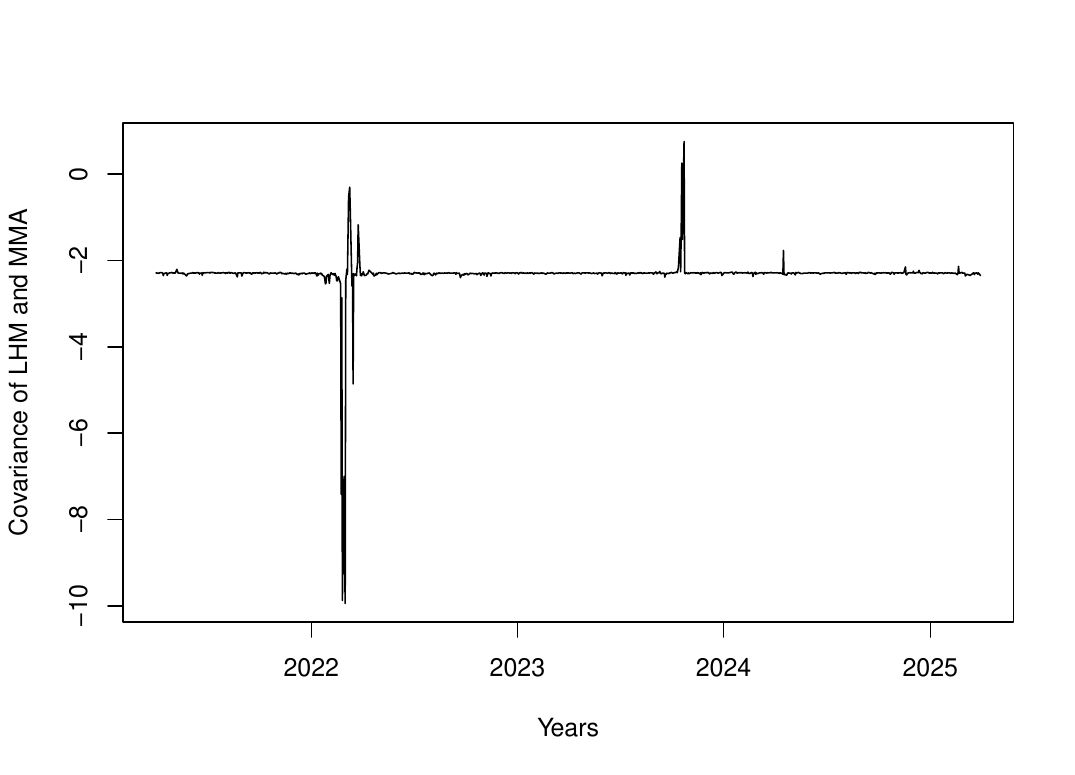}
        \caption{Cov(LHM, MMA) - Time}
        \label{fig:Var_Cov_GPR_Time h}
    \end{subfigure}
    \begin{subfigure}[t]{0.32\textwidth}
        \centering
        \includegraphics[width=\textwidth]{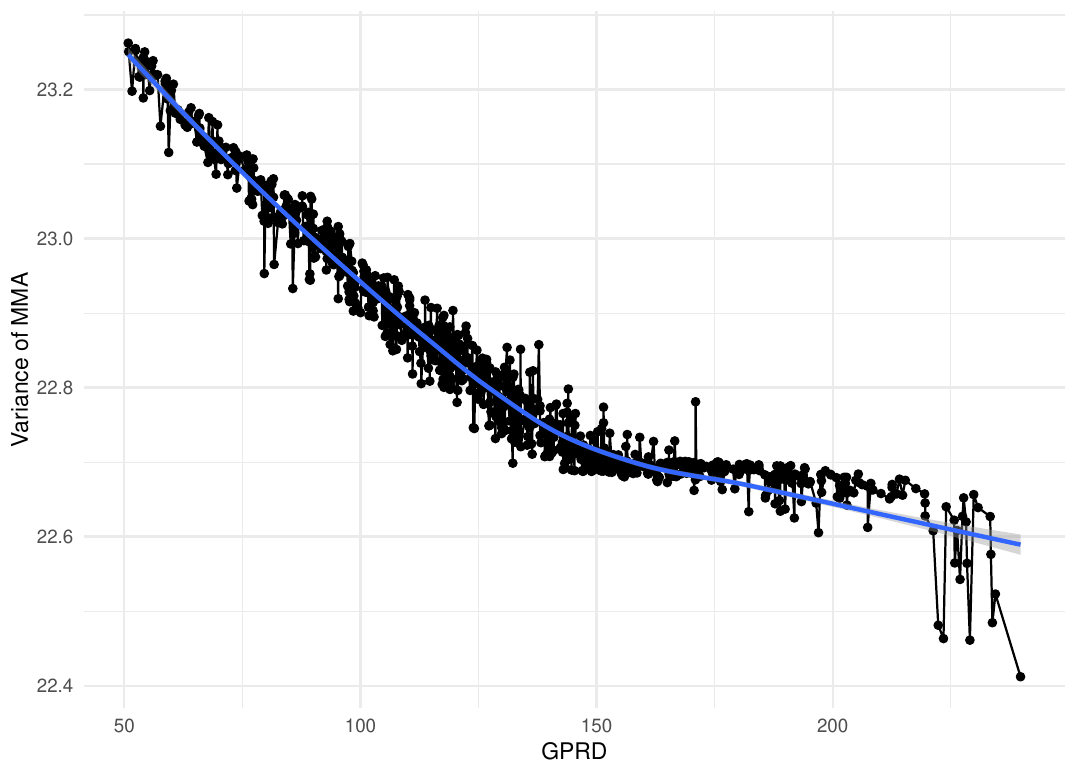}
        \caption{Var(MMA) - GPRD}
        \label{fig:Var_Cov_GPR_Time i}
        \vspace{0.3cm}
        
        \includegraphics[width=\textwidth]{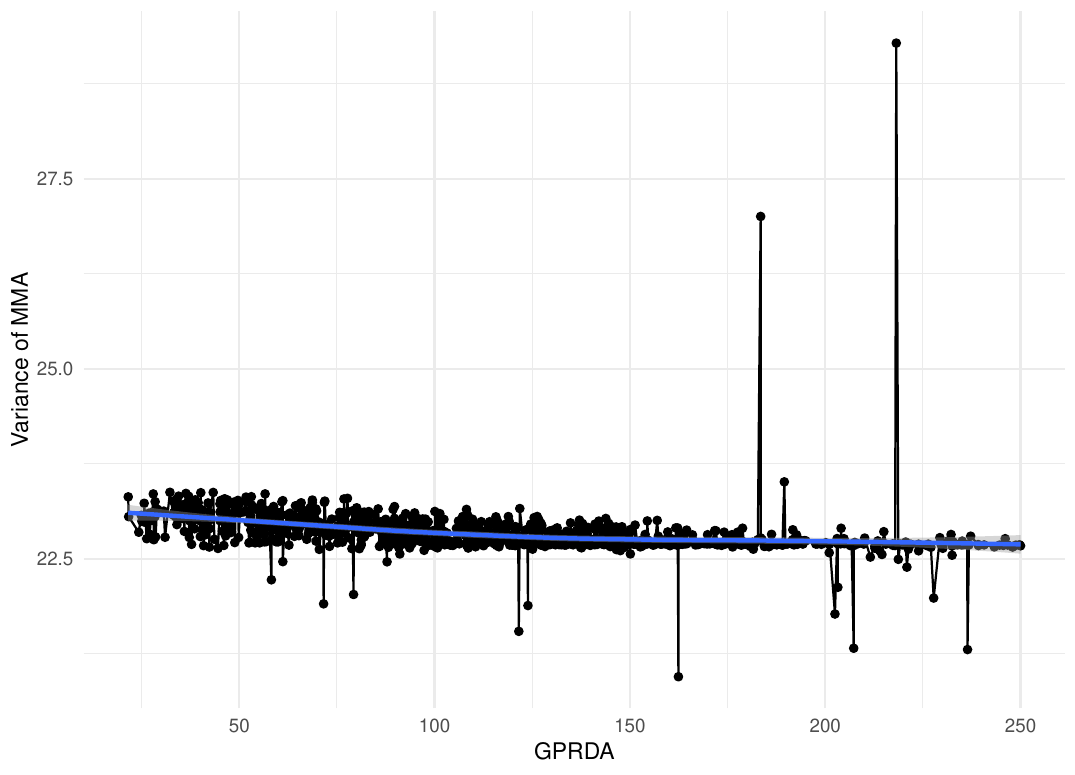}
        \caption{Var(MMA) - GPRDA}
        \label{fig:Var_Cov_GPR_Time j}
        \vspace{0.3cm}
        
        \includegraphics[width=\textwidth]{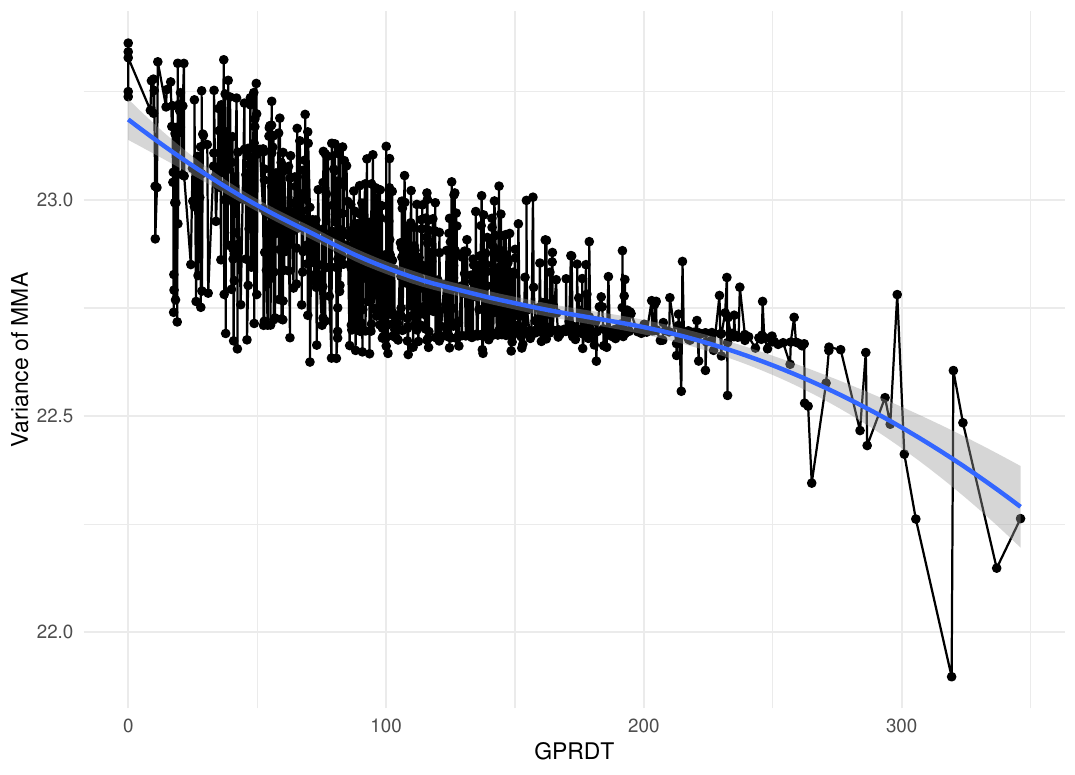}
        \caption{Var(MMA) - GPRDT}
        \label{fig:Var_Cov_GPR_Time k}
        \vspace{0.3cm}
        
        \includegraphics[width=\textwidth]{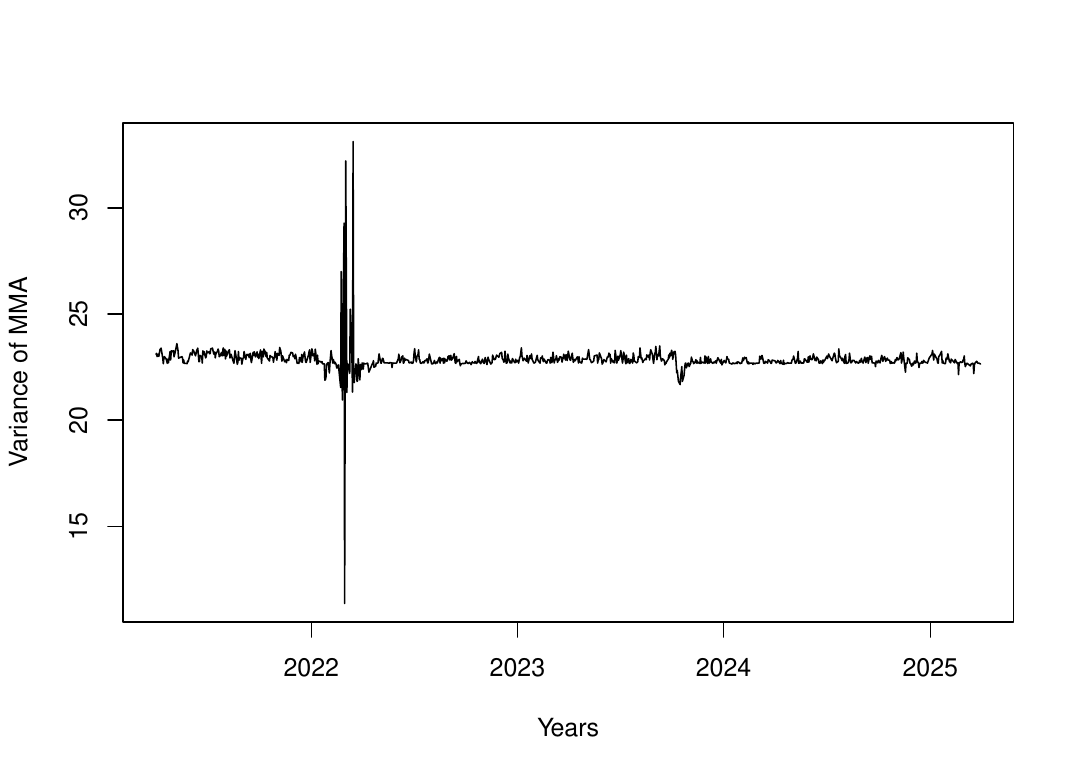}
        \caption{Var(MMA) - Time}
        \label{fig:Var_Cov_GPR_Time l}
    \end{subfigure}
    \hfill
    \caption{Components of the estimated conditional covariance matrix against the three geopolitical risk indices and time}
    \label{fig:Var_Cov_GPR_Time}
\end{figure}

\end{document}